\documentclass[a4paper,UKenglish,cleveref, autoref]{lipics-v2019}

\title{Twin-width I: tractable FO model checking}
\titlerunning{Twin-width I: tractable FO model checking}

\author{\'{E}douard Bonnet}{Univ Lyon, CNRS, ENS de Lyon, Université Claude Bernard Lyon 1, LIP UMR5668, France}{edouard.bonnet@ens-lyon.fr}{https://orcid.org/0000-0002-1653-5822}{}
\author{Eun Jung Kim}{Universit\'{e} Paris-Dauphine, PSL University, CNRS UMR7243, LAMSADE, Paris, France}{eun-jung.kim@dauphine.fr}{https://orcid.org/0000-0002-6824-0516}{}
\author{St\'{e}phan Thomass\'{e}}{Univ Lyon, CNRS, ENS de Lyon, Universit\'{e} Claude Bernard Lyon 1, LIP UMR5668, France}{stephan.thomasse@ens-lyon.fr}{}{}
\author{R\'{e}mi Watrigant}{Univ Lyon, CNRS, ENS de Lyon, Universit\'{e} Claude Bernard Lyon 1, LIP UMR5668, France}{remi.watrigant@univ-lyon1.fr}{https://orcid.org/0000-0002-6243-5910}{}

\authorrunning{\'E. Bonnet, E. J. Kim, S. Thomassé, R. Watrigant}

\Copyright{Édouard Bonnet, Eun Jung Kim, Stéphan Thomassé, Rémi Watrigant}



\category{}

\relatedversion{}

\supplement{}

\acknowledgements{}

\nolinenumbers 

\hideLIPIcs  

\EventEditors{John Q. Open and Joan R. Access}
\EventNoEds{2}
\EventLongTitle{42nd Conference on Very Important Topics (CVIT 2016)}
\EventShortTitle{CVIT 2016}
\EventAcronym{CVIT}
\EventYear{2016}
\EventDate{December 24--27, 2016}
\EventLocation{Little Whinging, United Kingdom}
\EventLogo{}
\SeriesVolume{42}
\ArticleNo{23}

\usepackage[utf8]{inputenc}  

\usepackage[T1]{fontenc}
\usepackage{lmodern}

\usepackage[colorinlistoftodos,bordercolor=orange,backgroundcolor=orange!20,linecolor=orange,textsize=normalsize]{todonotes}

\usepackage{amsmath}  
\usepackage{amssymb}     
\usepackage{bbm}

\usepackage{booktabs}
\usepackage{bm}
\usepackage{hyperref}
\usepackage{mathrsfs}
\usepackage{mathtools} 
\usepackage{paralist}
\usepackage{fixmath}
\usepackage{interval}

\crefformat{equation}{\textup{#2(#1)#3}}
\crefrangeformat{equation}{\textup{#3(#1)#4--#5(#2)#6}}
\crefmultiformat{equation}{\textup{#2(#1)#3}}{ and \textup{#2(#1)#3}}
{, \textup{#2(#1)#3}}{, and \textup{#2(#1)#3}}
\crefrangemultiformat{equation}{\textup{#3(#1)#4--#5(#2)#6}}%
{ and \textup{#3(#1)#4--#5(#2)#6}}{, \textup{#3(#1)#4--#5(#2)#6}}{, and \textup{#3(#1)#4--#5(#2)#6}}

\Crefformat{equation}{#2Equation~\textup{(#1)}#3}
\Crefrangeformat{equation}{Equations~\textup{#3(#1)#4--#5(#2)#6}}
\Crefmultiformat{equation}{Equations~\textup{#2(#1)#3}}{ and \textup{#2(#1)#3}}
{, \textup{#2(#1)#3}}{, and \textup{#2(#1)#3}}
\Crefrangemultiformat{equation}{Equations~\textup{#3(#1)#4--#5(#2)#6}}%
{ and \textup{#3(#1)#4--#5(#2)#6}}{, \textup{#3(#1)#4--#5(#2)#6}}{, and \textup{#3(#1)#4--#5(#2)#6}}

\usepackage{xspace}
\usepackage{tikz}

\usetikzlibrary{fit}
\usetikzlibrary{arrows}
\usetikzlibrary{patterns}
\usetikzlibrary{calc}
\usetikzlibrary{shapes}
\usetikzlibrary{positioning}
\usetikzlibrary{math}
\usetikzlibrary{shapes.symbols}

\usepackage[scr=boondox,scrscaled=1.05]{mathalfa}

\renewcommand{\geq}{\geqslant}
\renewcommand{\leq}{\leqslant}
\renewcommand{\preceq}{\preccurlyeq}

\newcommand{\defparproblem}[4]{
 \vspace{1mm}
\noindent\fbox{
 \begin{minipage}{0.96\textwidth}
 \begin{tabular*}{\textwidth}{@{\extracolsep{\fill}}lr} #1 & {\bf{Parameter:}} #3 \\ \end{tabular*}
 {\bf{Input:}} #2 \\
 {\bf{Question:}} #4
 \end{minipage}
 }
 \vspace{1mm}
}

\DeclareMathOperator*{\bigland}{\bigwedge}

\newtheorem{conjecture}[theorem]{Conjecture}

\theoremstyle{definition}

\begin{document}

\maketitle

\begin{abstract}
  Inspired by a \emph{width} invariant defined on permutations by Guillemot and Marx [SODA '14], we introduce the notion of twin-width on graphs and on matrices.
  Proper minor-closed classes, bounded rank-width graphs, map graphs, $K_t$-free unit $d$-dimensional ball graphs, posets with antichains of bounded size, and proper subclasses of dimension-2 posets all have bounded twin-width.
  On all these classes (except map graphs without geometric embedding) we show how to compute in polynomial time a \emph{sequence of $d$-contractions}, witness that the twin-width is at most $d$.
  We show that FO model checking, that is deciding if a given first-order formula $\phi$ evaluates to true for a given binary structure $G$ on a domain $D$, is FPT in $|\phi|$ on classes of bounded twin-width, provided the witness is given.
  More precisely, being given a $d$-contraction sequence for $G$, our algorithm runs in time $f(d,|\phi|) \cdot |D|$ where $f$ is a computable but non-elementary function.
  We also prove that bounded twin-width is preserved under FO interpretations and transductions (allowing operations such as squaring or complementing a graph).
  This unifies and significantly extends the knowledge on fixed-parameter tractability of FO model checking on non-monotone classes, such as the FPT algorithm on bounded-width posets by Gajarsk\'y et al.~[FOCS '15].
\end{abstract}



\section{Introduction}\label{sec:intro}

Measuring how complex a class of structures is often depends on the context.
Complexity can be related to algorithms (are computations easier on the class?), counting (how many structures exist per slice of the class?), size (can structures be encoded in a compact way?), decomposition (can structures be built with easy operations?), and so on.
The most successful and central complexity invariants like treewidth and VC-dimension tick many of these boxes and, as such, stand as cornerstone notions in both discrete mathematics and computer science.

In 2014, Guillemot and Marx \cite{Guillemot14} solved a long-standing question by showing that detecting a fixed pattern in some input permutation can be done in linear time.
This result came as a surprise: Many researchers thought the problem was W[1]-hard since all known techniques had failed so far.
In their paper, Guillemot and Marx observed that their proof introduces a parameter and a dynamic programming scheme of a new kind and wondered whether a graph-theoretic generalization of their permutation parameter could exist.

The starting point of our paper is to answer that question positively, by generalizing their width parameter to graphs and even matrices.
This new notion, dubbed \emph{twin-width}, proves remarkably well connected to other areas of computer science, logic, and combinatorics.
We will show that graphs of bounded twin-width define a very natural class with respect to computational complexity (FO model checking is linear), to model theory (they are stable under first-order interpretations), to enumerative combinatorics (they form small classes \cite{twin-width2}), and to decomposition methods (as a generalization of both proper minor-closed and bounded rank-width/clique-width classes).

\subsection{A dynamic generalization of cographs}

When it comes to graph decompositions, arguably one of the simplest graph classes is the class of \emph{cographs}.
Starting from a single vertex, cographs can be built by iterating disjoint unions and complete sums.
Another way to decompose cographs is to observe that they always contain \emph{twins}, that is two vertices $u$ and $v$ with the same neighborhood outside $\{u,v\}$ (hence contracting $u,v$ is equivalent to deleting $u$).
Cographs are then exactly graphs which can be contracted to a single vertex by iterating contractions of twins.
Generalizing the decomposition by allowing more complex bipartitions provides the celebrated notions of clique-width and rank-width, which extends treewidth to dense graphs.
However, bounded rank-width do not capture simple graphs such as unit interval graphs which have a simple linear structure, and allow polynomial-time algorithms for various problems.
Also, bounded rank-width does not capture large 2-dimensional grids, on which we know how to design FPT algorithms.

The goal of this paper is to propose a width parameter which is not only bounded on $d$-dimensional grids, proper minor-closed classes and bounded rank-width graphs, but also provides a very versatile and simple scheme which can be applied to many structures, for instance, patterns of permutations, hypergraphs, and posets.
The idea is very simple: a graph has bounded twin-width if it can be iteratively contracted to a singleton, where each contracted pair consists of near-twins (two vertices whose neighborhoods differ only on a bounded number of elements).
The crucial ingredient to add to this simplified picture is to keep track of the errors with another type of edges, that we call \emph{red edges}, and to require that the degree in red edges remains bounded by a threshold, say $d$.

In a nutshell (a more formal definition will be given in Section~\ref{sec:def}), we consider a sequence of graphs $G_n, G_{n-1}, \ldots, G_2, G_1$, where $G_n$ is the original graph $G$, $G_1$ is the one-vertex graph, $G_i$ has $i$ vertices, and $G_{i-1}$ is obtained from $G_i$ by performing a single contraction of two (non-necessarily adjacent) vertices.
For every vertex $u \in V(G_i)$, let us denote by $u(G)$ the vertices of $G$ which have been contracted to $u$ along the sequence $G_n, \ldots, G_i$.
A pair of disjoint sets of vertices is \emph{homogeneous} if, between these sets, there are either all possible edges or no edge at all.
The red edges mentioned previously consist of all pairs $uv$ of vertices of $G_i$ such that $u(G)$ and $v(G)$ are not homogeneous in $G$.
If the red degree of every $G_i$ is at most $d$, then $G_n, G_{n-1}, \ldots, G_2, G_1$ is called a \emph{sequence of $d$-contractions}, or \emph{$d$-sequence}.
The twin-width of $G$ is the minimum $d$ for which there exists a sequence of $d$-contractions.
Hence, graphs of twin-width $0$ are exactly the cographs (since a red edge never appears along the sequence when contracting twins).
See \cref{fig:twin-contraction} for an illustration of a 2-sequence.
\begin{figure}
  \centering
  \begin{tikzpicture}[
      vertex/.style={circle, draw, minimum size=0.85cm}
    ]
    \def\s{1.2}
    \foreach \i/\j/\l in {0/0/a,0/1/b,0/2/c,1/0/d,1/1/e,1/2/f,2/1/g}{
      \node[vertex] (\l) at (\i * \s,\j * \s) {$\l$} ;
    }
    \foreach \i/\j in {a/b,a/d,a/f,b/c,b/d,b/e,b/f,c/e,c/f,d/e,d/g,e/g,f/g}{
      \draw (\i) -- (\j) ;
    }

    \begin{scope}[xshift=3 * \s cm]
    \foreach \i/\j/\l in {0/0/a,0/1/b,0/2/c,1/0/d,2/1/g}{
      \node[vertex] (\l) at (\i * \s,\j * \s) {$\l$} ;
    }
    \foreach \i/\j/\l in {1/1/e,1/2/f}{
      \node[vertex,opacity=0.2] (\l) at (\i * \s,\j * \s) {$\l$} ;
    }
    \node[draw,rounded corners,inner sep=0.01cm,fit=(e) (f)] (ef) {ef} ;
    \foreach \i/\j in {a/b,a/d,b/c,b/d,b/ef,c/ef,c/ef,d/g,ef/g,ef/g}{
      \draw (\i) -- (\j) ;
    }
    \foreach \i/\j in {a/ef,d/ef}{
      \draw[red, very thick] (\i) -- (\j) ;
    }
    \end{scope}

    \begin{scope}[xshift=6 * \s cm]
    \foreach \i/\j/\l in {0/1/b,0/2/c,2/1/g,1/1/ef}{
      \node[vertex] (\l) at (\i * \s,\j * \s) {$\l$} ;
    }
    \foreach \i/\j/\l in {0/0/a,1/0/d}{
      \node[vertex,opacity=0.2] (\l) at (\i * \s,\j * \s) {$\l$} ;
    }
    \draw[opacity=0.2] (a) -- (d) ;
    \node[draw,rounded corners,inner sep=0.01cm,fit=(a) (d)] (ad) {ad} ;
    \foreach \i/\j in {ad/b,b/c,b/ad,b/ef,c/ef,c/ef,ef/g,ef/g}{
      \draw (\i) -- (\j) ;
    }
    \foreach \i/\j in {ad/ef,ad/g}{
      \draw[red, very thick] (\i) -- (\j) ;
    }
    \end{scope}

    \begin{scope}[yshift=-3 * \s cm]
    \foreach \i/\j/\l in {0/2/c,2/1/g,0.5/0/ad}{
      \node[vertex] (\l) at (\i * \s,\j * \s) {$\l$} ;
    }
    \foreach \i/\j/\l in {0/1/b,1/1/ef}{
      \node[vertex,opacity=0.2] (\l) at (\i * \s,\j * \s) {$\l$} ;
    }
    \draw[opacity=0.2] (b) -- (ef) ;
    \node[draw,rounded corners,inner sep=0.01cm,fit=(b) (ef)] (bef) {bef} ;
    \foreach \i/\j in {ad/bef,bef/c,bef/ad,c/bef,c/bef,bef/g}{
      \draw (\i) -- (\j) ;
    }
    \foreach \i/\j in {ad/bef,ad/g,bef/g}{
      \draw[red, very thick] (\i) -- (\j) ;
    }
    \end{scope}

    \begin{scope}[yshift=-3 * \s cm, xshift=3 * \s cm]
    \foreach \i/\j/\l in {0/2/c}{
      \node[vertex] (\l) at (\i * \s,\j * \s) {$\l$} ;
    }
     \foreach \i/\j/\l in {0.5/0/adg,0.5/1.1/bef}{
      \node[vertex] (\l) at (\i * \s,\j * \s) {\footnotesize{\l}} ;
    }
    \foreach \i/\j in {c/bef}{
      \draw (\i) -- (\j) ;
    }
    \foreach \i/\j in {adg/bef}{
      \draw[red, very thick] (\i) -- (\j) ;
    }
    \end{scope}

    \begin{scope}[yshift=-3 * \s cm, xshift=5 * \s cm]
    \foreach \i/\j/\l in {0.5/0/adg,0.5/1.1/bcef}{
      \node[vertex] (\l) at (\i * \s,\j * \s) {\footnotesize{\l}} ;
    }
    \foreach \i/\j in {adg/bcef}{
      \draw[red, very thick] (\i) -- (\j) ;
    }
    \end{scope}

    \begin{scope}[yshift=-3 * \s cm, xshift=6.5 * \s cm]
    \foreach \i/\j/\l in {1/0.75/abcdefg}{
      \node[vertex] (\l) at (\i * \s,\j * \s) {\tiny{\l}} ;
    }
    \end{scope}
    
    \end{tikzpicture}
  \caption{A 2-sequence of contractions to a single vertex shows that the original graph has twin-width at most~2.}
  \label{fig:twin-contraction}
\end{figure}
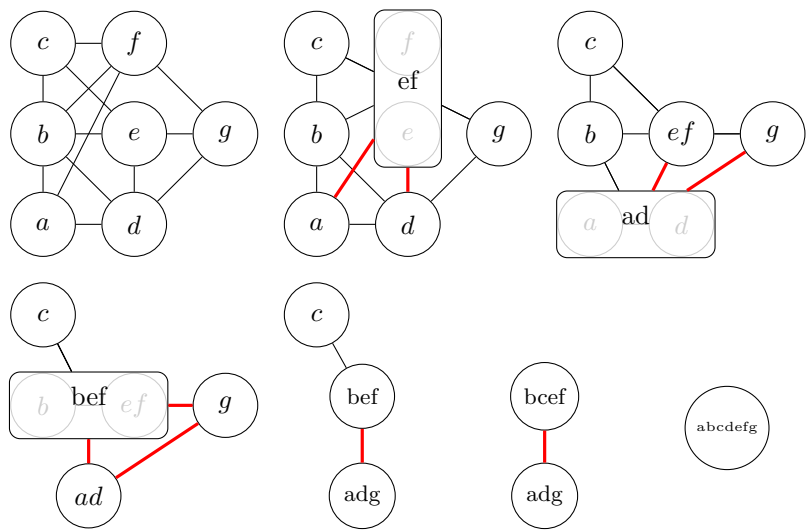

\subsection{How to compute the contraction sequences?}

Given an arbitrary graph or binary structure, it seems tremendously hard to compute a good --let alone, optimum-- contraction sequence.
Fortunately on classes with bounded twin-width, for which this endeavor is algorithmically useful (in light of \cref{thm:main}), we can often exploit structural properties of the class to achieve our goal.
In \cref{sec:first-ex} we present a simple polynomial-time algorithm outputting a $(2^{k+1}-1)$-contraction sequence on graphs of boolean-width at most~$k$ (see~\cref{thm:boolean-width}) and a linear-time algorithm for a $3d$-contraction sequence of (subgraphs of) the $d$-dimensional grid of side-length $n$ (see~\cref{thm:grids}).
The bottleneck for the former algorithm would lie in finding the boolean-width decomposition in the first place.
The latter result enables to find in polynomial time $(3 \lceil \sqrt{d} \rceil)^d k$-contraction sequences for unit $d$-dimensional ball graphs with clique number $k$, provided the geometric representation is given. 

For other classes, such as planar graphs, directly finding the sequence proves challenging.
Therefore we design in \cref{sec:grid-theorem} a framework that reduces this task to finding an ordering $\sigma$ --later called \emph{mixed-free order}-- of the $n$ vertices such that the adjacency matrix $A$ written compliantly to $\sigma$ is simple.
Here by ``simple'' we mean that $A$ cannot be divided into a large number of blocks of consecutive rows and columns, such that no cell of the division is vertical (repetition of the same row subvector) or horizontal (repetition of the same column subvector).
An important local object to handle this type of division is the notion of \emph{corner}, namely a consecutive 2-by-2 submatrix which is neither horizontal nor vertical.
The principal ingredient to show that simple matrices have bounded twin-width is the use of a theorem by Marcus and Tardos~\cite{MarcusT04} which states that $n \times n$ 0,1-matrices with at least $cn$ 1 entries (for a large enough constant~$c$) admit large divisions with at least one 1 entry in each cell.
This result is at the core of Guillemot and Marx's algorithm \cite{Guillemot14} to solve \textsc{Permutation Pattern} in linear FPT time.
As we now apply the Marcus-Tardos theorem to the corners (and not the 1 entries), we bring this engine to the dense setting.
Indeed the matrix can be packed with 1 entries, and yet we learn something non-trivial from the number of corners.

By the Marcus-Tardos theorem the number of corners cannot be too large, otherwise the matrix would not be simple.
From this fact, we are eventually able to find two rows or two columns with sufficiently small Hamming distance.
Therefore they can be contracted.
Admittedly some technicalities are involved to preserve the simplicity of the matrix throughout the contraction process.
So we adopt a two-step algorithm: In the first step, we build a sequence of partition coarsenings over the matrix, and in the second step, we extract the actual sequence of contractions.
The overall algorithm taking $A$ (or $\sigma$) as input, and outputting the contraction sequence, takes polynomial time in~$n$.
It can be implemented in quadratic time, or even faster if instead of the raw matrix, we get a list of pointers to corners of $A$.

We shall now find mixed-free orders.
\cref{sec:bounded-twinwidth} is devoted to this task for three different classes.
Dealing with permutations avoiding a fixed pattern (equivalently, a proper subclass of posets of dimension~2), the order is easy to find: it is imposed.
For posets of bounded width (that is, maximum size of an antichain or minimum size of a chain partition), a mixed-free order is attained by putting the chains in increasing order, one after the other.
Finally for $K_t$-minor free graphs, a Hamiltonian path would provide a good order.
As we cannot always expect to find a Hamiltonian path, we simulate it by a specific Lex-DFS.
The top part of \cref{fig:workflow} provides a visual summary of this section.

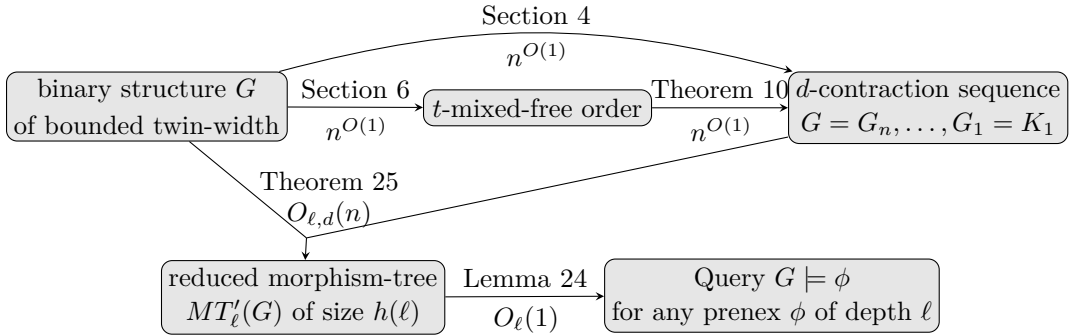
\begin{figure}
  \centering
  \begin{tikzpicture}
    \tikzstyle{shade}=[preaction={fill=black!10},draw,rounded corners]
    \def\d{0.45}
    \def\z{5.14}
    \node (g1) at (0,0) {binary structure $G$} ;
    \node (g2) at (0,-\d) {of bounded twin-width} ;
    \node[shade, inner sep=-0.03cm,fit=(g1) (g2)] (g) {} ;
    \node at (0,0) {binary structure $G$} ;
    \node at (0,-\d) {of bounded twin-width} ;

    \node[shade] (o) at (\z,-\d / 2) {$t$-mixed-free order} ;

    \node (s1) at (2*\z,0) {$d$-contraction sequence} ;
    \node (s2) at (2*\z,-\d) {$G=G_n, \ldots, G_1=K_1$} ;
    \node[shade, inner sep=-0.03cm,fit=(s1) (s2)] (s) {} ;
    \node at (2*\z,0) {$d$-contraction sequence} ;
    \node at (2*\z,-\d) {$G=G_n, \ldots, G_1=K_1$} ;
    
    \draw[-stealth] (g) -- node[align=center,text width=1.8cm,midway,above] {\cref{sec:bounded-twinwidth}} node[align=center,text width=1.8cm,midway,below] {$n^{O(1)}$} (o) ;
    \draw[-stealth] (o) -- node[align=center,text width=1.8cm,midway,above] {\cref{thm:gridtheorem}} node[align=center,text width=1.8cm,midway,below] {$n^{O(1)}$} (s) ;

    \draw[-stealth] (g) to [bend left=15] node[align=center,text width=1.8cm,midway,above] {\cref{sec:first-ex}} node[align=center,text width=1.8cm,midway,below] {$n^{O(1)}$} (s) ;

    \node (mt1) at (0.4 * \z,-2.5) {reduced morphism-tree} ;
    \node (mt2) at (0.4 * \z,-\d-2.5) {$MT'_\ell(G)$ of size $h(\ell)$} ;
    \node[shade, inner sep=-0.03cm,fit=(mt1) (mt2)] (mt) {} ;
    \node at (0.4 * \z,-2.5) {reduced morphism-tree} ;
    \node at (0.4 * \z,-\d-2.5) {$MT'_\ell(G)$ of size $h(\ell)$} ;

    \node (q1) at (1.6 * \z,-2.5) {Query $G \models \phi$} ;
    \node (q2) at (1.6 * \z,-\d-2.5) {for any prenex $\phi$ of depth $\ell$} ;
    \node[shade, inner sep=-0.03cm,fit=(q1) (q2)] (q) {} ;
    \node at (1.6 * \z,-2.5) {Query $G \models \phi$} ;
    \node at (1.6 * \z,-\d-2.5) {for any prenex $\phi$ of depth $\ell$} ;

    \node at (2.4,-1.2) {\cref{thm:FOmodelchecking2}} ;
    \node at (2.4,-1.2-\d) {$O_{\ell,d}(n)$} ;
    \coordinate (x) at (2.1,-1.2-\d-0.3) ;

    \draw (g) -- (x) ;
    \draw (s) -- (x) ;
    \draw[-stealth] (x) -- (mt) ;

     \draw[-stealth] (mt) -- node[align=center,text width=1.8cm,midway,above] {\cref{lem:obs-reduct}} node[align=center,text width=1.8cm,midway,below] {$O_\ell(1)$} (q) ;
  \end{tikzpicture}
  \caption{The overall workflow. Two paths are possible to get a $d$-contraction sequence from a bounded twin-width structure $G$. Either a direct polytime algorithm as for bounded boolean-width, or via a domain-ordering yielding a $t$-mixed free matrix followed by \cref{thm:gridtheorem} which converts it into a $d$-contraction sequence. From there, a tree of constant size (function of $\ell$ only) can be computed in linear FPT time. This tree captures the evaluation of all prenex sentences $\phi$ on $\ell$ variables for $G$. Queries ``$G \models \phi$'' can then be answered in constant time.}
  \label{fig:workflow}
\end{figure}

\subsection{How general are classes of bounded twin-width?}

As announced in the previous section, we will show that proper minor-closed classes have bounded twin-width.
As far as we know, all classes of polynomial expansion may also have bounded twin-width.
However on the one hand, as we will show in an upcoming paper~\cite{twin-width2}, cubic graphs have unbounded twin-width, whereas on the other hand, cliques have twin-width~0.
Thus bounded twin-width is incomparable with bounded degree, bounded expansion, and nowhere denseness.
Examples of graphs for which it is easy to show \emph{unbounded twin-width} include line graphs of bipartite complete graphs (also known as \emph{rook graphs}), high-degree graphs with girth at least~5, and Erd\H{o}s-R\'enyi random graphs drawn from $\mathcal G(n,1/2)$.
Indeed in all three cases, the first contraction would already create a vertex with large red degree, since no pair of \emph{near-twins} exists.

Nowhere dense classes are \emph{stable}, that is, no arbitrarily-long total order can be first-order interpreted from graphs of this class.
In particular, unit interval graphs are not FO interpretations (even FO \emph{transductions}, where in addition copying the structure and \emph{coloring} it with a constant number of unary relations is allowed) of nowhere dense graphs.
Thus even any class of first-order transductions of nowhere dense graphs, called \emph{structurally nowhere dense}, is incomparable with bounded twin-width graphs.
There have been recent efforts aiming to eventually show that FO model checking is fixed-parameter tractable on any structurally nowhere dense class.
Gajarsk\'y et al.~\cite{Gajarsky16} introduce near-uniform classes based on a so-called near-$k$-twin relation, and the equivalent near-covered classes.
They show that FO model checking admits an FPT algorithm on near-covered classes, and that these classes correspond to FO interpretations (even transductions) of bounded-degree graph classes.
Let us observe that the near-$k$-twin relation, as well as the related neighborhood diversity~\cite{Lampis12}, can be thought as a static version of our twin-width.
Gajarsk\'y et al.~\cite{Gajarsky18} gave the first step towards an FPT algorithm on classes with structurally bounded expansion by characterizing them via low shrub-depth decompositions.
A second step was realized by Gajarsk\'y and Kreutzer who presented a direct FPT algorithm computing shrub-depth decompositions~\cite{Gajarsky20}.

Despite cubic graphs having unbounded twin-width, some particular classes with bounded degree, such as subgraphs of $d$-dimensional grids, have bounded twin-width.
More surprisingly, some classes of expanders, will be shown to have bounded twin-width~\cite{twin-width2}.  
This showcases the ubiquity of bounded twin-width, and the wide scope of \cref{thm:main}.
As we will generalize twin-width to matrices, in order to handle permutations, posets, and digraphs, we can potentially define a twin-width notion on hypergraphs, groups, and lattices.
Furthermore we will see next that FO transductions preserve bounded twin-width.

As we saw, \emph{bounded twin-width} proves to be quite rich.
The main algorithmic application presented in this paper is the design of a linear-time FPT algorithm for FO model checking on binary structures with bounded twin-width, provided a sequence of $d$-contractions is given.

\subsection{FO model checking}

A natural algorithmic question given a graph class $\mathcal C$ (i.e., a set of graphs closed under isomorphism) is whether or not deciding first-order formulas $\varphi$ on graphs $G \in \mathcal C$ can be done in time whose superpolynomial blow-up is a function of $|\varphi|$ and $\mathcal C$ only.
A line of works spanning two decades settled this question for monotone (that is, closed under taking subgraphs) graph classes.
It was shown that one can decide first-order (FO) formulas in fixed-parameter time (FPT) in the formula size on bounded-degree graphs~\cite{Seese96}, planar graphs, and more generally, graphs with locally bounded treewidth~\cite{Frick01}, $H$-minor free graphs~\cite{Flum01}, locally $H$-minor free graphs~\cite{Dawar07}, classes with (locally) bounded expansion~\cite{Dvorak13}, and finally nowhere dense classes \cite{Grohe17}.
The latter result generalizes all previous ones, since nowhere dense graphs contain all the aforementioned classes.
Let us observe that the dependency on $|V(G)|$ of the FPT model checking algorithm on classes with bounded expansion is linear \cite{Dvorak13}, while it is almost linear (i.e., $|V(G)|^{1+\varepsilon}$ for every $\varepsilon > 0$) for nowhere dense classes \cite{Grohe17}.
In sharp contrast, if a monotone class $\mathcal C$ is not nowhere dense then FO model checking on $\mathcal C$ is AW[$*$]-complete~\cite{Kreutzer09}, hence highly unlikely to be FPT.
Thus the result of Grohe et al.~\cite{Grohe17} gives a final answer in the case of monotone classes.
We refer the reader interested in structural and algorithmic properties of nowhere dense classes to Nestril and Ossona de Mendez's book~\cite{sparsity}.

Since then, the focus has shifted to the complexity of model checking on (dense) non-monotone graph classes.
Our main result is that FO model checking is FPT on classes with bounded twin-width.
More precisely, we show that:
\begin{theorem}\label{thm:main}
  Given an $n$-vertex (di)graph $G$, a sequence of $d$-contractions $G=G_n, G_{n-1},$ $\ldots, G_1=K_1$, and a first-order sentence $\varphi$, we can decide $G \models \varphi$ in time $f(|\varphi|,d) \cdot n$ for some computable, yet non-elementary, function $f$.
\end{theorem}

This unifies and extends known FPT algorithms for
\begin{compactitem}
\item $H$-minor free graphs~\cite{Flum01},
\item posets of bounded width (i.e., size of the largest antichain)~\cite{Gajarsky15},
\item permutations avoiding a fixed pattern~\cite{Guillemot14}\footnote{Guillemot and Marx show that \textsc{Permutation Pattern} (not FO model checking in general) is FPT when the host permutation avoids a pattern, then a win-win argument proper to \textsc{Permutation Pattern} allows them to achieve an FPT algorithm for the class of \emph{all} permutations.} and hereditary (that is, closed under taking induced subgraphs) proper subclasses of permutation graphs,
\item graphs of bounded rank-width or bounded clique-width~\cite{Courcelle00},\footnote{for this class, even deciding MSO$_1$ is FPT, which is something that we do not capture.}
\end{compactitem}
since we will establish that these classes have bounded twin-width, and that, on them, a sequence of $d$-contractions can be found efficiently.
By transitivity, this also generalizes the FPT algorithm for $L$-interval graphs~\cite{Ganian15}, and may shed a new unified light on geometric graph classes for which FO model checking is FPT~\cite{Hlineny19}.
In that direction we show that a large class of geometric intersection graphs with bounded clique number, including $K_t$-free unit $d$-dimensional ball graphs, admits such an algorithm.
We also show that map graphs have bounded twin-width but we only provide a $d$-contraction sequence when the input comes with a planar embedding of the map. 
FO model checking was proven FPT on map graphs even when no geometric embedding is provided~\cite{Eickmeyer17}.
See \cref{fig:hasse} for the Hasse diagram of classes with a fixed-parameter tractable FO model checking.

\begin{figure}
  \centering
  \begin{tikzpicture}[
      nd/.style={draw,rectangle,rounded corners}]
    \node[nd] (pl) at (0,0) {planar} ;
    \node[nd] (mc) at (0,1) {proper minor-closed} ;
    \node[nd] (pe) at (0,2) {polynomial expansion} ;
    \node[nd] (be) at (0,3) {bounded expansion} ;
    \node[nd] (nd) at (0,4) {nowhere dense} ;

    \node (bd1) at (3,2) {bounded} ;
    \node (bd2) at (3,1.5) {degree} ;
    \node[draw, rounded corners, inner sep=-0.05cm, fit=(bd1) (bd2)] (bd) {} ;

    \draw (pl) -- (mc) -- (pe) -- (be) -- (nd) ;
    \draw (bd) -- (be.east) ;

    \node[draw, fill opacity=0.15, fill=green, rounded corners, fit=(pl) (pe) (nd) (bd)] (sp) {} ;
    \node at (3.3,0.3) {sparse} ;
    \node at (3.3,-0.1) {classes} ;
    
    \node[nd, thick] (bt) at (-4.5,4) {\textbf{bounded twin-width}} ;
    \node (br1) at (-7.5,2.15) {bounded} ;
    \node (br2) at (-7.5,1.65) {rank-width} ;
    \node[draw, rounded corners, inner sep=-0.05cm, fit=(br1) (br2)] (br) {} ;
    \node[nd] (cg) at (-7.5,0.8) {cographs} ;

    \draw (bt) -- (mc.west) ;
    \draw (bt) -- (br) -- (cg) ;
    \draw[dashdotted] (bt) -- (pe.west) ;

    \node (pbw1) at (-5.5,2.5) {posets of} ;
    \node (pbw2) at (-5.5,2.1) {bounded} ;
    \node (pbw3) at (-5.5,1.7) {width} ;
    \node[draw, rounded corners, inner sep=-0.05cm, fit=(pbw1) (pbw3)] (pbw) {} ;
    \node[nd] (li) at (-5.5,0.9) {$L$-interval} ;
    \node[nd] (ui) at (-5.5,0) {unit interval} ;

    \draw (bt) -- (pbw) -- (li) -- (ui) ;

    \node (pa1) at (-3.3,1.5) {pattern} ;
    \node (pa2) at (-3.3,1.1) {avoiding} ;
    \node (pa3) at (-3.3,0.7) {permuta-} ;
    \node (pa4) at (-3.3,0.3) {tions} ;
    \node[draw, rounded corners, inner sep=-0.05cm, fit=(pa1) (pa2) (pa3) (pa4)] (pa) {} ;

    \draw (bt) -- (pa) ;

    \node (mg1) at (-8,3.2) {map} ;
    \node (mg2) at (-8,2.8) {graphs} ;
    \node[draw, rounded corners, inner sep=-0.05cm, fit=(mg1) (mg2)] (mg) {} ;

    \draw (bt) -- (mg) ;

    \node[draw, fill opacity=0.15, fill=red, rounded corners, fit=(bt) (pa) (ui) (br) (cg) (mg)] (de) {} ;
    
    \node at (-8,0.3) {dense} ;
    \node at (-8,-0.1) {classes} ;
  \end{tikzpicture}
  \caption{Hasse diagram of classes on which FO model checking is FPT, with the newcomer twin-width. The dash-dotted edge means that polynomial expansion may well be included in bounded twin-width. Bounded twin-width and nowhere dense classes roughly subsume all the current knowledge on the fixed-parameter tractability of FO model checking. Do they admit a natural common superclass still admitting an FPT algorithm for FO model checking?}
  \label{fig:hasse}
\end{figure}
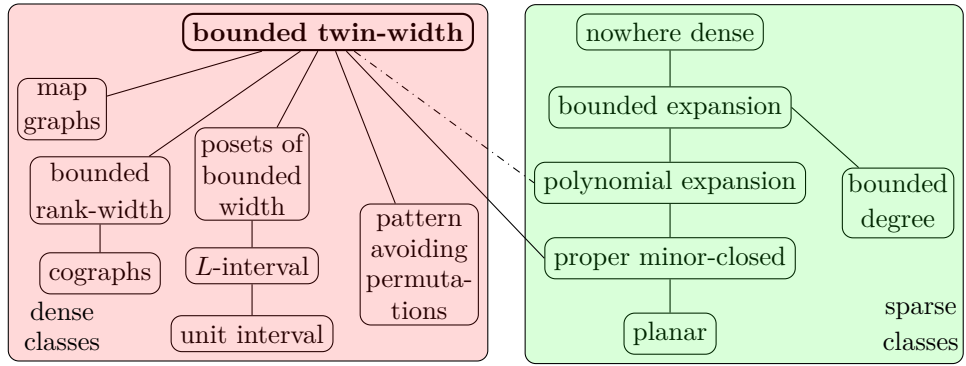

Permutation patterns can be represented as posets of dimension 2.
Any proper hereditary subclass of posets of dimension~2 contains all permutations avoiding a fixed pattern. 
In turn, posets can be encoded by directed graphs (or digraphs), with an arc from $u$ to $v$ if \emph{$u$ is smaller than $v$}.
Thus we formulated \cref{thm:main} with graphs and digraphs, to cover all the classes of bounded twin-width listed after the theorem (in particular, permutations excluding a fixed pattern).
Twin-width and the applicability of \cref{thm:main} is actually broader: one may replace ``an $n$-vertex (di)graph $G$'' by ``a~binary structure $G$ on a domain of size $n$'' in the statement of the theorem, where a binary structure is a finite set of binary relations.

\paragraph*{Roadmap for the proof of \cref{thm:main}.}
Instead of deciding ``$G \models \varphi$'' for a specific sentence $\varphi$, we build in FPT time a tree which contains enough information to answer all the queries of the form ``\emph{is $\phi$ true on $G$?},'' for every prenex sentence $\phi$ on $\ell$ variables.
A prenex sentence $\phi$ starts with a quantification (existential and universal) over the $\ell$ variables, followed, in the case of graphs, by a Boolean combination $\phi'(x_1, \ldots, x_\ell)$ of atoms of the form $x=y$ (interpreted as: vertex~$x$ is vertex~$y$) and $E(x,y)$ (interpreted as: there is an edge between $x$ and $y$).
A simple but important insight is that once Existential and Universal players have chosen the assignment $v_1, \ldots, v_\ell$, the truth of $\phi'(v_1, \ldots, v_\ell)$ only depends on the induced subgraph $G[\{v_1,\ldots,v_\ell\}]$ and the pattern of equality classes of the tuple $(v_1,\ldots,v_\ell)$.
Indeed the latter pair carries the truth value of each possible atom.

Imagine now the complete tree of all the possible ``moves'' assigning vertex $v_i$ to variable $x_i$.
Let us call it the \emph{game tree} for now (later it will be called~\emph{morphism-tree}).
This tree has arity $|V(G)|$ and depth $\ell$.
Thus it is too large to explicitly compute.
However, up to labeling its different levels with $\exists$ and $\forall$, it contains what is needed to evaluate any $\ell$-variable prenex formula on $G$.
It actually contains way too much information.
Assume, for instance, that two of its leaves $v_\ell, v'_\ell$ with the same parent node define the same induced subgraph $G[\{v_1,\ldots,v_{\ell-1},v_\ell\}] \cong G[\{v_1,\ldots,v_{\ell-1},v'_\ell\}]$ and the same pattern of equality classes.
Then it is safe to delete the ``move $v'_\ell$'' from the possibilities of whichever player shall play at level~$\ell$.
Indeed ``move $v_\ell$'' is perfectly equivalent: As it sets to true the same list of atoms, it will satisfy exactly the same formulas $\phi'$, irrelevant of the nature of the quantifier preceding~$x_\ell$.
This notion of equivalent sibling nodes can be generalized to any level of the game tree.
If one iteratively deletes equivalent moves (and their subtrees) while possible, it can be observed that the resulting tree is of size bounded by $\ell$ only.
We call \emph{reduct} such a tree.

Now the contraction sequence comes in.
Actually, more convenient here than the successive trigraphs $G=G_n, G_{n-1}, \ldots, G_1$, we consider the corresponding partition sequence: $\mathcal P_n,  \mathcal P_{n-1}, \ldots, \mathcal P_1$, where $\mathcal P_i$ is the partition $\{u(G)~|~u \in V(G_i)\}$ of $V(G)$.
Recall that $u(G)$ denotes the set of vertices of $G$ contracted into the single vertex $u \in V(G_i)$.
Recall also that two parts of $\mathcal P_i$ are \emph{homogeneous} if they are fully adjacent or fully non-adjacent in $G$.
Let $G_{\mathcal P_i}$ be the graph whose vertices are the parts of $\mathcal P_i$, and edges link every pair of non-homogeneous parts.
This graph is actually made of the red edges of trigraph $G_i$.
We extend game trees and their reducts to partitioned graphs $(G,\mathcal P_i)$, where equivalent moves have to further respect the partition.
More specifically we are interested in reducts of \emph{local game trees}, i.e., game trees where all the moves are played in the close neighborhood of a fixed vertex of $G_{\mathcal P_i}$, or equivalently a fixed part of $\mathcal P_i$.

By dynamic programming, we will maintain for $i$ going from $n$ down to 1, every game tree local to part $P \in \mathcal P_i$.
$\mathcal P_n$ is a partition into singletons $\{v\}$ (for each $v \in V(G)$), so the local game tree is easy to determine, and is naturally a reduct.
Indeed all the variables can only be instantiated to $v$, hence a simple tree of out-degree~1.
$\mathcal P_1$ is the trivial partition $\{V(G)\}$.
So the reduct of its local game tree coincides with the reduct of the (global) game tree, which is exactly what we are looking for.

Say that, to go from $\mathcal P_{i+1}$ to $\mathcal P_i$, we fuse two sets $X'_i, X''_i$ into $X_i$.
We shall now update the reducts of the local game trees in $(G,\mathcal P_i)$.
For the parts that are far enough from $X_i$, the local game trees (and their reducts) are unchanged.
Thus no update is needed.
This is because these parts are too far to ``interact'' with $X_i$ via non-homogeneous pairs of parts.

We therefore focus on the parts $P$ that are close to $X_i$ in $G_{\mathcal P_i}$.
We first combine, by a \emph{shuffle} operation, a bounded (by a function of the depth $\ell$ and the twin-width $d$) number of reducts of game trees that are local to parts $P'$ sufficiently close to $P$.
We then strategically prune redundant nodes, and delete further equivalent nodes.
The aggregation of the two former steps is dubbed \emph{pruned shuffle} and is the central operation of our algorithm.
To finally obtain the desired updated reduct, we \emph{project} the pruned shuffle on the nodes that are inherently \emph{rooted at $P$}.
To be formalized the latter requires to introduce an auxiliary graph, called \emph{tuple graph}, and a notion of \emph{local root}.
These objects are instrumental in handling overlap or redundant information.

A crucial aspect of the algorithm relies on the following fact, reminiscent of the Feferman-Vaught theorem~\cite{Feferman67}.
If two connected subsets, say, $X$ and $Y$ of $G_{\mathcal P_{i+1}}$ are united in $G_{\mathcal P_i}$, the reducts of games trees local to a part of $X \cup Y$ are simply obtained by interleaving (actually \emph{shuffling}) the reducts of game trees local to parts of $X$ with reducts of game trees local to parts of $Y$.
Indeed pairs of parts in $(X,Y)$ are by construction homogeneous to each other, so the precise choices of vertices within these parts is immaterial.
We finally observe that at each step $i$, we are updating a bounded number of reducts of bounded size.
Therefore the overall algorithm takes linear FPT time (see bottom part of \cref{fig:workflow}).

We take a very combinatorial stance towards FO model checking.
Formulas are quickly converted into trees whose nodes are naturally mapped to subgraphs induced by tuples.
Our use of the bounded-degree graphs $G_{\mathcal P_i}$ (red graphs) should remind of Gaifman's locality theorem~\cite{gaifman82}. 
And indeed, it is an exact transcription of it in combinatorial terms. 
Apart from the fact that every sentence can be put in prenex normal form, our algorithm and its presentation in \cref{sec:fo} are self-contained, thereby not assuming from the reader any knowledge in finite model theory.
As a by-product of the algorithm, we will show that bounded twin-width is preserved under FO interpretations and transductions, which makes it a robust class as far as first-order model checking is concerned.

\subsection{Organization of the paper}

\cref{sec:prelim} gives the necessary graph-theoretic and logic background. 
In \cref{sec:def} we formally introduce contraction sequences and the twin-width of a graph.
In \cref{sec:first-ex} we get familiar with these new notions.
In particular we show with direct arguments that bounded rank-width graphs, $d$-dimensional grids, and unit $d$-dimensional ball graphs with bounded clique number, have bounded twin-width.
In \cref{sec:grid-theorem} we extend twin-width to matrices and show a grid-minor-like theorem, which informally states that a graph has large twin-width if and only if all its vertex orderings yield an adjacency matrix with a complex large submatrix.  
This turns out to be a useful characterization for the next section.
In \cref{sec:bounded-twinwidth} we show how, thanks to this characterization, we can compute a witness of bounded twin-width, for permutations avoiding a fixed pattern, comparability graphs with bounded independence number (equivalently, bounded-width posets), and $K_t$-minor free graphs.
In \cref{sec:fo} we present a linear-time FPT algorithm for FO model checking on graphs given with a witness of bounded twin-width.
In \cref{sec:fo-inter} we show that FO interpretations (even transductions) of classes of bounded twin-width still have bounded twin-width.
Finally in \cref{sec:conclusion} we list a handful of promising questions left for future work.


\section{Preliminaries}\label{sec:prelim}

We denote by $[i,j]$ the set of integers $\{i,i+1,\ldots, j-1, j\}$, and by $[i]$ the set of integers $[1,i]$.
If $\mathcal X$ is a set of sets, we denote by $\cup \mathcal X$ the union of them.

\subsection{Graph definitions and notations}

All our graphs are undirected and simple (no multiple edge nor self-loop).
We denote by $V(G)$, respectively $E(G)$, the set of vertices, respectively of edges, of the graph $G$. 
For $S \subseteq V(G)$, we denote the \emph{open neighborhood} (or simply \emph{neighborhood}) of $S$ by $N_G(S)$, i.e., the set of neighbors of $S$ deprived of $S$, and the \emph{closed neighborhood} of $S$ by $N_G[S]$, i.e., the set $N_G(S) \cup S$.
For singletons, we simplify $N_G(\{v\})$ into $N_G(v)$, and $N_G[\{v\}]$ into $N_G[v]$.
We denote by $G[S]$ the subgraph of $G$ induced by $S$, and $G - S := G[V(G) \setminus S]$.
For $A, B \subseteq V(G)$, $E(A,B)$ denotes the set of edges in $E(G)$ with one endpoint in $A$ and the other one in $B$.
Two distinct vertices $u, v$ such that $N(u) = N(v)$ are called \emph{false twins}, and \emph{true twins} if $N[u] = N[v]$.
In particular, true twins are adjacent.
Two vertices are \emph{twins} if they are false twins or true twins.
If $G$ is an $n$-vertex graph and $\sigma$ is a total ordering of $V(G)$, say, $v_1, \ldots, v_n$, then $A_\sigma(G)$ denotes the adjacency matrix of $G$ in the order $\sigma$.
Thus the entry in the $i$-th row and $j$-th column is a 1 if $v_iv_j \in E(G)$ and a 0 otherwise.   

The length of a path in an unweighted graph is simply the number of edges of the path.
For two vertices $u, v \in V(G)$, we denote by $d_G(u,v)$, the distance between $u$ and $v$ in $G$, that is the length of the shortest path between $u$ and $v$.
The diameter of a graph is the longest distance between a pair of its vertices.
In all the above notations with a subscript, we omit it whenever the graph is implicit from the context.

An \emph{edge contraction} of two adjacent vertices $u, v$ consists of merging $u$ and $v$ into a single vertex adjacent to $N(\{u,v\})$ (and deleting $u$ and $v$).
A graph $H$ is a \emph{minor} of a graph $G$ if $H$ can be obtained from $G$ by a sequence of vertex and edge deletions, and edge contractions.
A graph $G$ is said \emph{$H$-minor free} if $H$ is not a minor of $G$.
Importantly we will overload the term ``contraction''.
In this paper, we call \emph{contraction} the same as an edge contraction without the requirement that the two vertices $u$ and $v$ are adjacent.
This is sometimes called an \emph{identification}, but we stick to the shorter \emph{contraction} since we will use that word often.
In the very rare cases in which we actually mean the classical (edge) contraction, the context will lift the ambiguity. 
We will also somewhat overload the term ``minor''.
Indeed, in \cref{sec:grid-theorem} we introduce the notions of ``$d$-grid minor'' and ``$d$-mixed minor'' on matrices.
They are only loosely related to (classical) graph minors, and it will always be clear which notion is meant.

\subsection{First-order logic, model checking, FO interpretations/transductions}

For our purposes, we define first-order logic without function symbols.
A finite \emph{relational signature} is a set $\tau$ of relation (or \emph{predicate}) symbols given with their arity $\{R^1_{a_1}, \ldots, R^h_{a_h}\}$; that is, relation $R^i_{a_i}$ has arity $a_i$.  
A first-order formula $\phi \in \text{FO}(\tau)$ over $\tau$ is any string generated from letter $\psi$ by the grammar: $$\psi \rightarrow \exists x \psi,~\forall x \psi,~\psi \lor \psi,~\psi \land \psi,~\neg \psi,~(\psi),~R^1_{a_1}(x, \ldots, x),~\ldots,~R^h_{a_h}(x, \ldots, x),~x=x,~\text{and}$$ $$x \rightarrow x_1,x_2, \ldots \text{~an infinite set of fresh variable labels.}$$

For the sake of simplicity, we will further impose that the same label cannot be reused for two different variables.
A variable $x_i$ is then said \emph{quantified} if it appears next to a quantifier ($\forall x_i$ or $\exists x_i$), and \emph{free} otherwise.
We usually denote by $\phi(x_{f_1}, \ldots, x_{f_h})$ a formula whose free variables are precisely $x_{f_1}, \ldots, x_{f_h}$.
A formula without quantified variables is said \emph{quantifier-free}.
A~\emph{sentence} is a formula without free variables.
With our simplification that the same label is not used for two distinct variables, when a formula $\phi$ contains a subformula $Qx_i \phi'$ (with $Q \in \{\exists, \forall\}$), all the occurrences of $x_i$ in $\phi$ lie in $\phi'$.

\paragraph*{Model checking.}
A first-order (FO) formula is purely syntactical.
An \emph{interpretation}, \emph{model}, or \emph{structure} $\mathcal M$ of the FO language $\text{FO}(\tau)$ specifies a \emph{domain of discourse} $D$ for the variables, and a relation $\mathcal M(R^i_{a_i}) = R^i \subseteq D^{a_i}$ for each symbol $R^i_{a_i}$.
$\mathcal M$ is sometimes called a \emph{$\tau$-structure}.
$\mathcal M$ is a \emph{binary structure} if $\tau$ has only relation symbols of arity~2.
It is said \emph{finite} if the domain $D$ is finite.
A sentence $\phi$ interpreted by $\mathcal M$ is \emph{true}, denoted by $\mathcal M \models \phi$, if it evaluates to true with the usual semantics for quantified Boolean logic, the equality, and $R^i_{a_i}(d_1, \ldots, d_{a_i})$ is true if and only if $(d_1, \ldots, d_{a_i}) \in \mathcal M(R^i_{a_i})$.
For a fixed interpretation, a formula $\phi$ with free variables $x_{f_1}, \ldots, x_{f_h}$ is \emph{satisfiable} if $\exists x_{f_1} \cdots \exists x_{f_h} \phi$ is true.

In the FO model checking problem, given a first-order sentence $\phi \in \text{FO}(\tau)$ and a finite model $\mathcal M$ of $\text{FO}(\tau)$, one has to decide whether $\mathcal M \models \phi$ holds.
The input size is $|\phi|+|\mathcal M|$, the number of bits necessary to encode the sentence $\phi$ and the model $\mathcal M$.
The brute-force algorithm decides $\mathcal M \models \phi$ in time $|\mathcal M|^{|\phi|}$, by building the tree of all possible assignments.
We will consider $\phi$ to be fixed or rather small compared to $|\mathcal M|$.
Therefore we wish to find an FPT algorithm for FO model checking parameterized by $|\phi|$, that is, running in time $f(|\phi|) |\mathcal M|^{O(1)}$, or even better $f(|\phi|) |D|$.

\defparproblem{\textsc{FO($\tau$) Model Checking}}{A $\tau$-structure $\mathcal M$ and a sentence $\phi$ of $\text{FO}(\tau)$.}{$|\phi|$}{Does $\mathcal M \models \phi$ hold?}

We restrict ourselves to FO model checking on finite binary structures, for which twin-width will be eventually defined. 
For the most part, we will consider FO model checking on graphs (and we may omit the signature $\tau$).
Let us give a simple example.
Let $\tau=\{E_2\}$ be a signature with a single binary relation.
Finite models of the language $\text{FO}(\tau)$ correspond to finite directed graphs with possible self-loops.
Let $\phi$ be the sentence $\exists x_1 \exists x_2 \cdots \exists x_k \bigland_{i < j} \neg(x_i = x_j)\land \bigland_{i \neq j} \neg E(x_i,x_j)$.
Let $G$ be a $\tau$-structure or graph.
$G \models \phi$ holds if and if $G$ has an independent set of size $k$.
This particular problem parameterized by $|\phi|$ (or equivalently $k$) is W[1]-hard on general graphs.
However it may admit an FPT algorithm when $G$ belongs to a specific class of graphs, as in the case, for instance, of planar graphs or bounded-degree graphs.  

\paragraph*{FO interpretations and transductions.}
An FO interpretation of a $\tau$-structure $\mathcal M$ is a $\tau$-structure $\mathcal M'$ such that for every relation $R$ of $\mathcal M'$, $R(a_1, \ldots, a_h)$ is true if and only if $\mathcal M \models \phi_R(a_1, \ldots, a_h)$ for a fixed formula $\phi_R(x_1, \ldots, x_h) \in \text{FO}(\tau)$.
Informally every relation of $\mathcal M'$ can be characterized by a formula evaluated on $\mathcal M$.

Again we shall give some example on graphs since it is our main focus.
Let $G$ be a simple undirected graph (in particular, $E(x,y)$ holds whenever $E(y,x)$ holds).
Then the \emph{FO ($\phi$-)interpretation $I_\phi(G)$} is a graph $H$ with vertex set $V(G)$ and $uv \in E(H)$ if and only if $G \models \phi(x,y) \land \phi(y,x)$.
If for instance $\phi(x,y)$ is the formula $\neg E(x,y)$, then $I_\phi(G)$ is the complement of $G$.
If instead $\phi(x,y)$ is $E(x,y) \lor \exists z \left( E(x,z) \land E(z,y)\right)$, then $I_\phi(G)$ is the square of $G$.
The FO ($\phi$-)interpretation of a class $\mathcal C$ of graphs is the set of all graphs that are $\phi$-interpretations of graphs in $\mathcal C$, namely $I_\phi(\mathcal C) := \{ H $ $|$ $H = I_\phi(G),~G \in \mathcal C\}$.
It is not very satisfactory that $I_\phi(\mathcal C)$ is not hereditary.
We will therefore either close $I_\phi(\mathcal C)$ by taking induced subgraphs, or use the more general notion of FO transductions (see for instance \cite{Blumensath10}).

An FO transduction is an enhanced FO interpretation.
We give a simplified definition for undirected graphs, but the same definition generalizes to general (binary) structures.
First a \emph{basic FO transduction} is slightly more general than an FO interpretation.
It is a triple $(\delta,\nu,\eta)$, with 0, 1, and, 2 free variables respectively, which maps every graph $G$ such that $G \models \delta$ to the graph $(\{v $ $|$ $G \models \nu(v)\},\{uv$ $|$ $G \models \nu(u) \land \nu(v) \land \eta(u,v)\})$.
Before we apply the basic FO transduction, we allow two operations: an expansion and a copy operation.
An \emph{$h$-expansion} maps a graph $G$ to the set of all the structures obtained by augmenting $G$ with $h$ unary relations $U^1, \ldots, U^h$.
A \emph{$\gamma$-copy operation} maps a graph $G$ to the disjoint union of $\gamma$~copies of $G$, say, $G^1, \ldots, G^\gamma$, where $V(G^j) = \{(v,j) $ $|$ $v \in V(G)\}$.
Moreover, it adds $\gamma$~unary relations $C_1, \ldots, C_\gamma$, and a binary relation $\sim$, where $C_i(v)$ holds whenever $v \in V(G^i)$ and $(u,i) \sim (v,j)$ holds when $u=v$.
Informally the unary relations indicate in which copy a vertex is, while the binary relation $\sim$ links the copies of a same vertex.

Now, the \emph{$(\phi,\gamma,h)$-transduction $\mathcal T_{\phi,\gamma,h}(G)$ of a graph $G$} is the set $\tau \circ \gamma_{\text{op}} \circ h_{\text{op}}(G)$ where $h_{\text{op}}$ is the $h$-expansion, $\gamma_{\text{op}}$ is the $\gamma$-copy operation, and $\tau=(\delta,\nu,\eta)$ is a basic FO transduction.
Note that the formulas $\nu$ and $\eta$ may depend on the edge relation of $G$ as well as all the added unary relations and the binary relation $\sim$.
Similarly to FO interpretations of classes, we define $\mathcal T_{\phi,\gamma,h}(\mathcal C) := \{ H $ $|$ $H \in\mathcal T_{\phi,\gamma,h}(G),~G \in \mathcal C\}$. 

As we will see in \cref{sec:fo-inter}, a worthwhile property of twin-width is that every FO interpretation/transduction of a bounded twin-width class has bounded twin-width itself. 

\section{Sequence of contractions and twin-width}\label{sec:def}

We say that two vertices $u$ and $v$ are \emph{twins} if they have the same neighborhood outside $\{u,v\}$.
A natural operation is to contract (or identify) them and try to iterate the process.
If this algorithm leads to a single vertex, the graph was initially a \emph{cograph}.
Many intractable problems become easy on cographs.
It is thus tempting to try and extend this tractability to larger classes.
One such example is the class of graphs with bounded clique-width (or equivalently bounded rank-width) for which any problem expressible in MSO$_1$ logic can be solved in polynomial-time~\cite{Courcelle00}.
A perhaps more direct generalization (than defining clique-width) would be to allow contractions of near twins, but the cumulative effect of the errors\footnote{By \emph{error} we informally refer to the elements in the (non-empty) symmetric difference in the neighborhoods of the contracted vertices.} stands as a barrier to algorithm design.

An illuminating example is provided by a bipartite graph $G$, with bipartition $(A, B)$, such that for every subset $X$ of $A$ there is a vertex $b \in B$ with neighborhood $X$ in $A$.
Surely $G$ is complex enough so that we should not entertain any hope of solving a problem like, say, \textsc{$k$-Dominating Set} significantly faster on any class containing $G$ than on general graphs.
For one thing, graphs like $G$ contain all the bipartite graphs as induced subgraphs.
Nonetheless $G$ can be contracted to a single vertex by iterating contractions of vertices whose neighborhoods differ on only one vertex.
Indeed, consider $a \in A$ and contract all pairs of vertices of $B$ differing exactly at $a$.
Applying this process for every $a \in A$, we end up by contracting the whole set $B$, and we can eventually contract $A$.

Thus the admissibility of a contraction sequence should not solely be based on the current neighborhoods.
The key idea is to keep track of the past errors in the contraction history and always require all the vertices to be involved in only a limited number of mistakes.
Say the errors are carried by the edges, and an erroneous edge is recorded as \emph{red}.
Note that in the previous contraction sequence of $G$, after contracting all pairs of vertices of $B$ differing at $a$, all the edges incident to $a$ are red, and vertex $a$ witnesses the non-admissibility of the sequence.
Let us now get more formal.

It appears, from the previous paragraphs, that the appropriate structure to define twin-width is a graph in which some edges are colored red.
A \emph{trigraph} is a triple $G=(V,E,R)$ where $E$ and $R$ are two disjoint sets of edges on $V$: the (usual) edges and the \emph{red edges}.
An informal interpretation of a red edge $uv \in R$ is that some errors have been made while handling $G$ and the existence of an edge between $u$ and $v$, or lack thereof, is uncertain.
A trigraph $(V,E,R)$ such that $(V,R)$ has maximum degree at most $d$ is a \emph{$d$-trigraph}.
We observe that any graph $(V,E)$ may be interpreted as the trigraph $(V,E,\emptyset)$.

Given a trigraph $G=(V,E,R)$ and two vertices $u,v$ in $V$, we define the trigraph $G/u,v=(V',E',R')$ obtained by \emph{contracting}\footnote{Or \emph{identifying}. Let us insist that $u$ and $v$ do not have to be adjacent.} $u,v$ into a new vertex $w$ as the trigraph on vertex set $V' = (V \setminus \{u,v\}) \cup \{w\}$ such that $G - \{u,v\}= (G/u,v) - \{w\}$ and with the following edges incident to $w$:
\begin{compactitem}
\item $wx\in E'$ if and only if $ux\in E$ and $vx\in E$,
\item $wx\notin E'\cup R'$ if and only if $ux\notin E\cup R$ and $vx\notin E\cup R$, and
\item $wx\in R'$ otherwise.
\end{compactitem}

In other words, when contracting two vertices $u,v$, red edges stay red, and red edges are created for every vertex $x$ which is not joined to $u$ and $v$ at the same time.
We say that $G/u,v$ is a \emph{contraction of $G$}.
If both $G$ and $G/u,v$ are $d$-trigraphs, $G/u,v$ is a \emph{$d$-contraction}.
We may denote by $V(G)$ the vertex set, $E(G)$ the set of \emph{black} edges, and $R(G)$ the set of \emph{red} edges, of the trigraph $G$. 

A (tri)graph $G$ on $n$ vertices is \emph{$d$-collapsible} if there exists a sequence of $d$-contractions which contracts $G$ to a single vertex.
More precisely, there is a \emph{$d$-sequence} of $d$-trigraphs $G=G_n, G_{n-1}, \ldots, G_2, G_1$ such that $G_{i-1}$ is a contraction of $G_i$ (hence $G_1$ is the singleton graph).
See \cref{fig:twin-contraction} for an example of a sequence of 2-contractions of a 7-vertex graph. 
The minimum $d$ for which $G$ is $d$-collapsible is the \emph{twin-width} of $G$, denoted by $\text{tww}(G)$.

If $v$ is a vertex of $G_i$ and $j \geqslant i$, then $v(G_j)$ denotes the subset of vertices of $G_j$ eventually contracted into $v$ in $G_i$.
Two disjoint vertex subsets $A, B$ of a trigraph are said \emph{homogeneous} if there is no red edge between $A$ and $B$, and there are not both an edge and a non-edge between $A$ and $B$.
In other words, $A$ and $B$ are fully linked by black edges or there is no (black or red) edge between them.
Observe that in any contraction sequence $G=G_n, \ldots, G_i, \ldots, G_1$, there is a red edge between $u$ and $v$ in $G_i$ if and only if $u(G)$ and $v(G)$ are not homogeneous.
We may sometimes (abusively) identify a vertex $v \in G_i$ with the subset of vertices of $G$ contracted to form $v$.

One can check that cographs have twin-width 0 (the class of graphs with twin-width~0 actually coincides with cographs), paths of length at least three have twin-width~1, red paths have twin-width at most~2, and trees have twin-width~2.
Indeed, they are not 1-collapsible, as exemplified by the 1-subdivision of $K_{1,3}$, and they admit the following 2-sequence.
Choose an arbitrary root and contract two leaves with the same neighbor, or, if not applicable, contract the highest leaf with its neighbor.
We observe that in this 2-sequence, every $G_i$ only contains red edges which are adjacent to leaves.
In particular, red edges are either isolated or are contained in a path of length two.


The definition of twin-width readily generalizes to directed graphs, where we create a red edge whenever the contracted vertices $u, v$ are not linked to $x$ in the same way.
This way we may speak of the twin-width of a directed graph or of a partial order.
One could also wish to define twin-width on graphs ``colored'' by a constant number of unary relations.
To have a unifying framework, we will later work with matrices (Section~\ref{sec:grid-theorem}).
Before that, we present in the next section some basic results about twin-width of graphs.

\section{First properties and examples of classes with bounded twin-width}\label{sec:first-ex}

Let us get familiar with contraction sequences and twin-width through simple operations: complementing the graph, taking induced subgraphs, and adding apices.
\subsection{Complementation, induced subgraphs, and adding apices}

The \emph{complement of a trigraph $G$} is the trigraph $\overline G$ obtained by keeping all its red edges while making edges its non-edges, and non-edges its edges.
Thus if $G=(V,E,R)$, then $\overline G=(V,{V \choose 2} \setminus (E \cup R),R)$, and it holds that $\overline{\overline G}=G$.
Twin-width is invariant under complementation.
One can observe that any sequence of $d$-contractions for $G$ is also a sequence of $d$-contractions for $\overline G$.
Indeed there is a red edge between two vertices $u, v$ in a trigraph obtained along the sequence if and only if $u(G)$ and $v(G)$ are homogeneous if and only if $u(\overline G)$ and $v(\overline G)$ are homogeneous.

We can extend the notion of induced subgraphs to trigraphs in a natural way.
A trigraph $H$ is an \emph{induced subgraph} of a trigraph $G$ if $V(H) \subseteq V(G)$, $E(H) = E(G) \cap {H \choose 2}$, and $R(H) = R(G) \cap {H \choose 2}$.
The twin-width of an induced subgraph $H$ of a trigraph $G$ is at most the twin-width of $G$.
Indeed the sequence of contractions for $G$ can be projected to $H$ by just ignoring contractions involving vertices outside $V(H)$.
Then the red degree of trigraphs in the contraction sequence of $H$ is at most the red degree of the corresponding trigraphs in the contraction sequence of $G$.

We now show that adding a vertex linked by black edges to an arbitrary subset of the vertices essentially at most doubles the twin-width.
\begin{theorem}\label{thm:apex}
  Let $G'$ be a trigraph obtained from a trigraph $G$ by adding one vertex $v$ and linking it with black edges to an arbitrary subset $X \subseteq V(G)$.
  Then $\text{tww}(G') \leqslant 2(\text{tww}(G)+1)$.
\end{theorem}
\begin{proof}
  Let $d = \text{tww}(G)$ and let $G=G_n, \ldots, G_1$ be a sequence of $d$-contractions.
  We want to build a good sequence of contractions for $G'$.
  The rules are that, while there are more than three vertices in the trigraph, we never contract two vertices $u$ and $u'$ such that $u(G) \subseteq X$ and $u'(G) \subseteq V(G) \setminus X$, neither do we contract $v$ with another vertex.
  In words, until the very end, we do not touch $v$, and we do only contractions internal to $X$ or to $V(G) \setminus X$.

  We start with $G'$.
  For $i$ ranging from $n$ down to 2, let us denote by $u_i,u'_i$ the $d$-contraction performed from $G_i$ to $G_{i-1}$.
  With our imposed rules, instead of having one set $u_i(G)$ of contracted vertices, we have two: $U_{i,X} := u_i(G) \cap X$ and $U_{i,\overline X} := u_i(G) \setminus X$.
  Similarly we can define the (potentially empty) $U'_{i,X}$ and $U'_{i,\overline X}$ based on $u'_i(G)$.
  Any of these sets, if non-empty, corresponds to a currently contracted vertex, that we denote with the same label.
  In the current trigraph obtained from $G'$, we contract $U_{i,X}$ and $U'_{i,X}$ if they both exist.
  Next we contract $U_{i,\overline X}$ and $U'_{i,\overline X}$ (again if they both exist).
  This preserves our announced invariant, and terminates with a 3-vertex trigraph made of $v$, all the vertices of $X$ contracted in a single vertex, all the vertices of $V(G) \setminus X$ contracted in a single vertex.
  Observe that a 3-vertex trigraph is 2-collapsible and $2 \leqslant 2(\text{tww}(G)+1)$.
  
  We shall finally justify that in the sequence of contractions built for $G'$, all the trigraphs have red degree at most $2(\text{tww}(G)+1)$.
  Before we simulate the contraction $u_i,u'_i$, each contracted vertex $u(G) \cap X$ (resp.~$u(G) \setminus X$) of $G'$ has red degree at most $2d+1$.
  Indeed $u(G) \cap X$ (resp.~$u(G) \setminus X$) can only have red edges to vertices $w(G) \cap X$ and $w(G) \setminus X$ such that $w$ is a red neighbor of $u$, and to $u(G) \setminus X$ (resp.~$u(G) \cap X$).
  After we contract (if they exist) $U_{i,X}$ and $U'_{i,X}$, the newly created vertex, say $U$, has red degree at most $2d+2$.
  The $+2$ accounts for $U_{i,\overline X}$ and $U'_{i,\overline X}$.
  The red degree of $U_{i,\overline X}$ and $U'_{i,\overline X}$ is at most $2d+1$, where the $+1$ accounts for $U$.
  All the other vertices have their red degree bounded by $2d+1$.
  After we also contract (if they exist) $U_{i,\overline X}$ and $U'_{i,\overline X}$, all the vertices have degree at most~$2d+1$.
  Overall the red degree never exceeds $2d+2=2(\text{tww}(G)+1)$.
\end{proof}

The previous result implies that bounded twin-width is preserved by adding a constant number of apices.
In \cref{sec:bounded-twinwidth} we will show a far-reaching generalization of this fact: $H$-minor free graphs have bounded twin-width.
We will not have to resort to the graph structure theorem.
Now if we have a second look at the proof of \cref{thm:apex}, we showed that twin-width does not arbitrarily increase when we add one or a constant number of unary relations (in \cref{sec:grid-theorem} we will formally define twin-width for graphs colored by unary relations, and even for arbitrary matrices on a constant-size alphabet). 
Again we will see in \cref{sec:fo-inter} a considerable generalization of that fact and of the conservation of twin-width by complementation: bounded twin-width classes are closed by first-order transductions.

As cliques have bounded twin-width (more precisely twin-width~0), \emph{bounded twin-width} is not preserved under (non-induced) subgraphs and minors.
Indeed the class of all graphs has unbounded twin-width.
To see that, consider $L$ the line graph of the biclique $K_{n,n}$ (with the edge set of $K_{n,n}$ as vertex set, and edges between every pair of incident edges in $K_{n,n}$).
Every pair of vertices in $L$ has at least $2(n-1)$ private neighbors (at least $n-1$ private neighbors for each vertex), hence its twin-width is at least $2(n-1)$.

\subsection{Bounded rank-width/clique-width, and $d$-dimensional grids} 

We now show that bounded rank-width graphs and $d$-dimensional grids (with or without diagonals) have bounded twin-width.
We transfer the twin-width boundedness of $d$-dimensional grids with diagonals to unit $d$-dimensional ball graphs with bounded clique number.

A natural inquiry is to compare twin-width with the width measures designed for dense graphs: rank-width $\text{rw}$, clique-width $\text{cw}$, module-width $\text{modw}$, and boolean-width $\text{boolw}$.
It is known that, for any graph $G$, $\text{boolw}(G) \leqslant \text{modw}(G) \leqslant \text{cw}(G) \leqslant 2^{\text{rw}(G)+1}-1$ (see for instance Chapter 4 of Vatshelle's PhD thesis~\cite{vatshelleThesis}).
It is thus sufficient to show that graphs with bounded boolean-width have bounded twin-width, to establish that bounded twin-width classes capture all these parameters.

Crucially twin-width does not capture bounded mim-width graphs (the actual definition of mim-width is not important here, and thus omitted).
This is but a fortunate fact, since the main result of the paper is an FPT algorithm for FO model checking on any bounded twin-width classes. 
Indeed, interval graphs have mim-width 1~\cite{Belmonte13} and do not admit an FPT algorithm for FO model checking (see for instance \cite{Ganian15}).

We briefly recall the definition of boolean-width.
The \emph{boolean-width} of a partition $(A,B)$ of the vertex set of a graph is the base-2 logarithm of the number of different neighborhoods in $B$ of subsets of vertices of $A$ (or equivalently, of different neighborhoods in $A$ of subset of vertices of $B$).
A \emph{decomposition tree of a graph $G$} is a binary tree\footnote{All internal nodes have degree 3, except the root which has degree 2. Equivalently all internal nodes have exactly two children.} $T$ whose leaves are in one-to-one correspondence with $V(G)$.
Each edge $e$ of $T$ naturally maps to a partition $P_e=(A_e,B_e)$ of $V(G)$, where the two connected components of $T-e$ contain the leaves labeled by $A_e$ and $B_e$, respectively.  
The \emph{boolean-width} of a decomposition tree $T$ is the maximum boolean-width of $P_e$ taken among every edge $e$ of $T$.
Finally, the \emph{boolean-width} of a graph $G$, denoted by $\text{boolw}(G)$, is the minimum boolean-width of $T$ taken among every decomposition tree $T$.

\begin{theorem}\label{thm:boolean-width}
Every graph with boolean-width $k$ has twin-width at most~$2^{k+1}-1$.
\end{theorem}

\begin{proof}
  Let $G$ be graph and let $T$ be a decomposition tree of $G$ with boolean-width $k := \text{boolw}(G)$.
  We assume that $G$ has at least $2^k+1$ vertices, otherwise the twin-width is immediately bounded by $2^k$.
  Starting from the root $r$ of $T$, we find a rooted subtree of $T$ with at least $2^k+1$ and at most $2^{k+1}$ leaves.
  If the current subtree has more than $2^{k+1}$ leaves, we move to the child node with the larger subtree.
  That way we guarantee that the new subtree has at least $2^k+1$ leaves.
  We stop when we reach a subtree $T'$ with at most $2^{k+1}$ leaves, and let $e$ be the last edge that we followed in the process of finding $T'$ (the one whose removal disconnects $T'$ from the rest of $T$).
  
  By definition, the boolean-width of the partition $P_e=(A_e,B_e)$ is at most $k$, which upperbounds the number of different neighborhoods of $A_e$ in $B_e$ by $2^k$. 
  In particular, among the $2^k+1$ leaves of $T'$, corresponding to, say, $A_e$, two vertices $u, v$ have the same neighborhood in $B_e$.
  We contract $u$ and $v$ in $G$ (and obtain the graph $G/u,v$).
  The only red edges in $G/u,v$ are within $A_e$, so the red degree is bounded by~$2^{k+1}-1$. 
  We update $T$ by removing the leaf labeled by $v$, and smoothing its parent node which became a degree-2 vertex (to keep a binary tree).
  We denote by $T/u,v$ the obtained binary decomposition tree of $G/u,v$.
  
  What we described so far yielded the first contraction.
  We start over with trigraph $G/u,v$ and decomposition tree $T/u,v$ to find the second contraction.
  We iterate this process until the current trigraph is a singleton.
  We claim that the built sequence of contractions only contains trigraphs with red degree at most~$2^{k+1}-1$.
  The crucial invariant is that our contractions never create a red component of size more than $2^{k+1}$.
  Hence the red degree remains bounded by~$2^{k+1}-1$.
\end{proof}

The \emph{$d$-dimensional $n$-grid} is the graph with vertex set $[n]^d$ with an edge between two vertices $(x_1, \ldots, x_d)$ and $(y_1, \ldots, y_d)$ if and only if $\sum_{i=1}^d |x_i - y_i| = 1$.
Equivalently the $d$-dimensional $n$-grid is the Cartesian product of $d$ paths on $n$ vertices, hence we write it $P_n^d$.
Thus the 1-dimensional $n$-grid is the path on $n$ vertices $P_n$, while the $2$-dimensional $n$-grid is the usual (planar) $n \times n$-grid.
While all the width parameters presented so far (including mim-width) are unbounded on the $n \times n$-grid, twin-width remains constant even on the $d$-dimensional $n$-grid, for any fixed $d$. 

\begin{theorem}\label{thm:grids}
  For every positive integers $d$ and $n$, the $d$-dimensional $n$-grid has twin-width at most~$3d$.
\end{theorem}

\begin{proof}
  Let $R_n^d$ the trigraph with vertex set $V(P_n^d)$, red edges $E(P_n^d)$, and no black edge.
  We will prove, by induction on $d$, that $R_n^d$ has twin-width at most $3d$.
  The base case ($d=1$) holds since, as observed in \cref{sec:def}, the twin-width of a red path is at most~$2$.
  As all the edges will be red (no black edge can appear), we allow ourselves the following abuse of language. 
  For this proof only, by \emph{edge} (resp.~\emph{degree}) we mean \emph{red edge} (resp. \emph{red degree}).
  We now assume that $d > 1$.
  
  We see $R_n^d$ as the Cartesian product of $R_n^{d-1}$ and $R_n^1=R_n$.
  In other words, $V(R_n^d)$ can be partitioned into $n$ sets $V_1, \ldots, V_n$, where each $V_i = \{v_1^i, \dots v_{n^{d-1}}^i\}$ induces a trigraph isomorphic to $R_n^{d-1}$, and there is an edge between $v_j^i$ and $v_j^{i+1}$ for all $i \in [n-1]$, $j \in [n^{d-1}]$.
  By induction hypothesis, there is a sequence of $3(d-1)$-contractions of $P_n^{d-1}$.
  The idea is to follow this sequence in each $V_i$ ``in parallel'', i.e., performing the first contraction in $V_1$, then in $V_2$, up to $V_n$, then the second contraction in $V_1$, then in $V_2$, up to $V_n$, and so on.
  By doing so, the following invariants are maintained:
 \begin{itemize}
 	\item when performing a contraction in $V_1$, the newly created vertex has degree at most $3d-3$ in $V_1$, and~$2$ in $V_2$ (and $0$ elsewhere), so $3d-1$ in total.
 	\item when performing a contraction in $V_i$, $i \in \{2, \dots, n-1\}$, the created vertex has degree at most $3d-3$ in $V_i$, $1$ in $V_{i-1}$ (since the same pair has been contracted in $V_{i-1}$ at the previous step) and~$2$ in $V_{i+1}$ (and $0$ elsewhere), so $3d$ in total.
 	\item when performing a contraction in $V_n$, the created vertex has degree at most $3d-3$ in $V_n$, and at most one in $V_{n-1}$ (and $0$ elsewhere), so $3d-2$ in total.
 \end{itemize}

 Furthermore every vertex not involved in the current contraction has degree at most~$3d-2$: Its degree within its own $V_i$ is $3d-3$ (by induction hypothesis) and it has exactly one neighbor in $V_{i-1}$ (if this set exists) and exactly one neighbor in $V_{i+1}$ (if this set exists).
 When this process terminates, each $V_i$ has been contracted into a single vertex.
 Hence the current trigraph is the red path $R_n$, which admits a sequence of 2-contractions.  
\end{proof}
As we even showed that the twin-width of the red graph $R_n^d$ is at most $3d$, it implies that the twin-width of any subgraph of the $d$-dimensional $n$-grid is bounded by $3d$.

The \emph{$d$-dimensional $n$-grid with diagonals} is the graph on $[n]^d$ with an edge between two distinct vertices $(x_1, \ldots, x_d)$ and $(y_1, \ldots, y_d)$ if and only if $\max_{i=1}^d |x_i - y_i| \leqslant 1$.
We denote this graph by $\mathcal K_{n,d}$ and by, $\mathcal K^r_{n,d}$ the trigraph $([n]^d,\emptyset,E(\mathcal K_{n,d}))$ with only red edges.
By the arguments of \cref{thm:grids}, one can see that every subgraph of $\mathcal K_{n,d}$ (even of $\mathcal K^r_{n,d}$) has twin-width bounded by a function of~$d$ (observe that $K^r_{n,d}$ has red degree at most $3^d$).
\begin{lemma}\label{lem:kings}
  Every subgraph of $\mathcal K^r_{n,d}$ has twin-width at most $2(3^d - 1)$.
\end{lemma}

This fact permits to bound the twin-width of unit $d$-dimensional ball graphs with bounded clique number; actually even their subgraphs.
\begin{theorem}\label{thm:unitBall}
  Every subgraph $H$ of a unit $d$-dimensional ball graph $G$ with clique number~$k$ has twin-width at most~$d' := (3 \lceil \sqrt{d} \rceil)^d k$.
  Furthermore if $G$ comes with a geometric representation (i.e., coordinates for each vertex of $G$ in a possible model), then a $d'$-contraction sequence of $H$ can be found in polynomial time.
\end{theorem}
\begin{proof}
  The result is immediate for $k=1$, so we assume that $k \geqslant 2$.
  We even show the result when all the edges of $H$ are in fact red edges, by exhibiting a sequence of contractions which keeps the (red) degree below $d'$.
  We draw a geometric regular $d$-dimensional fine grid on top of the geometric representation of $G$.
  The spacing of the grid is $2/\sqrt{d}$ so that a largest diagonal of each hypercubic cell has length exactly~2.
  Hence the unit balls centered within a given cell form a clique.
  In particular, each cell contains at most $k$ centers.
  We also consider the coarser tesselation where a \emph{supercell} is a hypercube made of $\lceil \sqrt{d} \rceil^d$ (smaller) cells.
  Hence a supercell contains at most $\lceil \sqrt{d} \rceil^d k$ centers.
  
  We contract the vertices of each supercell into a single vertex.
  This can be done in any order of the supercells, and in any order of the vertices within each supercell.
  Observe that, throughout this process, the (red) degree does not exceed $(3 \lceil \sqrt{d} \rceil)^d k$.

  After these $d'$-contractions, the graph that we obtain is a subgraph of $\mathcal K^r_{n,d}$.
  Hence it admits a $2(3^d - 1)$-sequence by \cref{lem:kings}.
  We conclude since $2(3^d - 1) \leqslant (3 \lceil \sqrt{d} \rceil)^d k$. 
\end{proof}
Of course the constructive result of \cref{thm:unitBall} can be proved in greater generality.
It would work with any collection of objects where the ratio between the smallest (taken over the objects) radius of a largest enclosed ball and the largest radius of a smallest enclosing ball is bounded, as well as the clique number.
In \cite{twin-width2} we will see that unit disk graphs (with no restriction on the clique number), as well as interval graphs and $K_t$-free unit segment graphs, have unbounded twin-width. 

%

%

\section{The grid theorem for twin-width}\label{sec:grid-theorem}

In this section, we will deal with matrices instead of graphs.
Our matrices have their entries on a finite alphabet with a special additional value $r$ (for red) representing errors made along the computations.
This is the analog of the red edges of the previous section.

\subsection{Twin-width of matrices, digraphs, and binary structures}\label{sec:dig-encoding}

The \emph{red number} of a matrix is the maximum number of red entries taken over all rows and all columns.
Given an ${n \times m}$ matrix $M$ and two columns $C_i$ and $C_j$, the \emph{contraction} of $C_i$ and $C_j$ is obtained by deleting $C_j$ and replacing every entry $m_{k,i}$ of $C_i$ by $r$ whenever $m_{k,i} \neq m_{k,j}$.
The same contraction operation is defined for rows.
A matrix $M$ has \emph{twin-width} at most~$k$ if one can perform a sequence of contractions starting from $M$ and ending in some ${1 \times 1}$ matrix in such a way that all matrices occurring in the process have red number at most~$k$.
Note that when $M$ has twin-width at most~$k$, one can reorder its rows and columns in such a way that every contraction will identify consecutive rows or columns.
The reordered matrix is then called \emph{$k$-twin-ordered}.
The \emph{symmetric twin-width} of an $n \times n$ matrix $M$ is defined similarly, except that the contraction of rows $i$ and $j$ (resp.~columns $i$ and $j$) is immediately followed by the contraction of columns $i$ and $j$ (resp.~rows $i$ and $j$), and the new red number is only computed after the two contractions are performed.

We can now extend the twin-width to digraphs, which in particular capture posets.
Unsurprisingly the twin-width of a digraph is defined as the symmetric twin-width of its adjacency matrix; only we write the adjacency matrix in a specific way.
Say, the vertices are labeled $v_1, \ldots, v_n$.
If there is an arc $v_iv_j$ (but no arc $v_jv_i$), we place a 1 entry in the $i$-th row $j$-column of the matrix and a -1 entry in the $j$-th row $i$-th column.
If there are two arcs $v_iv_j$ and $v_jv_i$, we place a 2 entry in both the $i$-th row $j$-column and $j$-th row $i$-th column.
If there is no arc $v_iv_j$ nor $v_jv_i$, we place a 0 entry in both the $i$-th row $j$-column and $j$-th row $i$-th column.
We then further extend twin-width to a binary structure $S$ with binary relations $E^1, \ldots, E^h$.
When building the adjacency matrix, the entry at $v_i,v_j$ is now $(e_1, \ldots, e_h)$ where $e_p \in \{-1,0,1,2\}$ is chosen accordingly to the encoding of the ``digraph $E^p$''.
Again the twin-width of a binary structure is the symmetric twin-width of the so-built adjacency matrix.

We call \emph{augmented binary structure} a binary structure augmented by a constant number of unary relations.
The twin-width is extended to \emph{augmented binary structures} by seeing unary relations as hard constraints.
More concretely, contractions between two vertices $u$ and $v$ are only allowed if they are in the exact same unary relations.
Formally, in a binary structure $G$ augmented by unary relations $U_1, \ldots, U_h$, the contraction of $u$ and $v$ is only possible when for every $j \in [h]$, $G \models U_j(u) \Leftrightarrow G \models U_j(v)$.
When this happens, the contracted vertex $z$ inherits the unary relations containing $u$ (or equivalently $v$).

Contrary to the contraction sequence of a binary structure (without unary relations), we cannot expect the contraction sequence to end on a single vertex.
Instead a sequence now ends when no pair of vertices are included in the same unary relations.
When this eventually happens, the number of vertices is nevertheless bounded by the constant $2^h$.
We could continue the contraction sequence arbitrarily, but, anticipating our use of augmented binary structures in \cref{sec:fo-inter}, it is preferable to stop the sequence there.

By a straightforward generalization of the proof of~\cref{thm:apex}, one can see that adding $h$~unary relations can at most multiply the twin-width by $2^h$.
\begin{lemma}\label{lem:unary}
  The twin-width of a binary structure $G$ augmented by $h$ unary relations is at most $2^h \cdot \text{tww}(G)$.
\end{lemma}

Given a total order $\sigma$ on the domain of a binary structure $G$, we denote by $A_\sigma(G)$ the adjacency matrix encoded accordingly to the previous paragraph and following the order~$\sigma$.
Denoting $M := A_\sigma(G) = (m_{ij} = (e^{ij}_1, \ldots, e^{ij}_h))_{i,j}$, the matrix $M$ satisfies the important following property, mixing symmetry and skew-symmetry.
If $e^{ij}_p \in \{0,2\}$ then $e^{ij}_p = e^{ji}_p$, and if $e^{ij}_p \in \{-1,1\}$ then $e^{ij}_p = - e^{ji}_p$.
We call this property \emph{mixed-symmetry} and $M$ is said \emph{mixed-symmetric}.
This will be useful to find \emph{symmetric} sequences of contractions.

\subsection{Partition coarsening, contraction sequence, and error value}

Here we present an equivalent way of seeing the twin-width with a successive coarsening of a partition, instead of explicitly performing the contractions with deletion.

A partition $\mathcal P$ of a set $S$ \emph{refines} a partition ${\mathcal P}'$ of $S$ if every part of ${\mathcal P}$ is contained in a part of ${\mathcal P}'$.
Conversely we say that $\mathcal P'$ is a coarsening of $\mathcal P$, or \emph{contains} $\mathcal P$.
When every part of ${\mathcal P}'$ contains at most $k$ parts of ${\mathcal P}$, we say that $\mathcal P$ \emph{$k$-refines} ${\mathcal P}'$.
Given a partition ${\mathcal P}$ and two distinct parts $P, P'$ of ${\mathcal P}$, the \emph{contraction} of $P$ and $P'$ yields the partition ${\mathcal P} \setminus \{P,P'\} \cup \{P\cup P'\}$.

Given an ${n\times m}$ matrix $M$, a \emph{row-partition} (resp.~\emph{column-partition}) is a partition of the rows (resp.~columns) of~$M$.
A \emph{$(k,\ell)$-partition} (or simply \emph{partition}) of a matrix $M$ is a pair $({\mathcal R}=\{R_1,\dots ,R_k\},$ $ {\mathcal C}=\{C_1,\dots ,C_\ell\})$ where $\mathcal R$ is a row-partition and $\mathcal C$ is a column-partition.
A \emph{contraction} of a partition $({\mathcal R},{\mathcal C})$ of a matrix $M$ is obtained by performing one contraction in ${\mathcal R}$ or in ${\mathcal C}$. 

We distinguish two extreme partitions of an $n \times m$ matrix $M$: the \emph{finest partition} where $({\mathcal R},{\mathcal C})$ have size $n$ and $m$, respectively, and the \emph{coarsest partition} where they both have size one.
The finest partition is sometimes called the \emph{partition of singletons}, since all its parts are singletons, and the coarsest partition is sometimes called the \emph{trivial partition}.
A \emph{contraction sequence} of an $n \times m$ matrix $M$ is a sequence of partitions $({\mathcal R}^1,{\mathcal C}^1),\dots ,({\mathcal R}^{n+m-1},{\mathcal C}^{n+m-1})$  where
\begin{itemize}
\item $({\mathcal R}^1,{\mathcal C}^1)$ is the finest partition,
\item $({\mathcal R}^{n+m-1},{\mathcal C}^{n+m-1})$ is the coarsest partition, and
\item for every $i \in [n+m-2]$, $({\mathcal R}^{i+1},{\mathcal C}^{i+1})$ is a contraction of $({\mathcal R}^i,{\mathcal C}^i)$.
\end{itemize}

Given a subset $R$ of rows and a subset $C$ of columns in a matrix $M$, the \emph{zone $R \cap C$} denotes the submatrix of all entries of $M$ at the intersection between a row of $R$ and a column of $C$.
A \emph{zone} of a partition pair $({\mathcal R},{\mathcal C})=(\{R_1, \ldots, R_k\},\{C_1, \ldots, C_\ell\})$ is any $R_i \cap C_j$ for $i \in [k]$ and $j \in [\ell]$.
A zone is \emph{constant} if all its entries are identical.
The \emph{error value} of $R_i$ is the number of non constant zones among all zones in $\{R_i \cap C_1, \ldots, R_i\cap C_\ell\}$.
We adopt a similar definition for the \emph{error value} of $C_j$.
The \emph{error value} of $({\mathcal R},{\mathcal C})$ is the maximum error value taken over all $R_i$ and $C_j$. 

We can now restate the definition of twin-width of a matrix $M$ as the minimum $t$ for which there exists a contraction sequence of $M$ consisting of partitions with error value at most $t$. The following easy technical lemma will be used later to upper bound twin-width.

\begin{lemma}\label{lem:subsequence}
If $({\mathcal R}^1,{\mathcal C}^1),\dots ,({\mathcal R}^s,{\mathcal C}^s)$ is a sequence of partitions of a matrix $M$ such that:
\begin{compactitem}
\item $({\mathcal R}^1,{\mathcal C}^1)$ is the finest partition,
\item $({\mathcal R}^s,{\mathcal C}^s)$ is the coarsest partition,
\item ${\mathcal R}^{i}$ $r$-refines ${\mathcal R}^{i+1}$ and ${\mathcal C}^{i}$ $r$-refines ${\mathcal C}^{i+1}$, and
\item all $({\mathcal R}^{i},{\mathcal C}^{i})$ have error value at most $t$,
\end{compactitem}
then the twin-width of $M$ is at most $rt$.
\end{lemma}

\begin{proof}
  We extend the sequence $({\mathcal R}^{i},{\mathcal C}^{i})$ into a contraction sequence by performing in any order the contractions to go from every pair $({\mathcal R}^{i},{\mathcal C}^{i})$ to the next pair $({\mathcal R}^{i+1},{\mathcal C}^{i+1})$.
  A worst-case argument gives that the error value cannot exceed $rt$. 
Indeed, assume that a partition $({\mathcal R},{\mathcal C})$ contains $({\mathcal R}^{i},{\mathcal C}^{i})$ and refines $({\mathcal R}^{i+1},{\mathcal C}^{i+1})$ and that $R$ is a part of ${\mathcal R}$. 
Every part of ${\mathcal C}$ is contained in a part of ${\mathcal C}^{i+1}$ and every part of ${\mathcal C}^{i+1}$ contains at most $r$ parts of ${\mathcal C}$. 
Moreover, at most $t$ parts of ${\mathcal C}^{i+1}$ form non-constant zones with $R$.
Therefore, at most $rt$ parts of ${\mathcal C}$ form non-constant zones with $R$.
\end{proof}

\subsection{Matrix division and the Marcus-Tardos theorem}

In a contraction sequence of a matrix $M$, one can always reorder the rows and the columns of $M$ in such a way that all parts of all partitions in the contraction sequence consist of consecutive rows or consecutive columns.
To mark this distinction, a \emph{row-division} is a row-partition where every part consists of consecutive rows; with the analogous definition for \emph{column-division}.
A \emph{$(k,\ell)$-division} (or simply \emph{division}) of a matrix $M$ is a pair $({\mathcal R},{\mathcal C})$ of a row-division and a column-division with respectively $k$ and $\ell$ parts.
A \emph{fusion} of a division is obtained by contraction of two consecutive parts of ${\mathcal R}$ or of ${\mathcal C}$.
Fusions are just contractions preserving divisions.
A \emph{division sequence} is a contraction sequence in which all partitions are divisions.

We now turn to the fundamental tool which is basically only applied once but is the cornerstone of twin-width.
Given a $0,1$-matrix $M=(m_{i,j})$, a \emph{$t$-grid minor} in $M$ is a $(t,t)$-division $({\mathcal R},{\mathcal C})$ of $M$ in which every zone contains a 1 (see left of \cref{fig:grid-mixed-minor}).
We say that a matrix is \emph{$t$-grid free} if it does not have a \emph{$t$-grid minor}. 
A celebrated result by Marcus and Tardos~\cite{MarcusT04} (henceforth \emph{the Marcus-Tardos theorem}) asserts that every $0,1$-matrix with large enough linear density has a $t$-grid minor.
Precisely: 

\begin{theorem}[\cite{MarcusT04}]\label{thm:marcustardos}
For every integer $t$, there is some $c_t$ such that every $n\times m$ $0,1$-matrix $M$ with at least $c_t\max(n,m)$ entries 1 has a $t$-grid minor.
\end{theorem}

Marcus and Tardos established this theorem with $c_t = 2t^4{t^2 \choose t}$.
Fox~\cite{Fox13} subsequently improved the bound to $3t2^{8t}$.
He also showed that $c_t$ has to be superpolynomial in $t$ (at least $2^{\Omega(t^{1/4})}$).
Then Cibulka and Kyn\v{c}l~\cite{Cibulka16} decreased $c_t$ further down to $8/3(t+1)^22^{4t}$.

Matrices with enough 1 entries are complex in the sense that they contain large $t$-grids minors.
However here the role of 1 is special compared to 0, and this result is only interesting for sparse matrices.
We would like to extend this notion of complexity to the dense case, that is to say for all matrices.
In the Marcus-Tardos theorem zones are \emph{not simple} if they contain a 1, that is, if they have rank at least 1.
A natural definition would consist of substituting ``rank at least 1'' by ``rank at least 2'' in the definition of a $t$-grid minor.
Since we mostly deal with $0,1$-matrices, and exclusively with discrete objects, we adopt a more combinatorial approach.

\subsection{Mixed minor and the grid theorem for twin-width}

A matrix $M=(m_{i,j})$ is \emph{vertical} (resp. \emph{horizontal}) if $m_{i,j}=m_{i+1,j}$ (resp. $m_{i,j}=m_{i,j+1}$) for all $i,j$.
Observe that a matrix which is both vertical and horizontal is constant.
We say that $M$ is \emph{mixed} if it is neither vertical nor horizontal.
A \emph{$t$-mixed minor}  in $M$ is a division $({\mathcal R},{\mathcal C})=(\{R_1,\dots ,R_t\},\{C_1,\dots ,C_t\})$ such that every zone $R_i\cap C_j$ is mixed (see right of \cref{fig:grid-mixed-minor}).
A matrix without $t$-mixed minor is \emph{$t$-mixed free}.
For instance, the $n \times n$ matrix with all entries equal to 1 is $1$-mixed free but admits an $n$-grid minor.

\begin{figure}
  \centering
  \begin{tikzpicture}
    \def\s{0.5}
    \def\hb{\s/2}
    \def\vb{\s/2}
    \def\he{8.5 * \s}
    \def\ve{7.5 * \s}
    \foreach \i/\j/\v in {1/1/1,1/2/0,1/3/1,1/4/0,1/5/0,1/6/0,1/7/1, 2/1/0,2/2/1,2/3/0,2/4/1,2/5/0,2/6/1,2/7/1, 3/1/1,3/2/1,3/3/0,3/4/0,3/5/0,3/6/1,3/7/1, 4/1/1,4/2/1,4/3/1,4/4/0,4/5/0,4/6/0,4/7/1, 5/1/1,5/2/1,5/3/1,5/4/1,5/5/0,5/6/0,5/7/1, 6/1/0,6/2/1,6/3/0,6/4/0,6/5/0,6/6/1,6/7/1, 7/1/0,7/2/0,7/3/1,7/4/1,7/5/0,7/6/0,7/7/1, 8/1/1,8/2/0,8/3/0,8/4/0,8/5/1,8/6/1,8/7/0}{
      \node (e\i\j) at (\s * \i,\s * \j) {$\v$} ;
    }
    \draw (\hb-0.05,\vb) -- (\hb-0.05,\ve) --++(0.1,0) ;
    \draw (\hb-0.05,\vb) --++(0.1,0) ;
    \draw (\he+0.05,\vb) -- (\he+0.05,\ve) --++(-0.1,0) ;
    \draw (\he+0.05,\vb) --++(-0.1,0) ;
    \foreach \i in {2,4,6}{
      \draw[very thick] (\i * \s + \s/2,\vb) -- (\i * \s + \s/2,\ve) ;
    }
    \foreach \j in {2,4,6}{
      \draw[very thick] (\hb,\j * \s + \s/2) -- (\he,\j * \s + \s/2) ;
    }

    \begin{scope}[xshift=12 * \s cm]
    \foreach \i/\j/\v in {1/1/1,1/2/0,1/3/1,1/4/0,1/5/0,1/6/0,1/7/1, 2/1/0,2/2/1,2/3/0,2/4/1,2/5/0,2/6/1,2/7/1, 3/1/1,3/2/1,3/3/0,3/4/0,3/5/0,3/6/1,3/7/1, 4/1/1,4/2/1,4/3/1,4/4/0,4/5/0,4/6/0,4/7/1, 5/1/1,5/2/1,5/3/1,5/4/1,5/5/0,5/6/0,5/7/1, 6/1/0,6/2/1,6/3/0,6/4/0,6/5/0,6/6/1,6/7/1, 7/1/0,7/2/0,7/3/1,7/4/1,7/5/0,7/6/0,7/7/1, 8/1/1,8/2/0,8/3/0,8/4/0,8/5/1,8/6/1,8/7/0}{
      \node (e\i\j) at (\s * \i,\s * \j) {$\v$} ;
    }
    \draw (\hb-0.05,\vb) -- (\hb-0.05,\ve) --++(0.1,0) ;
    \draw (\hb-0.05,\vb) --++(0.1,0) ;
    \draw (\he+0.05,\vb) -- (\he+0.05,\ve) --++(-0.1,0) ;
    \draw (\he+0.05,\vb) --++(-0.1,0) ;
    \foreach \i in {2,5}{
      \draw[very thick] (\i * \s + \s/2,\vb) -- (\i * \s + \s/2,\ve) ;
    }
    \foreach \j in {3,5}{
      \draw[very thick] (\hb,\j * \s + \s/2) -- (\he,\j * \s + \s/2) ;
    }
    \end{scope}
  \end{tikzpicture}
  \caption{To the left a $4$-grid minor: every zone contains at least one 1.
    To the right a $3$-mixed minor on the same matrix: no zone is horizontal or vertical.}
  \label{fig:grid-mixed-minor}
\end{figure}

The main result of this section is that $t$-mixed free matrices are exactly matrices with bounded twin-width, modulo reordering the rows and columns.
More precisely:
\begin{theorem}[grid minor theorem for twin-width]\label{thm:gridtheorem}
 Let $\alpha$ be the alphabet size for the matrix entries, and $c_t := 8/3(t+1)^22^{4t}$.
\begin{compactitem}
\item Every $t$-twin-ordered matrix is $2t+2$-mixed free.
\item Every $t$-mixed free matrix has twin-width at most $4c_t \alpha^{4c_t+2}=2^{2^{O(t)}}$.
\end{compactitem}
\end{theorem}

A contraction sequence is a fairly complicated object.
It can be seen as a sequence of coarser and coarser partitions of the vertices, or as a sequence of pairs of vertices.
The second bullet of \cref{thm:gridtheorem} simplifies the task of bounding the twin-width of a graph.
One only needs to find an ordering of the vertex set such that the adjacency matrix written down with that order has no $t$-mixed minor.
A typical use to bound the twin-width of a class $\mathcal C$: \\
(1) find a good vertex-ordering process based on properties of $\mathcal C$, \\
(2) assume that the adjacency matrix in this order has a $t$-mixed minor, \\
(3) use this $t$-mixed minor to derive a contradiction to the membership to $\mathcal C$, and \\
(4) conclude with \cref{thm:gridtheorem}.\\
\cref{sec:bounded-twinwidth} presents more and more elaborate instances of this framework and \cref{tbl:orders} reports the orders and the bounds for different classes.

A sanity check of \cref{thm:gridtheorem} is given by random 0,1-matrices.
They have large grid minors for any reordering of the rows and columns, and indeed, random bipartite graphs have unbounded twin-width.


\subsection{Corners}

The proof of \cref{thm:gridtheorem} will crucially rely on the notion of \emph{corner}.
Given a matrix $M=(m_{i,j})$, a \emph{corner} is any 2-by-2 mixed submatrix of the form $(m_{i,j},m_{i+1,j},m_{i,j+1},$ $m_{i+1,j+1})$.
Corners will play the same role as the 1 entries in the Marcus-Tardos theorem, as they localize the property of being mixed:

\begin{lemma}\label{lem:corner}
A matrix is mixed if and only if it contains a corner.
\end{lemma}

\begin{proof}
  A corner is certainly a witness of being mixed.
  Conversely let us assume that a matrix $M$ has no corner.
  Either $M$ is constant and we are done: $M$ is not mixed.
  Or, without loss of generality, there are in $M$ two distinct entries $m_{i,j} \neq m_{i+1,j}$. 
  To avoid a corner, both entries $m_{i,j+1}$ and $m_{i,j-1}$ are equal to $m_{i,j}$.
  Similarly, both entries $m_{i+1,j+1}$ and $m_{i+1,j-1}$ are equal to $m_{i+1,j}$.
  Therefore the whole $i$-th row is constant as well as the $i+1$-st row.
  This forces the rows of index $i-1$ and $i+2$ to be constant, and propagates to the whole matrix which is then horizontal.
  Observe that if the two distinct adjacent entries would initially be $m_{i,j} \neq m_{i,j+1}$, then the same arguments would show that the matrix is vertical.
\end{proof}


\subsection{Mixed zones, cuts, and values}

Let ${\mathcal R}=\{R_1, \dots, R_k\}$ be a row-division of a matrix $M$ and let $C$ be a set of consecutive columns.
We call \emph{mixed zone} of $C$ on ${\mathcal R}$ any zone $R_i \cap C$ which is a mixed matrix.
We call \emph{mixed cut} of $C$ on ${\mathcal R}$ any index $i \in [k-1]$ for which the 2-by-$|C|$ zone defined by the last row of $R_i$, the first row of $R_{i+1}$, and $C$ is a mixed matrix.
Now the \emph{mixed value} of $C$ on ${\mathcal R}$ is the sum of the number of mixed cuts and the number of mixed zones.
See \cref{fig:mixed-value} for an illustration, and for why we use the mixed value instead of the mere number of mixed zones.
Analogously we define the mixed value of a set $R$ of consecutive rows on a column-division~$\mathcal C$.

\begin{lemma}\label{lem:mixedvalue}
The contraction of two consecutive parts of ${\mathcal R}$ does not increase the mixed value of $C$ on ${\mathcal R}$.
\end{lemma}

\begin{proof}
  Assume that ${\mathcal R}=\{R_1,\dots ,R_k\}$ and ${\mathcal R}'$ is obtained by contraction of $R_i$ and $R_{i+1}$.
  We just have to show that if $R_i\cap C$, $R_{i+1}\cap C$ are not mixed zones and $i$ is not a mixed cut, then $(R_i \cup R_{i+1}) \cap C$ is not a mixed zone.
  Indeed, if $(R_i\cup R_{i+1})\cap C$ is a mixed zone, it contains a corner which must be in $R_i\cap C$, or in $R_{i+1} \cap C$, or otherwise sits in the mixed cut $i$.
\end{proof}

\begin{figure}
  \centering
  \begin{tikzpicture}
    \def\s{0.5}
    \def\hb{\s/2}
    \def\vb{\s/2}
    \def\he{8.5 * \s}
    \def\ve{7.5 * \s}
    \foreach \i/\j/\v in {1/1/1,1/2/0,1/3/1,1/4/0,1/5/1,1/6/1,1/7/1, 2/1/0,2/2/1,2/3/1,2/4/1,2/5/0,2/6/0,2/7/1, 3/1/\textcolor{red}{1},3/2/\textcolor{red}{1},3/3/0,3/4/0,3/5/1,3/6/1,3/7/1, 4/1/\textcolor{red}{0},4/2/\textcolor{red}{1},4/3/0,4/4/0,4/5/0,4/6/0,4/7/0, 5/1/\textcolor{red}{1},5/2/\textcolor{red}{0},5/3/1,5/4/1,5/5/0,5/6/0,5/7/0, 6/1/0,6/2/1,6/3/0,6/4/0,6/5/0,6/6/1,6/7/1, 7/1/0,7/2/0,7/3/1,7/4/1,7/5/0,7/6/0,7/7/1, 8/1/1,8/2/0,8/3/0,8/4/0,8/5/1,8/6/1,8/7/0}{
      \node (e\i\j) at (\s * \i,\s * \j) {$\v$} ;
    }
    \draw (\hb-0.05,\vb) -- (\hb-0.05,\ve) --++(0.1,0) ;
    \draw (\hb-0.05,\vb) --++(0.1,0) ;
    \draw (\he+0.05,\vb) -- (\he+0.05,\ve) --++(-0.1,0) ;
    \draw (\he+0.05,\vb) --++(-0.1,0) ;
    \foreach \i in {2,5,6}{
      \draw[very thick] (\i * \s + \s/2,\vb) -- (\i * \s + \s/2,\ve) ;
    }
    \foreach \j in {2,4,6}{
      \draw[very thick] (\hb,\j * \s + \s/2) -- (\he,\j * \s + \s/2) ;
    }
    \foreach \j in {2,4}{
      \draw[red,very thick] (2.5 * \s,\j * \s + \s/2) -- (5.5 * \s,\j * \s + \s/2) ;
    }
    \foreach \i/\j in {1.5/1,3.5/2,5.5/3,7/4}{
      \node at (-0.1,\i * \s) {$R_\j$} ;
    }
    \foreach \i/\j in {1.5/1,4/2,6/3,7.5/4}{
      \node at (\i * \s,0) {$C_\j$} ;
    }
    \draw[red,dashed] (2.7 * \s, 0.6 * \s) -- (4.3 * \s, 0.6 * \s) -- (4.3 * \s, 2.4 * \s) -- (2.7 * \s, 2.4 * \s) -- cycle ;
    \draw[red,dashed] (3.7 * \s, 1.6 * \s) -- (5.3 * \s, 1.6 * \s) -- (5.3 * \s, 3.4 * \s) -- (3.7 * \s, 3.4 * \s) -- cycle ;
    \draw[red,dashed] (2.7 * \s, 3.6 * \s) -- (4.3 * \s, 3.6 * \s) -- (4.3 * \s, 5.4 * \s) -- (2.7 * \s, 5.4 * \s) -- cycle ;

    \begin{scope}[xshift=6.5cm]
      \foreach \i/\j/\v in {1/1/1,1/2/0,1/3/1,1/4/0,1/5/1,1/6/1,1/7/1, 2/1/0,2/2/1,2/3/1,2/4/1,2/5/0,2/6/0,2/7/1, 3/1/\textcolor{red}{1},3/2/\textcolor{red}{1},3/3/\textcolor{red}{0},3/4/\textcolor{red}{0},3/5/\textcolor{red}{1},3/6/\textcolor{red}{1},3/7/1, 4/1/\textcolor{red}{0},4/2/\textcolor{red}{1},4/3/\textcolor{red}{0},4/4/\textcolor{red}{0},4/5/\textcolor{red}{0},4/6/\textcolor{red}{0},4/7/0, 5/1/\textcolor{red}{1},5/2/\textcolor{red}{0},5/3/\textcolor{red}{1},5/4/\textcolor{red}{1},5/5/\textcolor{red}{0},5/6/\textcolor{red}{0},5/7/0, 6/1/0,6/2/1,6/3/0,6/4/0,6/5/0,6/6/1,6/7/1, 7/1/0,7/2/0,7/3/1,7/4/1,7/5/0,7/6/0,7/7/1, 8/1/1,8/2/0,8/3/0,8/4/0,8/5/1,8/6/1,8/7/0}{
      \node (e\i\j) at (\s * \i,\s * \j) {$\v$} ;
    }
    \draw (\hb-0.05,\vb) -- (\hb-0.05,\ve) --++(0.1,0) ;
    \draw (\hb-0.05,\vb) --++(0.1,0) ;
    \draw (\he+0.05,\vb) -- (\he+0.05,\ve) --++(-0.1,0) ;
    \draw (\he+0.05,\vb) --++(-0.1,0) ;
    \foreach \i in {2,5,6}{
      \draw[very thick] (\i * \s + \s/2,\vb) -- (\i * \s + \s/2,\ve) ;
    }
    \foreach \j in {2,6}{
      \draw[very thick] (\hb,\j * \s + \s/2) -- (\he,\j * \s + \s/2) ;
    }
    \foreach \j in {2}{
      \draw[red,very thick] (2.5 * \s,\j * \s + \s/2) -- (5.5 * \s,\j * \s + \s/2) ;
    }
     \foreach \j in {4}{
       \draw[dotted,thin] (\hb,\j * \s + \s/2) -- (\he,\j * \s + \s/2) ;
    }
    \foreach \i/\j in {1.5/1,7/4}{
      \node at (-0.1,\i * \s) {$R_\j$} ;
    }
    \node at (-0.5,4.5 * \s) {$R_2 \cup R_3$} ;
    \foreach \i/\j in {1.5/1,4/2,6/3,7.5/4}{
      \node at (\i * \s,0) {$C_\j$} ;
    }
    \end{scope}
  \end{tikzpicture}
  \caption{To the left, the mixed value of $C_2$ on $\{R_1,R_2,R_3,R_4\}$ is 3: one mixed zone and two mixed cuts (all three in red, with a corner in each, highlighted by red dashed squares). To the right, the mixed value of $C_2$ on $\{R_1,R_2 \cup R_3,R_4\}$ is still 3. In general, the mixed value of a $C_j \in \mathcal C$ cannot increase after the fusion of $R_i, R_{i+1} \in \mathcal R$ since the only way for a new mixed zone to be created is that a mixed cut disappears, while new mixed cuts cannot be created. On the contrary, the number of mixed zones in $C_2$ can increase as it went from 1 to 2.}
  \label{fig:mixed-value}
\end{figure}

The \emph{mixed value} of a division $({\mathcal R},{\mathcal C})=(\{R_1,\dots ,R_k\},\{C_1, \ldots, C_\ell\})$ is the maximum mixed value of $R_i$ on ${\mathcal C}$, and of $C_j$ on ${\mathcal R}$, taken over all $R_i \in \mathcal R$ and $C_j \in \mathcal C$.
Observe that the finest division has mixed value 0 and the coarsest division has mixed value at most 1. 

\subsection{Finding a division sequence with bounded mixed value}

Leveraging the Marcus-Tardos theorem, we are ready to compute, for any $t$-mixed free matrix, a~division sequence with bounded mixed value.
This division sequence is not necessarily yet a contraction sequence with bounded error value (indeed a non-constant horizontal or vertical zone counts for 0 in the mixed value but for 1 in the error value).
But this division sequence will serve as a crucial frame to find the eventual contraction sequence.

\begin{lemma}\label{lem:fusion}
Every $t$-mixed free matrix $M$ has a division sequence in which all divisions have mixed value at most $2c_t$ (where $c_t$ is the one of \cref{thm:marcustardos}). 
\end{lemma}

\begin{proof}
  We start with the finest division of $M$ and greedily perform fusions as long as we can keep mixed value at most $2c_t$.
  Assume that we have reached a division $({\mathcal R},{\mathcal C})=(\{R_1,\dots ,R_k\},\{C_1, \dots,$ $C_\ell\})$, in which, without loss of generality, $k\geq \ell$.
  Assume also, for the sake of contradiction, that each fusion $R_{2i-1},R_{2i}$ for $i=1, \dots,\lfloor k/2 \rfloor$ leads to a mixed value exceeding $2c_t$.
  By \cref{lem:mixedvalue}, the mixed value of $C_j$ on ${\mathcal R}$ does not increase when performing a row-fusion.
  Thus, if the fusion of $R_{2i-1}$ and $R_{2i}$ is not possible, this is because the mixed value of $R'_i=R_{2i-1}\cup R_{2i}$ on ${\mathcal C}$ is more than $2c_t$.
  Therefore the number of mixed cuts or zones of each $R'_i$ (for $i=1, \dots,\lfloor k/2 \rfloor$) on ${\mathcal C}$ is greater than $2c_t$; hence $R'_i$ contains more than $2c_t$ corners in mixed zones and mixed cuts.
  Now we refine ${\mathcal C}$ in two possible ways: either ${\mathcal C}'=\{C_1\cup C_2, C_3\cup C_4, \dots\}$ or ${\mathcal C}''=\{C_1, C_2\cup C_3, C_4\cup C_5, \dots\}$.
  Observe that each mixed cut of $R'_i$ on $\mathcal C'$ (resp.~$\mathcal C''$) corresponds to a mixed zone of $R'_i$ on $\mathcal C''$ (resp.~$\mathcal C'$).
  Let ${\mathcal R}'=\{R'_1,\dots ,R'_{\lfloor k/2\rfloor}\}$ and consider the two divisions $({\mathcal R'},{\mathcal C'})$ and $({\mathcal R'},{\mathcal C''})$.
  Thus, in total, the zones contained in these two divisions contain at least $\lfloor k/2\rfloor \cdot 2c_t$ corners.
  So one of these subdivisions contains at least $\lfloor k/2 \rfloor c_t$ zones with a corner, hence $\lfloor k/2 \rfloor c_t$ mixed zones.
  By applying the Marcus-Tardos theorem (\cref{thm:marcustardos}) to the smaller auxiliary matrix with a 1 if the zone is mixed and a 0 otherwise, one can find a $t$-mixed minor in~$M$.
\end{proof}

\subsection{Finding a contraction sequence with bounded error value}

We are now equipped to prove the main result of this section, which is the second item of \cref{thm:gridtheorem}.
The division sequence with small mixed value, provided by \cref{lem:fusion}, will guide the construction of a contraction sequence (not necessarily a division sequence) of bounded error value.
This two-layered mechanism is also present in the proof of Guillemot and Marx, albeit in a simpler form since they have it tailored for sparse matrices, and importantly they start from a permutation matrix.

\begin{proof}[Proof of \cref{thm:gridtheorem}]
  We first show that every $t$-twin-ordered matrix $M$ is $2t+2$-mixed free.
  Let $({\mathcal R},{\mathcal C})=(\{R_1, \ldots, R_{2t+2}\},\{C_1, \ldots, C_{2t+2}\})$ be a division of an $n \times m$ matrix $M$ and assume for contradiction that all its zones are mixed.
  Since $M$ is $t$-twin-ordered, there is a division sequence $({\mathcal R}^1,{\mathcal C}^1), \ldots, ({\mathcal R}^{n+m-1},{\mathcal C}^{n+m-1})$ in which all divisions have error value at most $t$.
  Let us consider the first index $s$ such that some $R_i$ is contained in a part of ${\mathcal R}^s$ or some $C_j$ is contained in a part of ${\mathcal C}^s$.
  Assume without loss of generality that $R\in {\mathcal R}^s$ contains $R_i$.
  Since a zone $R_i\cap C_j$ in $M$ is mixed for each $C_j$ in $\mathcal{C}$, 
  it is not vertical, and therefore for each $j\in [2t+2]$ there exists a choice $C'_j$ in ${\mathcal C}^s$ 
  which intersects $C_j$ such that $R\cap C'_j$ is not constant. 
  Observe that we cannot have $C'_j=C'_{j+2}$ since this would mean that $C'_j$ contains $C_{j+1}$, a contradiction to the choice of $s$.
  In particular the error value of $R$ in ${\mathcal C}^s$ is at least $(2t+2)/2 > t$, a contradiction.

  We now show that every $n \times m$ matrix $M$ which does not contain a $t$-mixed minor has twin-width at most $4c_t \alpha^{4c_t+2}$, where $c_t$ is as defined in \cref{thm:marcustardos}, and $\alpha$ is the alphabet size for the entries of $M$.
  By~\cref{lem:fusion}, there exists a division sequence $({\mathcal R}^1,{\mathcal C}^1), \ldots, ({\mathcal R}^{n+m-1},{\mathcal C}^{n+m-1})$ with mixed value at most $t' := 2c_t$. 
  We now refine each division $({\mathcal R}^{s},{\mathcal C}^{s})=(\{R_1, \ldots, R_{a}\},$ $\{C_1, \ldots ,C_{b}\})$, into a partition $({\mathcal R}'^{s},{\mathcal C}'^{s})$ of $M$ (which is not necessarily a division). 
  We consider $R_i \in {\mathcal R}^s$ and we say that a subset $J$ of consecutive indices of $\{1, \dots, b\}$ is \emph{good} if $R_i\cap \cup_{j\in J}C_j$ is not mixed.
Now, observe that if $j \in [b-1]$ is not a mixed cut, and if $R_i \cap C_j$ and $R_i \cap C_{j+1}$ are both non-mixed zones, then $R_i \cap (C_j \cup C_{j+1})$ is a non-mixed zone. Since the mixed value of $R_i$ on $\mathcal C^s$ is at most $t'$, one can find at most $t'+1$ good subsets $J_1, \ldots ,J_r$ covering all the non-mixed zones of $R_i$ (each good subset spans all indices between two mixed zones/cuts). 
  We observe that a zone $Z_c := R_i \cap \cup_{j\in J_c}C_j$ is either vertical or horizontal.
  When $Z_c$ is vertical, all rows of $R_i$ are identical on indices in $J_c$.
  When $Z_c$ is horizontal, there are at most $\alpha$ possible rows of $R_i$ restricted to the  indices in $J_c$ where $\alpha$ is the size of the alphabet.
  In particular, there are at most $\alpha^r \leq \alpha^{t'+1}$ different rows in $R_i$, when we restrict them to $\{1, \ldots, b\} \setminus \{j $ $|$ $R_i \cap C_j~\text{is mixed}\}$.
  We then partition $R_i$ into these different types of rows and proceed in the same way for all parts in ${\mathcal R}^{s}$ and in ${\mathcal C}^{s}$ to obtain a partition $({\mathcal R}'^{s},{\mathcal C}'^{s})$ of $M$.

We show that the error value of $({\mathcal R}'^{s},{\mathcal C}'^{s})$ does not exceed $t' \alpha^{t'+1}$.
Suppose that a zone $R\cap C$ where $R\in {\mathcal R}'^{s}$ and $C\in {\mathcal C}'^{s}$ is not constant.
We denote by $R_i \in {\mathcal R}^{s}$ and $C_j\in {\mathcal C}^{s}$ the parts such that $R \subseteq R_i$ and $C \subseteq C_j$.
Note that the zone $R_i \cap C_j$ must be mixed, since otherwise, it has been divided into constant zones in $({\mathcal R}'^{s},{\mathcal C}'^{s})$.
In particular, the total number of such $C_j$ is at most $t'$.
Since $C_j$ has been partitioned at most $\alpha^{t'+1}$ times, the total number of zones $R \cap C$ is at most $t' \alpha^{t'+1}$.

Let us show that the partition $({\mathcal R}'^{s},{\mathcal C}'^{s})$ refines $({\mathcal R}'^{s+1},{\mathcal C}'^{s+1})$. Take for instance 
$R\in {\mathcal R}'^{s}$ and denote by $R_i\in {\mathcal R}^{s}$ the part such that $R\subseteq R_i$.
Now the rows in $R$ have been selected in $R_i$ as they coincide on all zones $R \cap C$ where $C \in {\mathcal C}'^{s}$ and $R_i \cap C$ is not mixed. 
Since the zones of $({\mathcal R}^{s+1},{\mathcal C}^{s+1})$ contain the zones of $({\mathcal R}^{s},{\mathcal C}^{s})$, the selection at stage $s+1$ is based on potentially less $C_j$ such that $R_i \cup C_j$ is not mixed (in case of a column fusion) or potentially more rows to choose $R$ from (in case of a row fusion with $R_i$). 
In both cases, $R$ has to appear in some part of ${\mathcal R}'^{s+1}$.
We established that $({\mathcal R}'^{s},{\mathcal C}'^{s})$ refines $({\mathcal R}'^{s+1},{\mathcal C}'^{s+1})$.
Moreover, since $({\mathcal R}'^{s},{\mathcal C}'^{s})$ $\alpha^{t'+1}$-refines $({\mathcal R}^{s},{\mathcal C}^{s})$ which in turn 2-refines $({\mathcal R}^{s+1},{\mathcal C}^{s+1})$, we have that $({\mathcal R}'^{s},{\mathcal C}'^{s})$ $2\alpha^{t'+1}$-refines $({\mathcal R}^{s+1},{\mathcal C}^{s+1})$.
As $({\mathcal R}'^{s+1},{\mathcal C}'^{s+1})$ refines $({\mathcal R}^{s+1},{\mathcal C}^{s+1})$, $({\mathcal R}'^{s},{\mathcal C}'^{s})$ $2\alpha^{t'+1}$-refines $({\mathcal R}'^{s+1},{\mathcal C}'^{s+1})$.

Finally, we apply~\cref{lem:subsequence} to the sequence $({\mathcal R}'^{s},{\mathcal C}'^{s})$ and conclude that the twin-width of $M$ is at most $2\alpha^{t'+1} \cdot t' \alpha^{t'+1} = 2t'\alpha^{2(t'+1)} = 4c_t \alpha^{4c_t+2}$.
\end{proof}


The second item of \cref{thm:gridtheorem} has the following consequence, which reduces the task of bounding the twin-width of $G$ and finding a contraction sequence to merely exhibiting a \emph{mixed free order}, that is a domain-ordering $\sigma$ such that the matrix $A_\sigma(G)$ is $t$-mixed free for a bounded $t$.
\begin{theorem}\label{cor:mixed}
  Let $G$ be a (di)graph or even a binary structure.
  If there is an ordering $\sigma: v_1, \ldots, v_n$ of $V(G)$ such that $A_\sigma(G)$ is $k$-mixed free, then $\text{tww}(G)=2^{2^{O(k)}}$.
\end{theorem}
\begin{proof}
  We shall just revisit the proof of \cref{thm:gridtheorem} and check that, starting from a mixed-symmetric matrix $M := A_\sigma(G)$, we can design a symmetric contraction sequence.
  As $M = (m_{ij})_{i,j}$ is mixed-symmetric, it holds that $m_{ij} = m_{i'j'} \Leftrightarrow m_{ji} = m_{j'i'}$.
  In particular the symmetric $Z'$ about the diagonal of an off-diagonal zone $Z$ is mixed if and only if $Z'$ is mixed.
  More precisely, $Z'$ is horizontal (resp.~vertical) if and only if $Z$ is vertical (resp.~horizontal).
  
  The division sequence with bounded mixed value, greedily built in \cref{lem:fusion}, can be then made symmetric.
  Say the first fusion merges the $i$-th and $i+1$-st rows, and let us call $R$ this new row-part.
  We perform the symmetric fusion of the $i$-th and $i+1$-st columns, and denote by $C$ the obtained column-part.
  After that operation, no mixed value among the row-parts has increased.
  In particular the mixed value of $R$ has not increased, and this new mixed value equals the mixed value of $C$.
  Therefore the symmetric fusion was indeed possible.
  We iterate this process and follow the rest of the proof of \cref{lem:fusion} to obtain a symmetric division sequence.

  The refinement of the division sequence into a sequence of partitions of bounded error value, in the second step of the proof of \cref{thm:gridtheorem}, is now symmetric since the division is symmetric and $M$ is mixed-symmetric (so two columns are equal on a set of zones if and only if the symmetric rows are equal on the symmetric set of zones).
  Finally the contraction sequence is provided by \cref{lem:subsequence}.
  In this lemma, we observed that the contractions going from the (symmetric) $(\mathcal R^i, \mathcal C^i)$ to the (symmetric) $(\mathcal R^{i+1}, \mathcal C^{i+1})$ can be done in any order.
  Thus we can perform a symmetric sequence of contractions.
  Overall we constructed a symmetric contraction sequence with error value $2^{2^{O(k)}}$.
  Hence the twin-width of $G$ is bounded by that quantity.
  This can be interpreted as a contraction sequence of the vertices of $G$ (or domain elements) with bounded red degree.
\end{proof}
We observe that the proof of~\cref{cor:mixed} is constructive.
It yields an algorithm which, given a $k$-mixed free $n \times n$ matrix $M$, outputs a $2^{2^{O(k)}}$-sequence of $M$ in $O(n^2)$-time.


\section{Classes with bounded twin-width}\label{sec:bounded-twinwidth} 

In this section we show that some classical classes of graphs and matrices have bounded twin-width.
Let us start with the origin of twin-width, which is the method proposed by Guillemot and Marx~\cite{Guillemot14} to understand permutation matrices avoiding a certain pattern.

\subsection{Pattern-avoiding permutations}
We associate to a permutation $\sigma$ over $[n]$ the $n \times n$ matrix $M_{\sigma}=(m_{ij})_{i,j}$ where $m_{i\sigma(i)} = 1$ and all the other entries are set to 0.
A permutation $\sigma$ is a \emph{pattern} of a permutation $\tau$ if $M_{\sigma}$ is a submatrix of $M_{\tau}$.
A central open question was the design of an algorithm deciding if a pattern $\sigma$ appears in a permutation $\tau$ in time $f(|\sigma|) \cdot |\tau|^{O(1)}$.
The brilliant idea of Guillemot and Marx, reminiscent of treewidth and grid minors, is to observe that permutations avoiding a pattern $\sigma$ can be iteratively decomposed (or collapsed), and that the decomposition gives rise to a dynamic-programming scheme.
This lead them to a linear-time $f(|\sigma|) \cdot |\tau|$ algorithm for permutation pattern recognition.
In \cref{sec:def,sec:grid-theorem} we generalized their decomposition to graphs and arbitrary (dense) matrices, and leveraged the Marcus-Tardos theorem, also in the dense setting.
\cref{sec:grid-theorem} would in principle readily apply here: If a permutation matrix $M_{\tau}$ does not contain a fixed pattern of size $k$, then it is certainly $k$-mixed free since otherwise the $k$-mixed minor would contain any pattern of size $k$.
Hence by \cref{thm:gridtheorem}, $M_{\tau}$ has bounded twin-width.

However, to be able to use our framework and derive that FO model checking is FPT in the class of permutations avoiding a given pattern, we need to transform $M_{\tau}$ into a different matrix.
Namely, we consider the directed graph ${D_{\tau}}$ whose vertex set is the union of two total orders, respectively the natural increasing orders on $\{1,\ldots,n\}$ and on $\{1',\ldots,n'\}$, where we add double arcs between $i$ and $\tau(i)'$.
The adjacency matrix $A(D_{\tau})$ of ${D_{\tau}}$ where the vertices are ordered $1,\dots,n,1',\dots,n'$ (recall the encoding mentioned in \cref{sec:dig-encoding}) consists of four blocks.
Two of them are $M_{\tau}$ and its transpose, and the two others (encoding the total orders) both consist of a lower triangle of 1, including the diagonal, completed by an upper triangle of -1.
If $M_{\tau}$ is $k$-mixed free, the matrix $A(D_{\tau})$ is $2k$-mixed free, and thus has bounded twin-width.
Note also that every first-order formula expressible in the permutation $\tau$ (where we can test equality and $\leq$) is expressible in the structure ${D_{\tau}}$.
In \cref{sec:fo} we will show that FO model checking is FPT for ${D_{\tau}}$, as we can efficiently compute a sequence of $d$-partitions.
Therefore FO model checking is also FPT in the class of permutations avoiding some fixed pattern $\sigma$.

As an illustrating example, let us consider the following artificial problem.
Let $\ell$ be a positive integer, and $\sigma, \sigma'$ be two fixed permutations.
Given an input permutation $\tau$, we ask if $\tau$ contains the pattern $\sigma'$ or every pattern of $\tau$ of size $\ell$ is contained in $\sigma$.
There is an $f(\ell, |\sigma|, |\sigma'|) \cdot |\tau|^2$ algorithm to solve this problem (actually the dependency in $|\tau|$ could be made linear in this particular case).
We first compute an upper bound on the twin-width of the matrix $M_{\tau}$ associated to $\tau$ (as defined previously).
Either $M_{\tau}$ has a $|\sigma'|$-mixed minor (and we can answer positively: $\sigma'$ appears in $\tau$), or $D_\tau$ has bounded twin-width.
One of these two outcomes can be reached in time $O(|\tau|^2)$ by the previous section (even $O(|\tau|)$).
We now assume that $D_\tau$ has bounded twin-width.
Then we observe that the property ``every pattern of $\tau$ of size $\ell$ is contained in $\sigma$'' is expressible by a first-order formula of size $g(\ell,|\sigma|)$.
By \cref{sec:fo} that property can be tested in time $f(\ell, |\sigma|, |\sigma'|) \cdot |\tau|$.

Given a permutation $\tau$, we can form the \emph{permutation graph} $G_\tau$ on vertex set $[n]$ where $ij$ is an edge when $i < j$ and $\tau(i) > \tau(j)$.
Note that $G_{\tau}$ can be first-order interpreted from the digraph $D_\tau$ (defined as above) and the partition of $V(D_{\tau})$ into $\{1, \ldots, n\}$ and $\{1', \ldots, n'\}$.
In \cref{sec:fo-inter} we will show that any FO interpretation of a graph $G$ by a formula $\phi(x,y)$ has twin-width bounded by a function of $\phi$ and $\text{tww}(G)$.
This implies the following:
\begin{lemma}\label{lem:permutationgraph}
FO model checking is FPT on every hereditary proper subclass of permutation graphs.
\end{lemma}

\begin{proof}
  By assumption, there is a permutation graph $G_{\sigma}$ which is not an induced subgraph of any graph $G_{\tau}$ in the class.
  We thus obtain that $D_{\tau}$ has bounded twin-width, as $M_{\tau}$ does not contain the pattern $M_{\sigma}$.
  Therefore $G_{\tau}$ itself has bounded twin-width, and a sequence of contractions can be efficiently found (by following the constructive proof of~\cref{sec:grid-theorem}).
  We conclude by invoking \cref{sec:fo}. 
\end{proof}

A similar argument works for partial orders of (Dushnik-Miller) dimension~2, i.e., intersections of two total orders defined on the same set.
We obtain:
\begin{lemma}\label{lem:dm2graph}
 FO model checking is FPT on every proper subclass of partial orders of dimension~2.
\end{lemma}

\subsection{Posets of bounded width}
The versatility of the grid minor theorem for twin-width is also illustrated with posets.
Let $P=(X,\leq)$ be a poset of \emph{width}~$k$, that is, its maximum antichain has size~$k$.
For $x_i, x_j \in X$, \emph{$x_i < x_j$} denotes that $x_i \leq x_j$ and $x_i \neq x_j$.
We claim that the twin-width of $P$ is bounded by a function of~$k$.
By Dilworth's theorem, $P$ can be partitioned into $k$ total orders (or \emph{chains}) $T_1, \ldots, T_k$.
Now one can enumerate the vertices precisely in this order, say $\sigma$, that is, increasingly with respect to~$T_1$, then increasingly with respect to $T_2$, and so on.
We rename the elements of $X$ so that in the order $\sigma$, they read $x_1, x_2, \ldots, x_n$, with $n := |X|$.
Let us write the adjacency matrix $A=(a_{ij}) := A_\sigma(P)$ of~$P$: $a_{ij}=1$ if $x_i \leq x_j$, $a_{ij}=-1$ if $x_j < x_i$, and $a_{ij}=0$ otherwise.
Recall that this is consistent with how we defined the adjacency matrix for the more general digraphs in \cref{sec:grid-theorem}.
We assume for contradiction that $A$ has a $3k$-mixed minor.

By the pigeon-hole principle, there is a submatrix of $A$ indexed by two chains, $T_i$ for the row indices and $T_j$ for the column indices, which has a $3$-mixed minor, realized by the $(3,3)$-division $(R_1,R_2,R_3), (C_1,C_2,C_3)$.
The zone $R_2 \cap C_2$ is mixed, so it contains a -1 or a 1.
If it is a -1, then by transitivity the zone $R_3 \cap C_1$ is entirely -1, a contradiction to its being mixed.
A similar contradiction holds when there is a 1 entry in $R_2 \cap C_2$: zone $R_1 \cap C_3$ is entirely 1.
See \cref{fig:bounded-width-posets} for an illustration.
Hence, by \cref{thm:gridtheorem}, the twin-width of $A$ (and the twin-width of $P$ seen as a directed graph) is bounded by $4c_k \cdot 4^{4c_k+2} = 2^{2^{O(k)}}$.

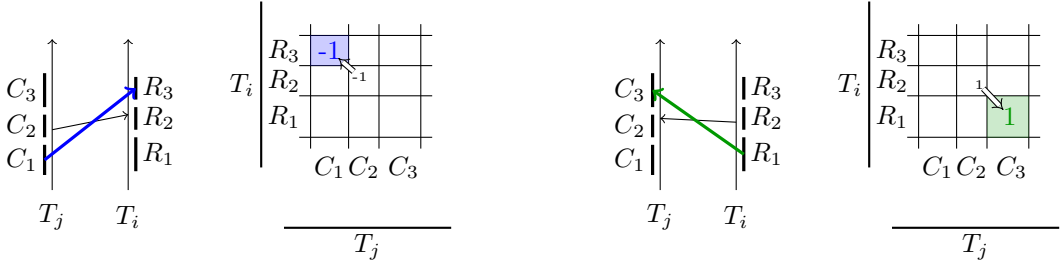
\begin{figure}
  \centering
  \begin{tikzpicture}[
      extended line/.style={shorten >=-#1,shorten <=-#1},
      extended line/.default=0.2cm]
    \def\h{2}
    \foreach \i/\l in {0/j,1/i}{
      \draw[->] (\i,0) -- (\i,\h) ;
      \node at (\i,-0.35) {$T_\l$} ;
    }
    \def\o{0.1}
    \foreach \a/\b/\c in {0.2/0.6/1,0.7/1/2,1.1/1.55/3}{
      \draw[very thick] (-\o,\a) -- (-\o,\b) ;
      \pgfmathsetmacro{\d}{(\a+\b)/2}
      \node at (-4 * \o,\d) {$C_\c$} ;
    }
    \foreach \a/\b/\c in {0.25/0.7/1,0.8/1.1/2,1.2/1.5/3}{
      \draw[very thick] (1+\o,\a) -- (1+\o,\b) ;
      \pgfmathsetmacro{\d}{(\a+\b)/2}
      \node at (1 + 4 * \o,\d) {$R_\c$} ;
    }
    \draw[->] (0,0.8) -- (1,1) ;
    \draw[blue,very thick,->] (- \o,0.4) -- (1+\o,1.35) ;

    \begin{scope}[xshift=3.25cm,yshift=0.5cm]
      \foreach \i in {0.15,0.65,1.05,1.6}{
        \draw[extended line] (\i,0.25) -- (\i,1.5) ;
      }
      \foreach \i/\j in {0.4/1,0.85/2,1.375/3}{
        \node at (\i,-0.2) {$C_\j$} ;
      }
      \draw[thick] (-0.2,-1) -- (2,-1) ;
      \node at (0.9,-1.25) {$T_j$} ;
      \foreach \i in {0.2,0.75,1.15,1.55}{
        \draw[extended line] (0.2,\i) -- (1.55,\i) ;
      }
      \foreach \i/\j in {0.475/1,0.95/2,1.35/3}{
        \node at (-0.2,\i) {$R_\j$} ;
      }
      \draw[thick] (-0.5,-0.2) -- (-0.5,2) ;
      \node at (-0.75,0.9) {$T_i$} ;

      \node[circle,inner sep=-0.01cm] (oe) at (0.8,1) {\tiny{-1}} ;
      \fill[blue,opacity=0.2] (0.15,1.15) -- (0.15,1.55) -- (0.65,1.55) -- (0.65,1.15) -- cycle ;
      \node[circle,inner sep=-0.23cm] (oz) at (0.4,1.35) {\textcolor{blue}{-1}} ;
      \draw[-implies,double equal sign distance] (oe) -- (oz) ;
    \end{scope}

    \begin{scope}[xshift=8cm]
    \foreach \i/\l in {0/j,1/i}{
      \draw[->] (\i,0) -- (\i,\h) ;
      \node at (\i,-0.35) {$T_\l$} ;
    }
    \def\o{0.1}
    \foreach \a/\b/\c in {0.2/0.6/1,0.7/1/2,1.1/1.55/3}{
      \draw[very thick] (-\o,\a) -- (-\o,\b) ;
      \pgfmathsetmacro{\d}{(\a+\b)/2}
      \node at (-4 * \o,\d) {$C_\c$} ;
    }
    \foreach \a/\b/\c in {0.25/0.7/1,0.8/1.1/2,1.2/1.5/3}{
      \draw[very thick] (1+\o,\a) -- (1+\o,\b) ;
      \pgfmathsetmacro{\d}{(\a+\b)/2}
      \node at (1 + 4 * \o,\d) {$R_\c$} ;
    }
    \draw[->] (1,0.9) -- (0,0.95) ;
    \draw[green!60!black,very thick,->] (1+\o,0.475) -- (-\o,1.325) ;

    \begin{scope}[xshift=3.25cm,yshift=0.5cm]
      \foreach \i in {0.15,0.65,1.05,1.6}{
        \draw[extended line] (\i,0.25) -- (\i,1.5) ;
      }
      \foreach \i/\j in {0.4/1,0.85/2,1.375/3}{
        \node at (\i,-0.2) {$C_\j$} ;
      }
      \draw[thick] (-0.2,-1) -- (2,-1) ;
      \node at (0.9,-1.25) {$T_j$} ;
      \foreach \i in {0.2,0.75,1.15,1.55}{
        \draw[extended line] (0.2,\i) -- (1.55,\i) ;
      }
      \foreach \i/\j in {0.475/1,0.95/2,1.35/3}{
        \node at (-0.2,\i) {$R_\j$} ;
      }
      \draw[thick] (-0.5,-0.2) -- (-0.5,2) ;
      \node at (-0.75,0.9) {$T_i$} ;

      \node[circle,inner sep=-0.01cm] (oe) at (0.95,0.9) {\tiny{1}} ;
      \fill[green!60!black,opacity=0.2] (1.05,0.2) -- (1.6,0.2) -- (1.6,0.75) -- (1.05,0.75) -- cycle ;
      \node[circle,inner sep=-0.18cm] (oz) at (1.35,0.475) {\textcolor{green!60!black}{1}} ;
      \draw[-implies,double equal sign distance] (oe) -- (oz) ;
    \end{scope}
    \end{scope}
  \end{tikzpicture}
  \caption{Left: If there is one arc from $C_2$ to $R_2$, then by transitivity there are all arcs from $C_1$ to $R_3$. On the matrix, this translates as: a -1 entry in $R_2 \cap C_2$ implies that all the entries of $R_3 \cap C_1$ are -1. Right: Similarly, a 1 entry in $R_2 \cap C_2$ implies that all the entries of $R_1 \cap C_3$ are 1. Hence at least one zone among $R_3 \cap C_1$, $R_2 \cap C_2$, $R_1 \cap C_3$ is constant, a contradiction to the $3k$-mixed minor.}
  \label{fig:bounded-width-posets}
\end{figure}

Of course there was a bit of work to establish \cref{thm:gridtheorem} inspired by the Guillemot-Marx framework, and supported by the Marcus-Tardos theorem.
There was even more work to prove that FO model checking is FPT on bounded twin-width (di)graphs.
It is nevertheless noteworthy that once that theory is established, the proof that bounded twin-width captures the posets of bounded width is lightning fast.
Indeed the known FPT algorithm on posets of bounded width~\cite{Gajarsky15} is a strong result, itself generalizing or implying the tractability of FO model checking on several geometric classes~\cite{Ganian15,Hlineny19}, as well as algorithms for \emph{existential} FO model checking on posets of bounded width~\cite{Bova16,Gajarsky14}.
We observe that posets of bounded twin-width constitute a strict superset of posets of bounded width.
Arcless posets are trivial separating examples, which have unbounded maximum antichain and twin-width 0.
A more elaborate example would be posets whose cover digraph is a directed path on $\sqrt n$ vertices in which all vertices are substituted by an independent set of size $\sqrt n$.
These posets have width $\sqrt n$ and twin-width 1.

We observe that while this paper was under review, Balab\'an and Hlinen\'y showed a linear upper bound $O(k)$ in the twin-width of posets of width~$k$~\cite{Balaban21}.
Their proof does not rely on the Marcus-Tardos theorem and gives directly a good contraction sequence.

The next example does not qualify as a ``lightning fast'' membership proof to bounded twin-width.
It shows however that the good vertex-ordering can be far less straightforward.  


\subsection{Proper minor-closed classes}
A more intricate example is given by proper minor-closed classes.
By definition, a proper minor-closed class does not contain some graph $H$ as a minor.
This implies in particular that it does not contain $K_{|V(H)|}$ as a minor.
Thus we only need to show that $K_t$-minor free graphs have bounded twin-width.

If the $K_t$-minor free graph $G$ admits a Hamiltonian path, things become considerably simpler.
We can enumerate the vertices of $G$ according to this path and write the corresponding adjacency matrix $A$.
The crucial observation is that a $k$-mixed minor yields a $K_{k/2,k/2}$-minor, hence a $K_{k/2}$-minor.
So $A$ cannot have a $2t$-mixed minor, and by \cref{thm:gridtheorem}, the twin-width of $G$ bounded (by $4c_{2t}2^{4c_{2t}+2}=2^{t^{O(t)}}$).
Unfortunately, a Hamiltonian path is not always granted in $G$.
A depth-first search (DFS for short) tree may emulate the path, but any DFS will not necessarily work.
Interestingly the main tool of the following theorem is a carefully chosen Lex-DFS.

\begin{table}
\begin{center}
 \begin{tabular}{cccc} 
 \toprule
  & Permutations avoiding $\sigma$ & Posets of width $w$ & $K_t$-minor free graphs  \\ 
 \midrule
 ordering  & imposed  & chains put one after the other & ad-hoc Lex-DFS    \\
 bound & $2^{O(|\sigma|)}$   & $2^{2^{O(w)}}$  &  $2^{2^{2^{O(t)}}}$ \\
 \bottomrule
 \end{tabular}
\end{center}
\caption{Choice of the ordering and bound on the twin-width for the classes tackled in \cref{sec:bounded-twinwidth}.}
\label{tbl:orders}
\end{table}

We note that a much simpler proof of the following theorem is obtained in~\cite{tww-seq} via a directed version of twin-width (where red edges come with an orientation).
However a different result in the same paper requires the proof that we are about to give here.

\begin{theorem}\label{thm:Ktminor}
  We set $g: t \mapsto 2(2^{4t+1}+1)^2$, $c_k := 8/3(k+1)^22^{4k}$, and $f: t \mapsto 4c_{g(t)} 2^{4c_{g(t)}+2}$. 
  Every $K_t$-minor free graph have twin-width at most $f(t)=2^{2^{2^{O(t)}}}$.
\end{theorem}

\begin{proof}
Let $G$ be a $K_t$-minor free graph, and $n := |V(G)|$.
We wish to upperbound the twin-width of $G$. 
We may assume that $G$ is connected since the twin-width of a graph is equal to the maximum twin-width of its connected components.

\textbf{Definition of the appropriate Lex-DFS.}
Let $v_1$ be an arbitrary vertex of $G$.
We perform a specific depth-first search from $v_1$.
A vertex is said \emph{discovered} when it is visited (for the first time) in the DFS.
The \emph{current discovery order} is a total order $v_1, \ldots, v_\ell$ among the discovered vertices, where $v_i$ was discovered before $v_j$ whenever $i < j$.
We may denote that fact by $v_i \prec v_j$, and $v_i \preceq v_j$ if $i$ and $j$ may potentially be equal.
The \emph{current DFS tree} is the tree on the discovered vertices whose edges correspond to the usual parent-to-child exploration.
The \emph{active} vertex is the lastly discovered vertex which still has at least one undiscovered neighbor.
Initially the active vertex is $v_1$, and when all vertices have been discovered, there is no longer an active vertex.
Before that, since $G$ is connected, the active vertex is always well-defined.
The (full) \emph{discovery order} is the same total order when all the vertices have been discovered.
  
We shall now describe how we break ties among the undiscovered neighbors of the active vertex.
Let $v_1, \ldots, v_\ell$ be the discovered vertices (with $\ell < n$), $\mathcal T_\ell$ be the current DFS tree, and $v$ be the active vertex.
  Let $C_1, \ldots, C_s$ be the vertex sets of the connected components of $G - V(\mathcal T_\ell)$ intersecting $N_G(v)$.
  By definition of the active vertex, $s \geqslant 1$.
  For each $i \in [s]$, we interpret $N_G(C_i) \cap V(\mathcal T_\ell)$ as a word $w_\ell(C_i)$ of $\{0,1\}^\ell$ where, for every $j \in [l]$, the $j$-th letter of $w_\ell(C_i)$ is a 1 if and only if $v_j \in N_G(C_i) \cap V(\mathcal T_\ell)$.
  If $w$ and $w'$ are two words on the alphabet $\{0,1\}$, we denote by $w \leq_{\text{lex}} w'$ the fact that $w$ is not greater than $w'$ in the lexicographic order derived from $0 < 1$.
  We can now define the successor of $v_\ell$ in the discovery order.
  The new vertex $v_{\ell+1}$ is chosen as an arbitrary vertex of $C_i \cap N_G(v)$ where $w_\ell(C_j) \leq_{\text{lex}} w_\ell(C_i)$ for every $j \in [s]$.
  Informally we visit first the component having the neighbors appearing first in the current discovery order.
  
  \textbf{The Lex-DFS discovery to order the adjacency matrix $\mathbf{M}$.}
  Let $v_1, \ldots, v_n$ be the eventual discovery order, and let $\mathcal T$ be the complete DFS tree.
  Let $M$ be the $\{0,1\}^{n \times n}$ matrix obtained by ordering the rows and columns of the adjacency matrix of $G$ accordingly to the discovery order.
  We set $g(t) := 2h(t)^2$ and $h(t) := 2^{4t+1}+2$.
  We will show that $M$ is $g(t)$-mixed free, actually even $g(t)$-grid free.
  For the sake of contradiction, let us suppose that $M$ has a $g(t)$-grid minor defined by the consecutive sets of columns $C_1, C_2, \ldots, C_{g(t)}$ and the consecutive sets of rows $R_1, R_2, \ldots, R_{g(t)}$.
  
  Now our goal is to show that we can contract a non-negligible amount of the $C_j$ and $R_i$, thereby exhibiting a $K_t$-minor.
  Actually the $K_t$-minors will arise from $K_{a,b}$-minors with $t \leqslant \min(a,b)$.
  We observe that either $\bigcup_{j \in [1,g(t)/2]} C_j$ and $\bigcup_{i \in [g(t)/2+1,g(t)]} R_i$ are disjoint, or $\bigcup_{j \in [g(t)/2+1,g(t)]} C_j$ and $\bigcup_{i \in [1,g(t)/2]} R_i$ are disjoint.
  Without loss of generality, let us assume that the former condition holds, and we will now try to find a $K_{t,t}$-minor between $C_1, \ldots, C_{g(t)/2}$ and $R_{g(t)/2+1}, \ldots, R_{g(t)}$.
  To emphasize the irrelevance of the first sets being columns and the second sets being rows, we rename $C_1, \ldots, C_{g(t)/2}$ by $A_1, \ldots, A_{g(t)/2}$, and $R_{g(t)/2+1}, \ldots, R_{g(t)}$ by $B_1, \ldots, B_{g(t)/2}$.
  
  Note that all the vertices of $\bigcup_{i \in [g(t)/2]}A_i$ are consecutive in the discovery order and appear before the consecutive vertices $\bigcup_{i \in [g(t)/2]}B_i$.
  Another important fact is that there is at least one edge between every pair $(A_i,B_j)$ (by definition of a mixed minor, or even grid minor).
  Thus let $a_{i,j} \in A_i$ be an arbitrary vertex with at least one neighbor $b_{i,j}$ in $B_j$.
  At this point, if we could contract each $A_i$ and $B_j$, we would be immediately done.
  This is possible if all these sets induce a connected subgraph.
  We will see that this is essentially the case for the sets of $\{A_i\}_{i \in [g(t)/2]}$, but not necessarily for the $\{B_j\}_{j \in [g(t)/2]}$.

 \textbf{The $\mathbf{\{A_i\}_i}$ essentially induce disjoint paths along the same branch.}
  Let $A'_i$ be the vertex set of the minimal subtree of $\mathcal T$ containing $\bigcup_{j \in [g(t)/2]} \{a_{i,j}\}$.
  The following lemma only uses the definition of a DFS, and not our specific tie-breaking rules.

  \begin{lemma}\label{lem:all-in-paths}
    All the vertices $a_{i,j}$, for $i,j \in [g(t)/2]$, lie on a single branch of the DFS tree with, in the discovery order, first $\bigcup_{j \in [g(t)/2]} \{a_{1,j}\}$, then $\bigcup_{j \in [g(t)/2]} \{a_{2,j}\}$, and so on, up to $\bigcup_{j \in [g(t)/2]} \{a_{g(t)/2,j}\}$.
    In particular, the sets $A'_i$ induce pairwise-disjoint paths in $\mathcal T$ along the same branch.
  \end{lemma}
  \begin{proof}
    Assume for the sake of contradiction that $a_{i,j}$ and $a_{i',j'}$, with $a_{i,j} \prec a_{i',j'}$, are not in an ancestor-descendant relationship in $\mathcal T$.
    Let $w$ be the least common ancestor of $a_{i,j}$ and $a_{i',j'}$, and $\mathcal T_w$ the current DFS tree the moment $w$ is discovered.
    Hence $w \prec a_{i,j}$.
    We claim that $b_{i,j}$ would be discovered before $a_{i',j'}$, a contradiction.
    Indeed when $a_{i,j}$ is discovered, it becomes the active vertex (due, for instance, to the mere existence of $b_{i,j}$).
    By design of a DFS, $a_{i,j}$ is not in the same connected component of $G - \mathcal T_w$ as $a_{i',j'}$, but its neighbor $b_{i,j}$ obviously is.
    So this connected component, and in particular $b_{i,j}$, is fully discovered before $a_{i',j'}$.
    This proves that the sets $A'_i$ induce paths in $\mathcal T$ along the same branch.

    We claim that these paths are pairwise disjoint and in the order (from root to bottom) $A'_1, A'_2, \ldots,$ $A'_{g(t)/2}$.
    This is immediate since, for every $i < i'$, $a_{i,j} \prec a_{i',j'}$.
    Thus $a_{i,j}$ can only be an ancestor of $a_{i',j'}$ in $\mathcal T$.
    One can also observe that $A'_i \subseteq A_i$ for every $i \in [g(t)/2]$.
  \end{proof}

  \textbf{Handling the $\mathbf{\{B_j\}_j}$ with the enhancements $\mathbf{\{B^*_j\}_j}$.}
  Let $B^*_j$ be the vertex set of the minimum subtree of $\mathcal T$ containing $B_j$.
  Since $B_j$ consist of consecutive vertices in the discovery order, $B^*_j = B_j \uplus P_j$ where $P_j$ is a path on a single branch of $\mathcal T$.
  One can see $B^*_j$ as an \emph{enhancement} of $B_j$.

  We show that except maybe the last $A'_i$, namely $A'_{g(t)/2}$, every set enhancement $B^*_j$ is disjoint from every $A'_i$.
  \begin{lemma}\label{lem:tree-contraction}
    For every $j \in [g(t)/2]$, for every $i \in [g(t)/2-1]$, $B^*_j \cap A'_i = \emptyset$. 
  \end{lemma}
  \begin{proof}
    There is an edge between $A'_{g(t)/2}$ and each $B_j$.
    Every $B_j$ succeeds $A'_{g(t)/2}$ in the discovery order.
    Therefore all the vertices of $\bigcup_{j \in [g(t)/2]} B_j$ appear in $\mathcal T$ in the subtree of the firstly discovered vertex, say $u$, of $A'_{g(t)/2}$.
    Hence all the trees $B^*_j$ are fully contained in $\mathcal T[u]$ the subtree of $\mathcal T$ rooted at $u$.
    We can then conclude since, by \cref{lem:all-in-paths}, all the vertices of $\bigcup_{j \in [g(t)/2-1]} A'_j$ are ancestors of~$u$.
  \end{proof}
  An enhancement is connected by design.
  Furthermore, by \cref{lem:tree-contraction} contracting (in the usual minor sense) a $B^*_j$ would not affect almost all $A'_i$.
  The remaining obvious issue that we are facing is that a pair of enhancements $B^*_j$ and $B^*_{j'}$ may very well overlap.
  Thus we turn our attention to their intersection graph.

  \textbf{The intersection graph $\mathbf{H}$ of the enhancements.}
  Let $H$ be the intersection graph whose vertices are $B^*_1, \ldots, B^*_{g(t)/2}$ and there is an edge between two vertices whenever the corresponding sets intersect.
  As an intersection graph of subtrees in a tree, $H$ is a chordal graph.
  In particular $H$ is a perfect graph, thus $\alpha(H) \omega(H) \geqslant |V(H)|=g(t)/2$.
  Therefore either $\alpha(H) \geqslant \sqrt{g(t)/2}$ or $\omega(H) \geqslant \sqrt{g(t)/2}$. 
  Moreover in polynomial-time, we can compute an independent or a clique of size $\sqrt{g(t)/2}=h(t)=2^{4t+1}+2 > t$.
  If we get a large independent set $I$ in $H$, we can contract the edges of each $B^*_j$ corresponding to a vertex of $I$.
  By~\cref{lem:tree-contraction} we can also contract any $h(t)$ paths $A'_i$ which are not $A'_{g(t)/2}$, and obtain a $K_{h(t),h(t)}$ (which contains a $K_{h(t)}$-minor, hence a $K_t$-minor).
  We thus assume that we get a large clique~$C$ in~$H$.

  \textbf{$\mathbf{H}$ has a clique $\mathbf{C}$ of size at least $\mathbf{h(t)}$.}
  By the Helly property satisfied by the subtrees of a tree, there is a vertex $v$ of $\mathcal T$ (or of $G$) such that every $B^*_j \in C$ contains $v$.
  If we potentially exclude the $B^*_j$ of $C$ with smallest and largest index, all the other elements of $C$ are fully contained in $\mathcal T[v]$ the subtree of $\mathcal T$ rooted at $v$.
  Let $C_1, \ldots, C_s$ be the connected components of $\mathcal T[v] - \{v\}$, ordered by the Lex-DFS discovery order.
  Thus $v$ has $s$ children in $\mathcal T$.
  
  \textbf{The enhancements of $\mathbf{C}$ essentially intersect only at $\mathbf{v}$.}
  We show that each connected component may intersect only a very limited number of $B^*_j \in C$.
  \begin{lemma}\label{lem:few-in-cc}
    For every $i \in [s]$, the connected component $C_i$ intersects at most two $B^*_j \in C$.
  \end{lemma}
  \begin{proof}
    Assume by contradiction that there is a connected component $C_i$ intersecting $B^*_{j_1}, B^*_{j_2}, B^*_{j_3} \in C$, with $j_1 < j_2 < j_3$.
    Since $B_{j_2}$ appears after $B_{j_1}$ and before $B_{j_3}$ in the discovery order, $B_{j_2}$ is fully contained in $C_i$.
    Hence $B^*_{j_2}$ is also contained in $C_i$ and cannot contain $v$, a contradiction.
  \end{proof}
  Moreover \cref{lem:few-in-cc} shows that only two consecutive $B^*_{j_1}, B^*_{j_2} \in C$ (by consecutive, we mean that there is no $B^*_j \in C$ with $j_1 < j < j_2$) may intersect the same connected component of $\mathcal T[v] - \{v\}$.
  Let us relabel $D_1, \ldots, D_{(h(t)-1)/2}$, every other elements of $C$ except the last one (keeping the same order).
  Now no connected component $C_i$ intersects two distinct sets $D_j, D_{j'}$.
  Each $D_j$ defines an interval $I_j := [\ell(j),r(j)]$ of the indices $i$ such that $D_j$ intersects $C_i$.
  The sets $I_j$ are pairwise-disjoint intervals.

  \textbf{Definitions of the pointers $\mathbf{z, j_b, j_e}$ to iteratively build $\mathbf{\mathcal S}$ and $\mathbf{\mathcal L}$.}
  Let $z_1 \in N_G(C_{r(1)})$ be such that for every $z' \in N_G(C_{r(1)})$, $z_1 \preceq z'$.
  This vertex exists by our DFS tie-breaking rule and the fact that there is an edge between, say, $a_{2,1}$ and $b_{2,1}$ (recall that this edge links $A_2$ and $B_1$).
  We initialize three pointers $z, j_b, j_e$ and two sets $\mathcal S, \mathcal L$ as follows: $z := v_1$ (the starting vertex in the DFS discovery order), $j_b := 1$, $j_e := (h(t)-2)/2=2^{4t}$, $\mathcal S := \emptyset$, and $\mathcal L := \emptyset$.
  Informally the indices $j_b$ (begin) and $j_e$ (end) lowerbound and upperbound, respectively, the indices of the sets $\{D_j\}_j$ we are still working with.
  Every vertex $v \prec z$ is simply disregarded.
  
  The sets $\mathcal S$ and $\mathcal L$ collect vertices (all discovered before $B_1$ in the Lex-DFS order) which can be utilized to form a large biclique minor in two different ways.
  Vertices stored in $\mathcal S$ are not adjacent to too many $\{D_j\}_j$, thus they can be used to ``connect'' the components of some $D_j - \{v\}$ without losing too many other $D_{j'}$.
  Vertices stored in $\mathcal L$ are adjacent to very many $\{D_j\}_j$, so they can directly form a biclique minor with the leftmost connected component of the corresponding $\{D_j\}_j$. 

  Let $j_1 \in [(h(t)-2)/2]$ be the smallest index such that $N_G(C_{\ell(j_1)})$ does not contain $z_1$.
  We distinguish two cases: $j_1 \leqslant (h(t)-2)/4=2^{4t-1}$ and $j_1 > 2^{4t-1}$.
  If $j_1 \leqslant 2^{4t-1}$, we will use $z_1$ to connect all connected components intersecting $D_1$: that is, $C_{\ell(1)}, C_{\ell(1)+1}, \ldots, C_{r(1)}$.
  In that case, we set: $j_b := j_1$ and $\mathcal S := \mathcal S \cup \{z_1\}$.
  
  If instead $j_1 > 2^{4t-1}$, we will use $z_1$ itself as a possible vertex of a biclique minor.
  In that case we set: $j_e := j_1 - 1$ and $\mathcal L := \mathcal L \cup \{z_1\}$.
  Observe that in both cases the length $|j_e - j_b|$ is at most halved.
  Hence we can repeat this process $\log{2^{4t}}/2=2t$ times.
  In both cases we replace the current $z$ by the successor of $z_1$ in the DFS discovery order.

  At the second step, we let $z_2 \in N_G(C_{r(j_b)})$ be such that for every $z' \in N_G(C_{r(j_b)})$ with $z \preceq z'$, then $z_2 \prec z'$.
  In words, $z_2$ is the first vertex (in the discovery order) appearing after $z$ with a neighbor in the last connected component $C_i$ intersecting the current first $D_j$, namely $D_{j_b}$.
  Again this vertex exists by the DFS tie-breaking rule.
  We define $j_2 \in [j_b,j_e]$ as the smallest index such that $N_G(C_{\ell(j_2)})$ does not contain $z_2$.
  We distinguish two cases: $j_2$ below or above the threshold $(j_b+j_e)/2$, and so on.

  \textbf{Building a large minor when $\mathbf{|\mathcal L|}$ is large.}
  After $\log{((h(t)-2)/2)}/2=2t$ steps, $\max(|\mathcal S|,|\mathcal L|)$ $ \geqslant t$.
  Indeed at each step, we increase $|\mathcal S|+|\mathcal L|$ by one unit.
  Also the length $|j_e - j_b|$ after these steps is still not smaller than $2^{4t}/2^{2t} = 2^{2t}$.
  If $|\mathcal L| \geqslant t$, then we exhibit a $K_{t,t}$-minor in $G$ in the following way.
  We contract $C_{\ell(j)}$ to a single vertex, for every $j \in [j_b,j_e]$ (recall that $|j_e - j_b|>2^{2t}$).
  These vertices form with the vertices of $\mathcal L$ a $K_{2^{2t},|\mathcal L|}$, thus a $K_{t,t}$-minor, and a $K_t$-minor.

  \textbf{Building a large minor when $\mathbf{|\mathcal S|}$ is large.}
  If instead $|\mathcal S| \geqslant t$, then we exhibit the following $K_{t,t}$-minor.
  We use each $z_i \in \mathcal S$, to connect the corresponding sets $D_j \setminus \{v\}$.
  We contract $\{z_i\} \cup D_j \setminus \{v\}$ to a single vertex.
  We then contract all the disjoint paths $A'_i$ (recall \cref{lem:all-in-paths}) which are not $A'_{g(t)/2}$ nor contain a vertex in $\mathcal S$.
  This represents at least $g(t)/2-1-2t > t$ vertices.
  This yields a biclique $K_{t,t}$, hence $G$ as a $K_t$-minor.

  \textbf{Concluding on the twin-width of $\mathbf{G}$.}
  The two previous paragraphs reach a contradiction.
  Hence the adjacency matrix $M$ is $g(t)$-mixed free, and even $g(t)$-grid free.
  By \cref{thm:gridtheorem} this implies that the twin-width of $G$ is at most $4c_{g(t)} 2^{4c_{g(t)}+2}$, where $c_k := 8/3(k+1)^22^{4k}$, which was the announced triple-exponential bound.
\end{proof}

Applied to planar graphs, which are $K_5$-minor free, the previous theorem gives us a constant bound on the twin-width, but that constant has billions of digits.
We believe that the correct bound should have only one digit.
It is natural to ask for a more reasonable bound in the case of planar graphs.
An attempt could be to show that for a large enough integer $d$, every planar $d$-trigraph admits a $d$-contraction which preserves planarity.
However \cref{fig:planar-issue} shows that this statement does not hold.

\begin{figure}[h!]
  \centering
  \begin{tikzpicture}[scale=0.9]
    \def\t{9}
    \def\s{0.3}
    \def\se{-0.05}
    \pgfmathtruncatemacro\tm{\t-1}
    \foreach \j in {1,2,3}{
    \begin{scope}[xshift=2 * \j * \tm * \s cm, scale=\s]
    \foreach \i in {1,...,\tm}{
      \node[draw, circle, inner sep=\se cm] (a\i) at (2 * \i,0) {} ;
      \node[draw, circle, inner sep=\se cm] (x\i) at (2 * \i + 1,-0.5) {} ;
      \node[draw, circle, inner sep=\se cm] (y\i) at (2 * \i + 1,0.5) {} ;
      \draw (a\i) -- (x\i) -- (y\i) -- (a\i) ;
    }
    \node[draw, circle, inner sep=\se cm] (a\t) at (2 * \t,0) {} ;
   
    \foreach \i in {1,...,\tm}{
      \pgfmathtruncatemacro\ip{\i+1}
      \draw (y\i) -- (a\ip) -- (x\i) ;
    }
    \node[draw, circle] (xx\j) at (\t+1,-8) {$x_\j$} ;
    \node[draw, circle] (yy\j) at (\t+1,8) {$y_\j$} ; 

    \foreach \i in {1,...,\tm}{
      \draw (y\i) -- (yy\j) -- (a\i) -- (xx\j) -- (x\i) ;
    }
    \draw (yy\j) -- (a\t) -- (xx\j) ;
    \foreach \i in {1,3,5,7}{
      \draw[very thick, red] (xx\j) -- (x\i) ;
    }
    \foreach \i in {2,4,6,8}{
      \draw[very thick, red] (yy\j) -- (y\i) ;
    }
    \end{scope}
    }
    \foreach \la in {xx,yy}{
      \draw (\la1) -- (\la2) -- (\la3) --++ (2.5,0) ;
      \draw (\la1) --++ (-2.5,0) ;
    }
  \end{tikzpicture}
  \caption{For every integer $d$ (here $d=4$), a planar $d$-trigraph without any $d$-contraction to a planar graph.
  The graph should be thought of as wrapped around a cylinder: there are edges $x_1x_3$ and $y_1y_3$, and the leftmost and rightmost vertices are actually the same vertex.}
  \label{fig:planar-issue}
\end{figure}
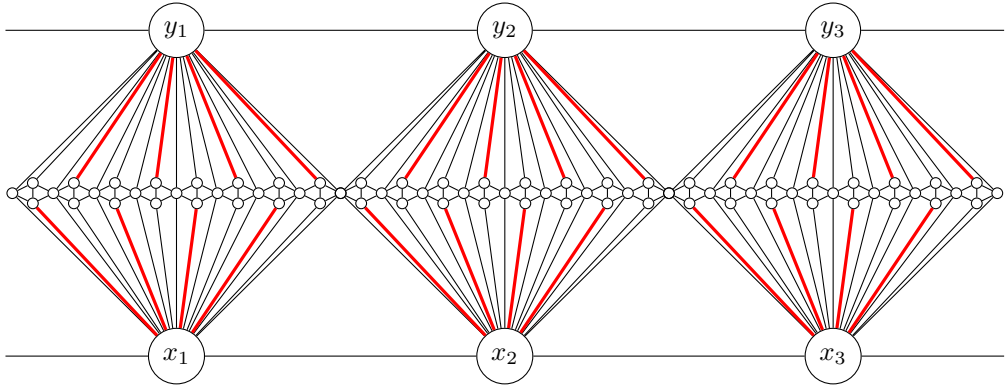

\section{FO model checking}\label{sec:fo}

In this section, we show that deciding first-order properties in $d$-collapsible graphs is fixed-parameter tractable in $d$ and the size of the formula. 
We let $E$ be a binary relation symbol.
A~graph~$G$ is seen as an $\{E\}$-structure with universe $V(G)$ and binary relation $E(G)$ (matching the arity of $E$).
A \emph{sentence} is a formula without free variables.

A \emph{formula $\phi$ in prenex normal form}, or simply \emph{prenex formula}, is any sentence written as a sequence of non-negated quantifiers followed by a quantifier-free formula:
$$\phi =  Q_1x_1Q_2x_2\dots Q_{\ell}x_{\ell} \phi^*$$
where for each $i \in [\ell]$, the variable $x_i$ ranges over $V(G)$, $Q_i \in \{\forall, \exists\}$, while $\phi^*$ is a Boolean combination in atoms of the form $x_i=x_j$ and $E(x_i,x_j)$.
Here we call \emph{length} of $\phi$ its number of variables $\ell$.
Note that this also corresponds to its quantifier depth.
Every formula with quantifier depth $k$ can be rewritten as a prenex formula of depth $\text{Tower}(k+\log^* k+3)$ (see Theorem 2.2. and inequalities (32) in \cite{Pikhurko11}).


\begin{theorem}\label{thm:FOmodelchecking}
Given as input a prenex formula $\phi$ of length $\ell$, an $n$-vertex graph $G$, and a $d$-sequence of $G$, one can decide $G \models \phi$ in time $f(\ell,d) \cdot n$.
\end{theorem}

Our proof of \cref{thm:FOmodelchecking} is not specific to a single formula.
Instead we compute a tree of size bounded by a function of $\ell$, which is sufficient to check every prenex formula $\phi$ of length~$\ell$.

\subsection{morphism-trees and shuffles}

All our trees are rooted and the root is denoted by $\varepsilon$.
An \emph{internal node} is a node with at least one child.
Non-internal nodes are called \emph{leaves}.
Given a node $x_i$ in a tree $T$, we call \emph{current path} of $x_i$ the unique path $\varepsilon, x_1, \ldots ,x_i$ from $\varepsilon$ to $x_i$ in $T$.
We will see this current path as the tuple $(x_1, \ldots , x_i)$.
The current path of $\varepsilon$ is the empty tuple, also denoted by~$\varepsilon$.
The \emph{depth} of a node $x$ is the number of edges in the current path of $x$.
A node $x$ is a \emph{descendant} of $y$ if $y$ belongs to the current path of $x$.
Given a tree $T$, we denote the parent of $x$ by $p_T(x)$.
Two nodes with the same parent are \emph{siblings}.
We denote by $T^*$ the set of nodes of $T$ distinct from its root $\varepsilon$, that is $V(T) \setminus \{\varepsilon\}$.

A bijection $f$ between the node sets of two trees $T_1, T_2$ is an \emph{isomorphism} if it commutes with the parent relation, i.e., $p_{T_2}(f(x))=f(p_{T_1}(x))$ for every node $x \in T_1^*$.
One can observe that $f^{-1}: V(T_2) \rightarrow V(T_1)$ is then also an isomorphism.
Two trees are said \emph{isomorphic} if there is an isomorphism between them.
An isomorphism mapping $T$ to itself is called an \emph{automorphism}.
Given a node $x$ in $T$, the \emph{subtree} of $x$, denoted by $B_T(x)$, is the subtree of $T$ rooted at $x$ and containing all descendants of $x$.

An \emph{$i$-tuple} is a tuple on exactly $i$ elements, and a \emph{$\leqslant i$-tuple} is a tuple on at most $i$ elements.
A \emph{subtuple} of a tuple $a$ is any tuple obtained by erasing some entries of $a$.
Given a tuple $a=(a_i)$ and a set $X$, the subtuple of $a$ \emph{induced by $X$}, denoted by $a_{|X}$ is the subtuple consisting of the entries $a_i$ which belongs to $X$.
Given two disjoint sets $A$ and $B$, and two tuples $a \in A^s$ and $b \in B^t$, a \emph{shuffle} $c$ of $a$ and $b$ is any tuple of $(A \cup B)^{s+t}$ such that $c_{|A}=a$ and $c_{|B}=b$.
For instance $(2,0,3,1,0)$ is one of the ten shuffles of $(0,1,0)$ and $(2,3)$.
Given a tuple $x = (x_1, \ldots, x_{k-1}, x_k)$, the \emph{prefix} of $x$ is $(x_1, \ldots, x_{k-1})$ if $k>1$, and $\varepsilon$ if $k=1$.

Given two trees $T_1$ and $T_2$ whose nodes are supposed disjoint, the \emph{shuffle} $s(T_1,T_2)$ of $T_1$ and $T_2$ is the tree whose nodes are shuffles of all pairs of tuples $P_1, P_2$ where $P_1$ is a current path in $T_1$ and $P_2$ is a current path in $T_2$.
The parent relation in $s(T_1,T_2)$ is the prefix relation.
The \emph{$\ell$-shuffle} $s_{\ell}(T_1,T_2)$ of $T_1$ and $T_2$ is the subtree of $s(T_1,T_2)$ obtained by keeping only the nodes with depth at most $\ell$. 

The formal definition of shuffle is somewhat cumbersome since the current path of the node $(x_1,x_2,\ldots,x_i)$ is the tuple $((x_1),(x_1,x_2),\ldots,(x_1,x_2,\ldots,x_i))$.
Given a set $V$, a \emph{morphism-tree in $V$} is a pair $(T,m)$ where $T$ is a tree and $m$ is a mapping from $T^*$ to $V$. Given a set $V$ and an integer $\ell$, we define the (complete) \emph{$\ell$-morphism-tree} $MT_{\ell}(V)=(T_{V,\ell},m_{V,\ell})$ as the morphism-tree in $V$ such that for every positive integer $i\leq \ell$ and every $i$-tuple $(v_1,\ldots,v_i)$ of possibly repeated elements of $V$, there is a unique node $x_i$ of $T_{V,\ell}$ whose current path $(x_1,\ldots,x_i)$ satisfies $m_{V,\ell}(x_j)=v_j$ for all $j=1,\ldots ,i$. Informally, $MT_{\ell}(V)$ represents all the ways of extending the empty set by iteratively adding one (possibly repeated) element of $V$ up to depth $\ell$ in a tree-search fashion. Note that if $V$ has size $n$, the number of nodes of $MT_{\ell}(V)$ is $n^{\ell}+n^{\ell -1}+\ldots  +1$. The formal way of defining $MT_{\ell}(V)$ is to consider that $T_{V,\ell}$ is the set of all tuples $u=(u_1, \ldots, u_i)$ of elements of $V$ with $0\leq i \leq \ell$, the parent relation is the prefix relation, and the image by $m_{V,\ell}$ of a tuple $(u_1, \ldots, u_i)$ is $u_i$.

Again, the formal definition of $MT_{\ell}(V)$ is cumbersome since the current path of the node $(u_1,u_2,\ldots,u_i)$ is the tuple $((u_1),(u_1,u_2),\ldots,(u_1,u_2,\ldots,u_i))$.
Hence, as an abuse of language, we may identify a node $(u_1,u_2,\ldots,u_i)$ to its current path. 
We can extend the notion of shuffle to morphism-trees by defining $(T,m)$ as the \emph{shuffle} of $(T_1,m_1)$ and $(T_2,m_2)$ where $T$ is the shuffle of $T_1$ and $T_2$ (supposed again on disjoint node sets) and for every node $x=(x_1, \ldots, x_k)$ of $T$, we let $m(x)=m_1(x_k)$ if $x_k\in T_1^*$ and $m(x)=m_2(x_k)$ if $x_k\in T_2^*$.
Again, we define the $\ell$-shuffle by pruning the nodes with depth more than $\ell$. 

\begin{lemma}\label{lem:unionshuffle}
  Let $(V_1,V_2)$ be a partition of a set $V$.
  The $\ell$-shuffle of $MT_{\ell}(V_1)$ and $MT_{\ell}(V_2)$ is $MT_{\ell}(V)$.
\end{lemma}
 
\begin{proof}
This follows from the fact that every $\leq \ell$-tuple of $V$ is uniquely obtained as the shuffle of some $\leq \ell$-tuple of $V_1$ and some $\leq \ell$-tuple of $V_2$.
\end{proof}

One can extend the definition of shuffle to several trees. Given a sequence of (node disjoint) morphism-trees $(T_1,m_1),\ldots ,(T_k,m_k)$, the nodes of the shuffle $(T,m)$ are all tuples which are shuffles $S$ of current paths $P_1,\ldots,P_k$. Precisely, a tuple $S$ is a node of $(T,m)$ if all its entries are non-root nodes of $T_i$'s, and such that each subtuple $S_i$ of $S$ induced by the nodes of $T_i$ is a (possibly empty) current path of $T_i$. As usual the parent relation is the prefix relation. Finally $m(x_1,\ldots ,x_i)$ is equal to $m_j(x_i)$ where $x_i\in T_j$. We speak of $\ell$-shuffle when we prune out the nodes with depth more than $\ell$. Note that $MT_{\ell}(V)$ is the $\ell$-shuffle of $MT_{\ell}(\{v\})$ for all $v\in V$.

\subsection{morphism-trees in graphs and reductions}

We extend our previous definitions to graphs.
The first step is to introduce graphs on tuples.
A \emph{tuple graph} is a pair $(x,G)$ where $x$ is a tuple $(x_1, \ldots,x_t)$ and $G$ is a graph on the vertex set $\{x_1, \ldots,x_t\}$ (where repeated vertices are counted only once).
Thus there is an edge $x_ix_j$ in $(x,G)$ if $x_ix_j$ is an edge of $G$.
The main difference with graphs is that vertices can be repeated within a tuple.
In particular if $x_1=x_3$ and there is an edge $x_1x_2$, then the edge $x_2x_3$ is also present.
Two tuple graphs $(x,G)$ and $(y,H)$ are \emph{isomorphic} if $x=(x_1, \ldots, x_t)$, $y=(y_1, \ldots, y_t)$ and we have both $x_i=x_j \Leftrightarrow y_i=y_j$, and $x_ix_j \in E(G)  \Leftrightarrow y_iy_j \in E(H)$, for every $i,j \in [t]$.

A \emph{morphism-tree in $G$} is a morphism-tree $(T,m)$ in $V(G)$, supporting new notions based on the edge set of $G$.
Given a node $x_i$ of $T$ with current path $(x_1, \ldots, x_i)$, the graph $G$ induces a tuple graph on $(m(x_1), \ldots, m(x_i))$, namely $((m(x_1), \ldots, m(x_i)), G[\{m(x_1), \ldots, m(x_i)\}])$.
We call \emph{current graph of $x_i$} this tuple graph.
Given a node $x_i$ and one of its children $x_{i+1}$, observe that the current graph of $x_{i+1}$ extends the one of $x_i$ by one (possibly repeated) vertex. Informally, a morphism-tree in $G$ can be seen as a way of iteratively extending induced subgraphs of $G$ in a tree-search fashion.

Two morphism-trees $(T,m)$ in $G$ and $(T',m')$ in $G'$ are \emph{isomorphic} if there exists an isomorphism $f$ from $T$ to $T'$ such that for every node $x \in T^*$ and $y$ descendant of $x$:
\begin{itemize}
\item $m(x)=m(y)$ if and only $m'(f(x))=m'(f(y))$.
\item $m(x)m(y)$ is an edge of $G$ if and only $m'(f(x))m'(f(y))$ is an edge of $G'$.
\end{itemize}

In particular, the current graph of a node is isomorphic to the current graph of its image.
Again an isomorphism $f$ from $(T,m)$ into itself is called an \emph{automorphism}.
Two sibling nodes $x, x'$ of a morphism-tree $(T,m)$ are \emph{equivalent} if there exists an automorphism $f$ of $(T,m)$ such that $f(x)=x'$ and $f(x')=x$.
Note that if such an automorphism exists, then there is one which is the identity function outside of $B_T(x) \cup B_T(x')$.
The interpretation of $x, x'$ being equivalent is that the current graph $H$ of their parent can be extended up to depth $\ell$ in $G$ in exactly the same way starting from $x$ or from $x'$.


The (complete) \emph{$\ell$-morphism-tree} $MT_{\ell}(G)$ of a graph $G$ is simply\footnote{Technically, we should denote it by $(MT_{\ell}(V(G)),G)$ but we will stick to this simpler notation.} $MT_{\ell}(V(G))$.
Observe that while $E(G)$ is irrelevant for the syntactic aspect of $MT_{\ell}(G)$, the structure of $G$ is nonetheless important for semantic properties of $MT_{\ell}(G)$.
Indeed equivalent nodes are defined in $MT_{\ell}(G)$ but not in $MT_{\ell}(V(G))$.
Let us give a couple of examples to clarify that point.
When $G$ is a clique, all the sibling nodes are equivalent in $MT_{\ell}(G)$.
When $G$ is a path on the same vertex set, the depth-1 nodes of $MT_{\ell}(G)$ mapped to the first and second vertices of the path are in general \emph{not} equivalent. 

Given two equivalent (sibling) nodes $x,x'$ of a morphism-tree $(T,m)$ in $G$, the \emph{$x,x'$-reduction} of $(T,m)$ is the morphism-tree obtained by deleting all descendants of $x'$ (including itself).
A \emph{reduction} of a morphism-tree is any morphism-tree obtained by iterating a sequence of $x,x'$-reductions. Finally a \emph{reduct} of $(T,m)$ is a reduction in which no further reduction can be performed; that is, none of the pairs of siblings are equivalent.

\begin{lemma}\label{lem:MTtree}
Any reduct of an $\ell$-morphism-tree has size at most $h(\ell)$ for some function $h$.
\end{lemma}
 
\begin{proof}
  Assume that $(T,m)$ is a reduct of an $\ell$-morphism-tree in a graph $G$.
  Consider a node $x_{\ell-1}$ of depth $\ell-1$ in $T$.
  The maximum number of pairwise non-equivalent children $x_\ell$ of $x_{\ell-1}$ is at most $2^{\ell-1}+\ell-1$.
  Indeed there are (at most) $2^{\ell-1}$ non isomorphic extensions of the current graph of $x_{\ell-1}$ by adding the new node $m(x_\ell)$, and (at most) $\ell-1$ possible ways for $m(x_\ell)$ to be a repetition of a vertex among $m(x_1), \ldots, m(x_{\ell-1})$.
  In particular $x_{\ell-1}$ has a bounded number of children in the reduct $(T,m)$, and therefore, there exist only a bounded number of non-equivalent $x_{\ell-1}$ which are children of some $x_{\ell -2}$.
  This bottom-up induction bounds the size of $(T,m)$ by a tower function in~$\ell$.
\end{proof}

Since $MT_{\ell}(G)$ represents all possible ways of iterating at most $\ell$ vertex extensions of induced subgraphs of $G$ (starting from the empty set), one can check any prenex formula $\phi$ of depth at most $\ell$ on $MT_{\ell}(G)$.
In the language of games, $MT_{\ell}(G)$ captures all possible games for Player $\exists$ and Player $\forall$ to form a joint assignment of the variables $x_1, \ldots, x_\ell$. 
So far this does not constitute an efficient algorithm since the size of $MT_{\ell}(G)$ is $O(n^{\ell + 1})$.
However reductions --deletions of one of two equivalent alternatives for a player-- do not change the score of the game.
Thus we want to compute reductions, or even reducts, and decide $\phi$ on these smaller trees.

\begin{lemma}\label{lem:obs-reduct}
  Given a reduction of $MT_{\ell}(G)$ of size $s$ and a prenex formula on $\ell$ variables, $G \models \phi$ can be decided in time $O(s)$, and in time $h(\ell)$ if the reduction is a reduct.
\end{lemma}
\begin{proof}
  Let $\phi = Q_1 x_1 Q_2 x_2 \dots Q_\ell x_\ell \phi^*$, where $\phi^*$ is quantifier-free.
  Let $T$ be the tree of the given reduction of $MT_{\ell}(G)$.
  We relabel the nodes of $T$ in the following way.
  At each leaf $(v_1, \ldots, v_\ell)$ of $T$, we put a 1 if $\phi^*(v_1, \ldots, v_\ell)$ is true, and a 0 otherwise.
  For each $i \in [0,\ell-1]$, at each internal node of depth $i$, we place a $\max$ if $Q_{i+1} = \exists$, and a $\min$ if $Q_{i+1} = \forall$.
  The computed value at the root of this minimax tree is 1 if $G \models \phi$, and 0 otherwise.
  Indeed this value does not change while we perform reductions on $MT_{\ell}(G)$.
  The overall running time is~$O(|T|)$.
  By~\cref{lem:MTtree}, if $T$ is a reduct then the overall running time is $h(\ell)$ for some tower function $h$.
\end{proof}

Let us now denote by $MT'_{\ell}(G)$ any reduct of $MT_{\ell}(G)$.
It can be shown by local confluence that $MT'_{\ell}(G)$ is indeed unique up to isomorphism, but we do not need this fact here.
Now our strategy is to compute $MT'_{\ell}(G)$ in linear FPT time using bounded twin-width.


We base our computation on a sequence of partitions of $V(G)$ achieving twin-width~$d$.
Let $\mathcal P = \{X_1, \ldots, X_p\}$ be a partition of $V(G)$.
Two distinct parts $X_i, X_j$ of $\mathcal P$ are \emph{homogeneous} if there are between $X_i$ and $X_j$ either all the edges or no edges.
Let $G_{\mathcal P}$ be the graph on vertex set $\mathcal P$ and edge set all the pairs $X_iX_j$ such that $X_i, X_j$ are distinct and not homogeneous.
If $G_{\mathcal P}$ has maximum degree at most~$d$, we say that $\mathcal P$ is a \emph{$d$-partition} of~$G$.
Note that an $n$-vertex graph $G$ has twin-width at most $d$ if it admits a \emph{sequence of $d$-partitions} ${\mathcal P}_n, {\mathcal P}_{n-1}, \ldots, {\mathcal P}_1$ where $\mathcal P_n$ is the finest partition, and for every $i \in [n-1]$, the partition ${\mathcal P}_i$ is obtained by merging two parts of ${\mathcal P}_{i+1}$.

Our central result is:
\begin{theorem}\label{thm:FOmodelchecking2}
 A reduct $MT'_{\ell}(G)$ can be computed in time $f(\ell,d) \cdot n$, given as input a~sequence of $d$-partitions of $G$. 
\end{theorem}

The proof will compute $MT'_{\ell}(G)$ iteratively by combining partial morphism-trees obtained alongside the sequence of $d$-partitions.
We start with the finest partition ${\mathcal P}_n$, where each morphism-tree is defined on a single vertex, and we finish with the coarsest partition ${\mathcal P}_1$ which results in the sought $MT'_{\ell}(G)$.
We will thus need to define a morphism-tree for a partitioned graph.
Before coming to these technicalities, let us illustrate how shuffles come into play for computing $MT'_{\ell}(G)$.
The following two lemmas are not needed for the rest of the proof, but they provide a good warm-up for the more technical arguments involving partitions.

The \emph{disjoint union} $G_1 \cup G_2$ of two graphs $G_1, G_2$ with pairwise-disjoint vertex sets is the graph on $V(G_1) \cup V(G_2)$ with no edges between the two graphs $G_1,G_2$.
In this particular case, reductions commute with shuffle.

\begin{lemma}\label{lem:shufflecommute}
  Let $(T_1,m_1)$ and $(T_2,m_2)$ be two morphism-trees in $G_1$ and in $G_2$, respectively (on disjoint vertex sets).
  Let $(T,m)$ be the shuffle of $(T_1,m_1)$ and $(T_2,m_2)$, defined in $G_1 \cup G_2$.
  Let $(T'_1,m'_1)$ be a reduction of $(T_1,m_1)$.
  Then the shuffle $(T',m')$ of $(T'_1,m'_1)$ and $(T_2,m_2)$ is a reduction of $(T,m)$.
\end{lemma}

\begin{proof}
  We just need to show the lemma for single-step reductions.
  Indeed after we prove that shuffling morphism-trees defined on a disjoint union commutes with a single reduction performed in the first morphism-tree, we can iterate this process to establish that it commutes with reductions in general.
  Let $f$ be an automorphism of $(T_1,m_1)$ which swaps the equivalent nodes $x,x'$ and is the identity outside of the subtrees rooted at $x$  and $x'$.
  Let $(T'_1,m'_1)$ be the $x,x'$-reduction of $(T_1,m_1)$.
  Consider the mapping $g$ from $V(T)$ into itself which preserves the root $\varepsilon$ and maps every node $Z=(z_1, \ldots, z_k)$ to $Z'=(\tilde f(z_1), \ldots, \tilde f(z_k))$ where $\tilde f(z_i)=f(z_i)$ if $z_i \in T_1^*$ and $\tilde f(z_i)=z_i$ if $z_i \in T_2^*$.

  We claim that $g$ is an automorphism of $(T,m)$.
  It is bijective since $f$ is bijective.
  It commutes with the parent relation since $p_T(g(Z))=p_T(g(z_1, \ldots, z_{k-1}, z_k))=p_T(\tilde f(z_1), \ldots,$ $\tilde f(z_{k-1}), \tilde f(z_k))=(\tilde f(z_1), \ldots, \tilde f(z_{k-1}))=g(p_T(Z))$.
  Furthermore $g$ behaves well with the morphism~$m$.
  Indeed, for every node $Z_1=(z_1, \ldots, z_i)$ of $T$ and descendant $Z_2=(z_1, \ldots, z_i, z_{i+1}, \ldots, z_k)$, we have:
  \begin{itemize}
  
  \item  If $m(Z_1)=m(Z_2)$, we either have $z_i,z_k\in T_1^*$ and $m_1(z_i)=m_1(z_k)$ and thus $m_1(f(z_i))=m_1(f(z_k))$ which implies $m(g(Z_1))=m_1(f(z_i))=m_1(f(z_k))=m(g(Z_2))$.
    Or we have $z_i,z_k\in T_2$ and $m_2(z_i)=m_2(z_k)$ which implies $m(g(Z_1))=m_2(z_i)=m_2(z_k)=m(g(Z_2))$.
  
  \item If $m(Z_1)m(Z_2)$ is an edge of $G_1 \cup G_2$ we either have $z_i,z_k\in T_1^*$ and $m_1(z_i)m_1(z_k)$ is an edge of $G_1$, or $z_i,z_k\in T_2$ and $m_2(z_i)m_2(z_k)$ is an edge of $G_2$.
    In the first case, $m_1(f(z_i))m_1(f(z_k))$ is an edge of $G_1$ and we conclude since $m_1(f(z_i))m_1(f(z_k))=m(g(Z_1))m(g(Z_2))$.
    In the second case, $m_2(z_i)m_2(z_k)=m(g(Z_1))m(g(Z_2))$ is an edge of $G_2$. Thus $g$ maps edges to edges, and therefore non-edges to non-edges.
   \end{itemize}
   
  Finally, consider any node $Z=(z_1, \ldots, z_k)$ of $(T,m)$ such that $z_k=x$.
  By definition of the shuffle and the fact that $x, x'$ are siblings, there is a node $Z'=(z_1, \ldots z_{k-1},x')$ in $(T,m)$.
  By construction, we have $g(Z)=Z'$ and $g(Z')=Z$ and thus $Z, Z'$ are equivalent in $(T,m)$.
  Therefore we can reduce all such pairs $Z, Z'$ in $(T,m)$ in order to find a reduction in which we have deleted all nodes of $(T,m)$ containing the entry $x'$, and therefore also all its descendants in $T_1$.
  This is exactly the shuffle $(T',m')$ of $(T'_1,m'_1)$ and $(T_2,m_2)$. 
\end{proof}
The previous lemma similarly holds for $\ell$-shuffles.
We can now handle the disjoint union of two graphs.

\begin{lemma}\label{lem:MTunion}
Given as input $MT'_{\ell}(G)$ and $MT'_{\ell}(H)$, two reducts of the graphs $G$ and $H$, one can compute a reduct $MT'_{\ell}(G\cup H)$ in time only depending on $\ell$.
\end{lemma}

\begin{proof}
  We just have to compute the $\ell$-shuffle $(T,m)$ of $MT'_{\ell}(G)$ and $MT'_{\ell}(H)$, in time depending on $\ell$ only.
  Indeed, by Lemma~\ref{lem:unionshuffle} the $\ell$-shuffle of $MT_{\ell}(G)$ and $MT_{\ell}(H)$ is  $MT_{\ell}(G\cup H)$.
  Therefore, by repeated use of~\cref{lem:shufflecommute} applied to the sequence of reductions from $MT_{\ell}(G)$ to $MT'_{\ell}(G)$ and from $MT_{\ell}(H)$ to $MT'_{\ell}(H)$, the morphism-tree $(T,m)$ is a reduction of $MT_{\ell}(G\cup H)$.
  Note that $(T,m)$ is not necessarily a reduct but its size is bounded, and we can therefore reduce it further by a brute-force algorithm to obtain a reduct $MT'_{\ell}(G\cup H)$.
\end{proof}

We now extend our definitions to partitioned graphs.
Let $G$ be a graph and $\mathcal P$ be a partition of $V(G)$.
A morphism-tree $(T,m)$ in $(G,\mathcal P)$ is again a morphism-tree in $V(G)$.
The difference with a morphism-tree in $G$ lies in the allowed reductions.
Now an automorphism~$f$ of $(T,m)$ in $(G,\mathcal P)$ is an automorphism of $(T,m)$ in $G$ which respects the partition $\mathcal P$.
Formally, for any node $x \in T^*$, the vertices $m(x)$ and $m(f(x))$ belong to the same part of~$\mathcal P$.
Two sibling nodes $x, x'$ in a morphism-tree $(T,m)$ in $(G,\mathcal P)$ are equivalent if there is an automorphism of $(T,m)$ in $(G,\mathcal P)$ which swaps $x$ and $x'$ (and in particular, $m(x)$ and $m(x')$ are in the same part of $\mathcal P$).

As previously, we define $MT_{\ell}(G,\mathcal P)$ for a partitioned graph $(G,\mathcal P)$ as equal to $MT_{\ell}(V(G))$, and we define $MT'_{\ell}(G,\mathcal P)$ as any reduct of $MT_{\ell}(G,\mathcal P)$, where reductions are performed in $(G,\mathcal P)$.
Observe that $MT'_{\ell}(G,\mathcal P)$ can be very different from $MT'_{\ell}(G)$.
For instance if $\mathcal P$ is the partition into singletons, no reduction is possible and thus $MT'_{\ell}(G,{\mathcal P})=MT_{\ell}(G,{\mathcal P})$.
At the other extreme, if $\mathcal P=\{V(G)\}$, then $MT'_{\ell}(G,{\mathcal P})$ is a reduct of $MT_{\ell}(G)$.

Our ultimate goal in order to use twin-width is to dynamically compute $MT'_{\ell}(G,{\mathcal P}_1)$ by deriving $MT'_{\ell}(G,{\mathcal P}_{i})$ from $MT'_{\ell}(G,{\mathcal P}_{i+1})$.
This strategy cannot directly work since the initialization requires $MT'_{\ell}(G,{\mathcal P}_n)$ which is equal to $MT_{\ell}(G,{\mathcal P}_n)$ of size $O(n^{\ell})$.
Instead, we only compute a partial information for each $(G,{\mathcal P}_i)$ consisting of all \emph{partial morphism-trees $MT'_{\ell}(G,{\mathcal P}_{i},X)$} centered around $X$, where $X$ is a part of ${\mathcal P}_i$.
We will make this formal in the next section.
Let us highlight though that for the initialization, the graph $G_{{\mathcal P}_n}$ consists of isolated vertices, therefore its connected components are singletons.
So the initialization step of our dynamic computation only consists of computing $MT'_{\ell}(\{v\})$ for all vertices $v$ in~$G$.
Since all such trees consist of a path of length $\ell$ whose non-root nodes are mapped to $v$, the total size of the initialization step is linear.
However, observe that the $\ell$-shuffle of all these $MT'_{\ell}(\{v\})$ gives $MT'_{\ell}(G,{\mathcal P}_n)$.
The essence of our algorithm can be summarized as: Maintaining a linear amount of information, enough to build\footnote{while not explicitly computing it since it has linear size and would entail a quadratic running time} $MT'_{\ell}(G,{\mathcal P}_{i+1})$, and updating this information at each step in time bounded by a function of $d$ and $\ell$ only. 

To illustrate how we can make an update, let us assume that we are given a partitioned graph $(G,{\mathcal Q}_1\cup {\mathcal Q}_2)$ which can be obtained from the union of two partitioned graphs $(G_1,{\mathcal Q}_1)$ and $(G_2,{\mathcal Q}_2)$ on disjoint sets of vertices by making every pair $X \in \mathcal Q_1, Y\in \mathcal Q_2$ homogeneous.
The proof of the next lemma is similar to the proof of~\cref{lem:shufflecommute}.

\begin{lemma}\label{lem:MTpartitionedunion}
The $\ell$-shuffle of the reducts $MT'_{\ell}(G_1,{\mathcal Q}_1)$ and $MT'_{\ell}(G_2,{\mathcal Q}_2)$ is a reduction of~$MT_{\ell}(G,{\mathcal Q}_1\cup {\mathcal Q}_2)$.
\end{lemma}

\cref{lem:MTpartitionedunion} indicates how to merge two partial results into a larger one, when the partial computed solutions behave well, i.e., are pairwise homogeneous.
But we are now facing the main problem: How to merge two partial solutions in the case of errors (red edges) in~$G_{{\mathcal P}_i}$?
The solution is to compute the morphism-trees of overlapping subsets of parts of~${{\mathcal P}_i}$. 
Dropping the disjointness condition comes with a cost since shuffles of morphism-trees defined in overlapping subgraphs can create several nodes which have the same current graph.
The difficulty is then to keep at most one copy of these nodes, in order to remain in the set of reductions of $MT_{\ell}(G,{\mathcal P})$ of bounded size.
The solution of pruning multiple copies of the same current graph is slightly technical, but relies on a fundamental way of decomposing a tuple graph induced by a partitioned graph $(G,{\mathcal P})$.

\subsection{Pruned shuffles}

Let $\ell>0$ be some fixed integer, $G$ be a graph and $\mathcal P$ be a partition of $V(G)$.
Given, for $i \leqslant \ell$, a tuple $S=(v_1, \ldots, v_i)$ of vertices of $G$ which respectively belong to the (non-necessarily distinct) parts $(X_1,X_2,\ldots,X_i)$ of  $\mathcal P$, the \emph{$\ell$-sequence graph} $\text{sg}_{\ell}(S)$ on vertex set $[i]$ is defined as follows: there exists an edge $jk$, with $j<k$, if the distance between the part $X_j$ and the part $X_k$ is at most $3^{\ell-k}$ in the graph $G_{\mathcal P}$ (see \cref{fig:seqgraph} for an illustration).
Recall that $G_{\mathcal P}$ has vertex set the parts of $\mathcal P$, and edge set all the pairs of non-homogeneous parts; in other words, it is the red graph of the corresponding trigraph.
This is rather technical, but $\text{sg}_{\ell}(S)$ has some nice properties.

\begin{lemma}\label{lem:chordal}
If for $a < b < c \in [i]$, $ac$ and $bc$ are edges of $\text{sg}_{\ell}(S)$, then $ab$ is also an edge.
\end{lemma}

\begin{proof}
  In $G_{\mathcal P}$, both the distances between $X_a$ and $X_c$, and between $X_b$ and $X_c$, are at most $3^{\ell-c}$.
  So the distance between $X_a$ and $X_b$ is at most~$2 \cdot 3^{\ell-c}$ which is less than $3^{\ell-b}$.
  Hence $ab$ is also an edge.
\end{proof}

Let $j\in [i]$ be the minimum index of an element of the connected component of $k\in [i]$ in $\text{sg}_{\ell}(S)$. We call $X_j$ the \emph{local root} of $v_k$ in $S$. 
\begin{lemma}\label{lem:localroot}
  Let 	$S=(v_1, \ldots, v_i)$ and $k<i$.
  The local root $X_j$ of $v_k$ in $S$ is equal to the local root of $v_k$ in the prefix $S'=(v_1, \ldots, v_{i-1})$.
  Thus by induction the local root of $v_k$ in $S$ is the local root of $v_k$ in $(v_1, \ldots, v_k)$.
\end{lemma}

\begin{proof}
  From the definition, $\text{sg}_{\ell}(S')$ is an induced subgraph of $\text{sg}_{\ell}(S)$.
  We just have to show that if there exists a path $P$ from $j$ to $k$ in $\text{sg}_{\ell}(S)$, then there exists also a path in $\text{sg}_{\ell}(S')$.
  Let $P$ be a shortest path from $j$ to $k$ in $\text{sg}_{\ell}(S)$.
  If $P$ does not go through $i$, we are done.  
  If $P$ goes through $i$, by Lemma~\ref{lem:chordal} the two neighbors of $i$ in $P$ are joined by an edge, contradicting the minimality of $P$.
\end{proof}

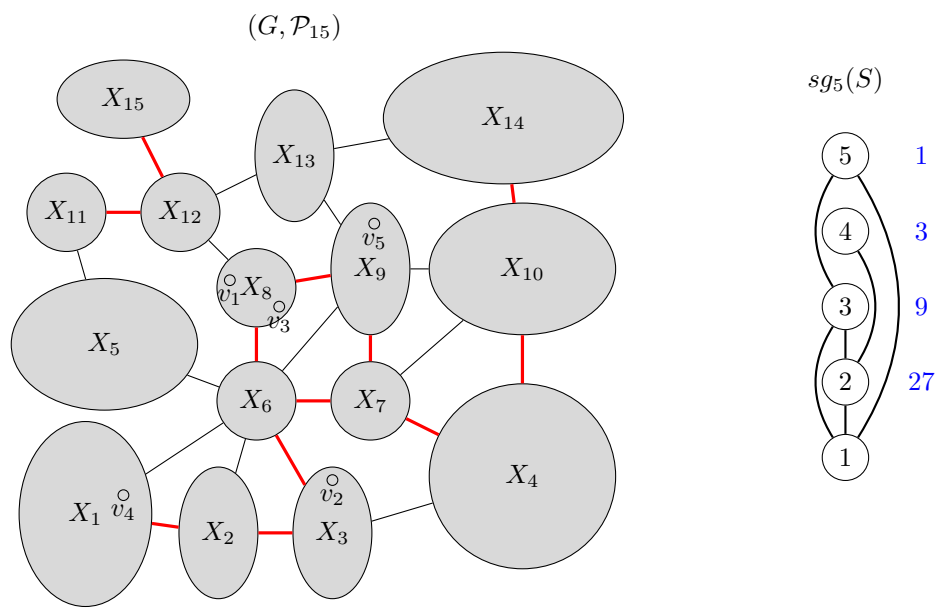
\begin{figure}
  \centering
  \begin{tikzpicture}
    \def\s{0.5}
    \def\t{15}
    \foreach \i in {1,...,\t}{
      \foreach \j in {1,...,\t}{
      \coordinate (a\i\j) at (\s * \i,\s * \j) {} ;
      }
    }
    \foreach \i/\j in {1/(a11)(a34),2/(a51)(a63),3/(a81)(a93),4/(a122)(a155),5/(a16)(a48),6/(a65)(a76),7/(a95)(a106),8/(a68)(a79),9/(a98)(a1010),10/(a128)(a1510),11/(a110)(a211),12/(a410)(a511),13/(a711)(a813),14/(a1112)(a1514),15/(a213)(a414)}{
      \node[preaction={fill=black!15},draw,ellipse,fit=\j,label=center:$X_{\i}$] (x\i) {} ;
    }
    \foreach \i/\j in {1/2,2/3,8/6,6/3,8/9,9/7,7/6,12/15,12/11,7/4,10/14,4/10}{
      \draw[very thick, red] (x\i) -- (x\j) ;
    }
    \foreach \i/\j in {3/4,6/2,13/12,13/9,1/6,6/9,14/13,6/5,8/12,5/11,7/10,9/10}{
      \draw[thin] (x\i) -- (x\j) ;
    }
    \foreach \i/\j/\k in {1/5.8/8.7,2/8.5/3.4,3/7.1/8,4/3/3,5/9.6/10.2}{
      \node[draw,circle,inner sep=-0.05cm] at (\s * \j,\s * \k) {} ;
      \node at (\s * \j,\s * \k - 0.23) {$v_{\i}$} ;
    }
    \node at (7.5 * \s, 15.4 * \s) {$(G,\mathcal P_{15})$} ;

    \begin{scope}[xshift=11cm]
      \foreach \i in {1,...,5}{
        \node[draw,circle] (s\i) at (0,\i + 1) {\i} ;
      }
      \foreach \i/\j in {2/27,3/9,4/3,5/1}{
        \node at (1,\i + 1) {\textcolor{blue}{\j}} ;
      }
   \foreach \i/\j/\k in {1/2/0,2/3/0,1/3/32,2/4/-32,1/5/-32,3/5/32}{
      \draw[thick] (s\i) to [bend left=\k] (s\j) ;
   }
   \node at (0,7) {$sg_5(S)$} ;
    \end{scope}
  \end{tikzpicture}
  \caption{Left: Partitioned graph $(G,\mathcal P_{15})$ with the edges of $G_{\mathcal P_{15}}$ in red. Right: The 5-sequence graph of $S := (v_1 \in X_8,v_2 \in X_3,v_3 \in X_8,v_4 \in X_1,v_5 \in X_9)$. In blue beside vertex $i$, the upperbound on the distance in $G_{\mathcal P_{15}}$ for $j<i$ to be linked to $i$. The graph $sg_5(S)$ is connected so $v_1, v_2, v_3, v_4, v_5$ have the same local root $X_8 \ni v_1$ in $S$. Thus $S$ is a connected tuple rooted at $X_8$.}
  \label{fig:seqgraph}
\end{figure}

Note that by the definition of $\text{sg}_{\ell}$, if $S'$ is a subtuple of $S$, the graph $\text{sg}_{\ell}(S')$ is a supergraph of the induced restriction of $\text{sg}_{\ell}(S)$ to the indices of $S'$.
Indeed, an entry $v_k$ with index $k$ of the tuple $S$ which appears in $S'$ has an index $k'\leq k$ in $S'$.
Hence if $j\leq k$ is connected to $k$ in $\text{sg}_{\ell}(S)$ and $v_j$ appears in $S'$ with index $j'$, we have the edge $j'k'$ since $3^{\ell -k'}\geq 3^{\ell -k}$.
In particular, if $S'$ corresponds to a connected component of $\text{sg}_{\ell}(S)$, the sequence graph $\text{sg}_{\ell}(S')$ is also connected.

When the sequence graph $\text{sg}_{\ell}(S)$ is connected, we say that $S$ is a \emph{connected tuple rooted at~$X_1$} (see~\cref{fig:seqgraph}). 
Given a part $X$ of $\mathcal P$, a \emph{morphism-tree in $(G,{\mathcal P},X)$} is a morphism-tree $(T,m)$ in $(G,{\mathcal P})$ such that every current path $(x_1, \ldots, x_i)$ satisfies that $(m(x_1),\ldots ,m(x_i))$ is a connected tuple rooted at $X$.
In particular, all nodes $x$ at depth~1 satisfy $m(x) \in X$.
Given a morphism-tree $(T,m)$ in $(G,{\mathcal P})$ and a part $X$ of $\mathcal P$, we denote by $(T,m)_X$ the subtree of $(T,m)$ which consists of the root $\varepsilon$ and all the nodes $x_i$ of $T$ whose current path $(x_1, \ldots, x_i)$ satisfies that $(m(x_1),\ldots,m(x_i))$ is a connected tuple rooted at $X$.
The fact that this subset of nodes forms indeed a subtree follows from the fact that connected tuples are closed by prefix (by Lemma~\ref{lem:localroot}), and hence by the parent relation.
We denote by $MT_{\ell}(G,{\mathcal P},X)$ the subtree $MT_{\ell}(G,{\mathcal P})_X$.
We finally denote by $MT'_{\ell}(G,{\mathcal P},X)$ any reduct of $MT_{\ell}(G,{\mathcal P},X)$.
The allowed reductions follow the same rules as in $MT_{\ell}(G,{\mathcal P})$ since the additional $X$ does not play any role in the automorphisms.

\begin{lemma}\label{lem:inducedreduction}
If $(T,m)$ is a morphism-tree in $(G,{\mathcal P})$ and $X$ is part of $\mathcal P$, then for any reduction $(T^r,m^r)$ of $(T,m)$ in $(G,{\mathcal P})$, we have that $(T^r,m^r)_X$ is a reduction of  $(T,m)_X$.
\end{lemma}

\begin{proof}
It suffices to consider the case of $(T^r,m^r)$ being an $x,x'$-reduction. Let $f$ be an automorphism of $(T,m)$ which swaps the equivalent nodes $x,x'$ and is the identity outside of their descendants. Since $f$ preserves ${\mathcal P}$, it maps the set of nodes corresponding to connected tuple rooted at $X$ to itself. Hence the restriction of $f$ to $(T,m)_X$ is an automorphism and thus $(T^r,m^r)_X$ is the $x,x'$-reduction of $(T,m)_X$ if $x,x'\in (T,m)_X$, and is equal to $(T,m)_X$ if $x,x'\notin (T,m)_X$.
\end{proof}

Let $X_1,\ldots ,X_p$ be a set of distinct parts of $\mathcal P$, and $(T_1,m_1), \ldots, (T_p,m_p)$ be a set of morphism-trees, each $(T_i,m_i)$ being in $(G,{\mathcal P},X_i)$, respectively.
We define the \emph{pruned shuffle} of the $(T_i,m_i)$'s as their usual shuffle $(T,m)$ in which some nodes are deleted or \emph{pruned}.
To decide if a node $(x_1, \ldots, x_i)$ of $T$ is pruned, we consider its current graph, that is the tuple graph induced by $G$ on the tuple of vertices $(v_1, \ldots, v_i)$, where each $v_j$ is $m(x_1, x_2, \ldots, x_j)$ for $j \in [i]$.
For every $j$, let $k$ be the (unique) index such that $x_j \in V(T_k)$.
If the local root of $v_j$ in $(v_1, \ldots, v_i)$ is different from $X_k$ we say that $x_j$ is \emph{irrelevant}.
By extension, a node $(x_1, \ldots, x_i)$ which has an irrelevant entry $x_j$ is also \emph{irrelevant}.
We prune off all the irrelevant nodes of $(T,m)$ to form the pruned shuffle.
The pruned $\ell$-shuffle is defined analogously from the $\ell$-shuffle.

A node $x$ of $T_k$ has local root $X_k$ since its current path is a connected tuple rooted in~$X_k$.
Informally speaking, we insist that every node $(x_1, \ldots, x_i)$ of the pruned shuffle with $x_i=x$ still has local root $X_k$. 
Crucially the pruned shuffle commutes with reductions, and the next lemma is the cornerstone of the whole section.

\begin{lemma}\label{lem:prunedshufflereduction}
With the previous notations, if $(T^r_1,m^r_1)$ is a reduction in $(G,{\mathcal P})$ of $(T_1,m_1)$, then the pruned shuffle $(T^r,m^r)$ of $(T^r_1,m^r_1),(T_2,m_2),\ldots ,(T_p,m_p)$ is a reduction of the pruned shuffle $(T,m)$ of $(T_1,m_1),\ldots ,(T_p,m_p)$.
\end{lemma}

\begin{proof}
It suffices to consider the case of $(T_1^r,m_1^r)$ being an $x,x'$-reduction of $(T_1,m_1)$. Let $f$ be an automorphism of $(T_1,m_1)$ which swaps the equivalent nodes $x,x'$ and is the identity outside of their descendants.

Consider the mapping $g$ from $V(T)$ into itself which preserves the root $\varepsilon$ and maps every node $Z=(z_1, \ldots, z_k)$ to $Z'=(\tilde f(z_1), \ldots, \tilde f(z_k))$ where $\tilde f(z_i)=f(z_i)$ if $z_i \in T_1^*$ and $\tilde f(z_i)=z_i$ if $z_i \notin T_1^*$.
We also define $\tilde m(z_i)=m_j(z_i)$ if $z_i \in T_j^*$.
Note that the current graph of $Z$ is the tuple graph induced by $G$ on the tuple of vertices $(\tilde m(z_1), \ldots, \tilde m(z_k))$. 
 
As we have seen in the proof of Lemma~\ref{lem:shufflecommute}, $g$ is an automorphism of the tree $T$.
Moreover $m(Z)= \tilde m(z_k)$ and $m(g(Z))=\tilde m(\tilde f(z_k))$ belong to the same part of $\mathcal P$ since $f$ respects the partition $\mathcal P$.
However, $g$ does not necessarily respect $m$.
For instance we could have $z_k=x$ and $z_1\in T_2^*$, with $m_1(x)m_2(z_1) \in E(G)$ while $m_1(x')m_2(z_1) \notin E(G)$.
This can happen since $X_1$ and $X_2$ need \emph{not} be homogeneous.
However observe that in this case, $X_1X_2$ is an edge in $G_{\mathcal P}$, and therefore the local root of $\tilde m(z_k)$ would be the same as the one of $\tilde m(z_1)$.
But if $Z$ is not a pruned node, the local root of $\tilde m(z_k)$ must be $X_1$, and the one of $\tilde m(z_1)$ is $X_2$.
So this potential problematic node $Z$ in fact disappears thanks to the pruning.
We now formally prove it.
 
Note that if a node  $Z=(z_1, \ldots, z_i)$ is pruned, it has an entry $z_j \in T_k^*$ such that the local root $X$ of $\tilde m(z_j)$ in the tuple $(\tilde m(z_1),\dots, \tilde m(z_i))$ is not~$X_k$.
By construction $\tilde f(z_j)\in T_k^*$, and the local root of $m(\tilde f(z_j))$ in the tuple $(\tilde m(\tilde f(z_1)),\dots, \tilde m(\tilde f(z_i)))$ is also $X$.
Thus the pruned nodes of $T$ are mapped by $g$ to pruned nodes of $T$, so $g$ is bijective on the pruned shuffle tree $(T,m)$. Consequently, to show that $g$ is an automorphism of the pruned shuffle $(T,m)$, we just have to show that it respects edges and equalities.

Consider a node $Z_1=(z_1, \ldots, z_i)$ of $T$ and a descendant $Z_2=(z_1, \ldots, z_i, z_{i+1}, \ldots, z_k)$ of $Z_1$, we have:

  \begin{itemize}
 
 \item  If $m(Z_1)=m(Z_2)$, we have four cases:
 
 \begin{itemize}
    	\item If  $z_i,z_k \in T_1^*$, we have $m_1(z_i)=m_1(z_k)$ and thus $m_1(f(z_i))=m_1(f(z_k))$ which implies $m(g(Z_1))=m_1(f(z_i))=m_1(f(z_k))=m(g(Z_2))$.  
    	\item If  $z_i,z_k \in T_j^*$ with $j>1$, we have $m_j(z_i)=m_j(z_k)$ which implies $m(g(Z_1))=m_j(z_i)=m_j(z_k)=m(g(Z_2))$. 
    	
    	\item  If $z_i\in T_1$ and $z_k\in T_j$ with $j>1$, we have $m(g(Z_2))=m(Z_2)=m_j(z_k)$ which belongs to some part $X$ of ${\mathcal P}$.
          Moreover, both $m(g(Z_1))$ and $m(Z_1)$ belong to the part $Y$ containing $m_1(z_i)$ (and also $m_1(f(z_i))$).
          In particular, since $m(Z_1)=m(Z_2)$, we have $X=Y$.
          Therefore, in the $\ell$-sequence graph of $(\tilde m(z_1),\dots, \tilde m(z_k))$ we have an edge $ik$ since $\tilde m(z_i)=m(Z_1)=m(Z_2)=\tilde m(z_k)$, and thus the local root of $\tilde m(z_i)$ and $\tilde m(z_k)$ are the same.
          But this is a contradiction since by the fact that $Z_2$ is not pruned, the local root of $\tilde m(z_k)$ is $X_j$ and the local root of  $\tilde m(z_i)$ is $X_1$.
    \item The last case $z_j\in T_1$ and $z_i\in T_j$ is equivalent to the third.
 \end{itemize}
   
  \item When $m(Z_1)m(Z_2)$ is an edge of $G$, we have four cases:
    \begin{itemize}
    	\item If $z_i,z_k\in T_1$, since $f$ respects edges, $m_1(f(z_i))m_1(f(z_k))=m(g(Z_1))m(g(Z_2))$ is an edge of $G$.
    	\item If $z_i,z_k\notin T_1$, by definition of $g$, we have $m(g(Z_1))=m(Z_1)$ and $m(g(Z_2))=m(Z_2)$, and thus $m(g(Z_1))m(g(Z_2))$ is an edge of $G$.
    	\item  If $z_i\in T_1$ and $z_k\in T_j$ with $j>1$, we have $m(g(Z_2))=m(Z_2)=m_j(z_k)$ which belongs to the part $X$ of ${\mathcal P}$, and both $m(g(Z_1))$ and $m(Z_1)$ belong to the part $Y$ containing $m_1(z_i)$.
          The crucial fact is that the local root of $\tilde m(z_k)$ in $(\tilde m(z_1),\dots, \tilde m(z_k))$ is $X_j$ (since $Z_2$ is not pruned and $z_k\in T_j$) and the local root of $\tilde m(z_1)$ is $X_1$.
          Thus $X,Y$ is a homogeneous pair since otherwise $ik$ would be an edge of the $\ell$-sequence graph of $(\tilde m(z_1),\dots, \tilde m(z_k))$, and therefore $\tilde m(z_k)$ and $\tilde m(z_1)$ would have the same local root.
          Therefore by homogeneity and the fact that $m(Z_1)m(Z_2)$ is an edge, we have all edges between $X$ and $Y$, and in particular $m(g(Z_1))m(g(Z_2))$ is an edge of $G$.
    	\item The last case $z_j\in T_1$ and $z_i\in T_j$ is equivalent to the third.  	
    \end{itemize}

Note that $m(g(Z_1)) = m(g(Z_2)) \Rightarrow m(Z_1) = m(Z_2)$ since $g$ is an automorphism and therefore by iterating $g$, we can map $g(Z_1),g(Z_2)$ to $Z_1,Z_2$. The same argument shows that if  $m(g(Z_1))m(g(Z_2))$ is an edge, then $m(Z_1)m(Z_2)$ is also an edge. 
    
    \end{itemize}
   
  Finally, consider any node $Z=(z_1, \ldots, z_k)$ of $(T,m)$ such that $z_k=x$.
  By definition of the shuffle and the fact that $x, x'$ are siblings, there is a node $Z'=(z_1, \ldots z_{k-1},x')$ in $(T,m)$.
  By construction, we have $g(Z)=Z'$ and $g(Z')=Z$ and thus $Z, Z'$ are equivalent in $(T,m)$.
  Therefore we can reduce all such pairs $Z, Z'$ in $(T,m)$ in order to find a reduction in which all elements of the subtree of $x'$ in $T_1$ are deleted.
  This is exactly the pruned shuffle $(T^r,m^r)$. 
\end{proof}

Again the previous lemma readily works with pruned $\ell$-shuffles.
The pruned shuffle operation is the crux of the construction of $MT_{\ell}(G, \mathcal P)$ using only local information.

\begin{lemma}\label{lem:MTlbyprunedshuffle}
  Let $(G,{\mathcal P})$ be a partitioned graph.
  Then the pruned $\ell$-shuffle $(T,m)$ of all $MT_{\ell}(G,{\mathcal P},X)$ where $X$ ranges over the parts of ${\mathcal P}$ is exactly $MT_{\ell}(G,\mathcal P)$.
\end{lemma}

\begin{proof}
  We just have to prove that every tuple $S=(v_1, \ldots, v_i)$ of nodes of $G$ appears exactly once as a node of $T$.
  Consider a subtuple $S'$ of $S$ corresponding to a component of  $\text{sg}_{\ell}(S)$.
  Recall that  $\text{sg}_{\ell}(S')$ is connected.
  Moreover, if we denote by $X_{S'}$ the part of ${\mathcal P}$ which contains the first entry of $S'$, we have that $S'$ is a connected tuple rooted at $X_{S'}$.
  Thus $S'$ is a node of $MT_{\ell}(G,{ \mathcal P},X_{S'})$ and thus $S$ appears in the pruned shuffle as the shuffle of all its components.
  Moreover $S$ appears exactly once in the shuffle since any entry $v_j$ in the subtuple $S'$ must come from $MT_{\ell}(G,{ \mathcal P},X_{S'})$, otherwise the pruning would have deleted it. 
\end{proof}

We now state the central result of this section, directly following from~\cref{lem:MTlbyprunedshuffle,lem:prunedshufflereduction}.

\begin{lemma}\label{lem:MTlbyprunedshufflereduction}
  Let $(G,{\mathcal P})$ be a partitioned graph.
  Then the pruned $\ell$-shuffle of the reducts $MT'_{\ell}(G,{ \mathcal P},X)$, where $X$ ranges over the parts of $\mathcal P$, is a reduction of $MT_{\ell}(G,\mathcal P)$.
\end{lemma}

We can now finish the proof by showing how our dynamic programming works.

\begin{figure}
  \centering
  \begin{tikzpicture}
    \def\s{0.5}
    \def\t{15}
    \foreach \i in {1,...,\t}{
      \foreach \j in {1,...,\t}{
      \coordinate (a\i\j) at (\s * \i,\s * \j) {} ;
      }
    }
    \foreach \i/\j in {2/(a51)(a63),3/(a81)(a93),4/(a122)(a155),6/(a65)(a76),7/(a95)(a106),9/(a98)(a1010),11/(a110)(a211),13/(a711)(a813),15/(a213)(a414)}{
      \node[preaction={fill=blue!15},draw,ellipse,fit=\j,label=center:$X_{\i}$] (x\i) {} ;
    }
    \foreach \i/\j in {1/(a11)(a34),10/(a128)(a1510),14/(a1112)(a1514)}{
      \node[preaction={fill=green!15},draw,ellipse,fit=\j,label=center:$X_{\i}$] (x\i) {} ;
    }
    \foreach \i/\j in {5/(a16)(a48)}{
      \node[preaction={fill=black!15},draw,ellipse,fit=\j,label=center:$X_{\i}$] (x\i) {} ;
    }
    \foreach \i/\j in {8/(a68)(a79),12/(a410)(a511)}{
      \node[ellipse,fit=\j,opacity=0.3] (x\i) {$X_{\i}$} ;
    }
    \node[draw,ellipse,inner sep=-0.08cm,rotate fit=45,fit=(x8) (x12)] {} ;
    \node[fill opacity=0.1,fill=red,ellipse,inner sep=-0.08cm,rotate fit=45,fit=(x8) (x12),label=center:$X_{16}$] (xnx) {} ;
    \foreach \i/\j in {1/2,2/3,nx/6,13/nx,6/3,nx/9,9/7,7/6,nx/15,nx/11,7/4,10/14,4/10}{
      \draw[very thick, red] (x\i) -- (x\j) ;
    }
    \foreach \i/\j in {3/4,6/2,13/9,1/6,6/9,14/13,6/5,5/11,9/10,7/10}{
      \draw[thin] (x\i) -- (x\j) ;
    }
    \node at (7.5 * \s, 15.2 * \s) {$(G,\mathcal P_{14})$} ;
  \end{tikzpicture}
  \caption{Dynamic programming update (with the not-so-interesting $\ell=1$ so that the important threshold $3^\ell$ is manageably small).
    Right after the contraction of $X_8$ and $X_{12}$ into $X_{16}$ in $(G,\mathcal P_{15})$, we want to maintain the new $MT'_\ell(G,\mathcal P_{14},X)$ for all $X \in \mathcal P_{14}$.
    The parts $X_i$ which are not $X_{16}$ (red) nor blue are far enough from $X_{16}$ (distance in $G_{\mathcal P_{14}}$ $> 3^\ell$), so that $MT'_\ell(G,\mathcal P_{14},X_i) := MT_\ell(G,\mathcal P_{15},X_i)$ does not need an update.
    For the red and blue parts $X_i$, we compute $(T,m)$ the pruned shuffle of $MT'(G,\mathcal P_{15},Y)$ where $Y$ runs through $\{$blue and green parts$\}~\cup~\{X_8,X_{12}\}$ (distance to $X_{16}$ in $G_{\mathcal P_{14}}$ $\leqslant 2 \cdot 3^\ell$).
    We then set $MT'_\ell(G,\mathcal P_{14},X_i) := \text{reduct}((T,m)_{X_i})$.}
  \label{fig:progdyn}
\end{figure}
\begin{theorem}\label{thm:dynamicprog}
  Let  ${\mathcal P}_{i+1}$ and ${\mathcal P}_i$ be two $d$-partitions of a graph $G$ where ${\mathcal P}_i$ is obtained by merging the parts $X_1, X_2$ of ${\mathcal P}_{i+1}$.
  Given a family of reducts $MT'_{\ell}(G,\mathcal P_{i+1},X)$ for all parts $X$ in $\mathcal P_{i+1}$, we can compute a family of reducts $MT'_{\ell}(G,\mathcal P_i,Y)$ for all parts $Y$ in $\mathcal P_i$ in time only depending on $\ell$ and $d$.
\end{theorem}

\begin{proof}
  The first observation is that we only need to update a bounded number of reducts.
  Indeed for every part $X$ which is at distance more than $3^\ell$ from $X_1 \cup X_2$ in the graph $G_{\mathcal P_i}$, we just set $MT'_\ell(G,\mathcal P_i,X)=MT'_\ell(G,\mathcal P_{i+1},X)$ since connected tuples of vertices rooted at~$X$ do not involve parts with distance more than $3^{\ell}$ from $X$.
  Since $G_{\mathcal P_i}$ has degree at most~$d$, the number of parts at distance at most $3^\ell$ is at most~$d^{3^\ell + 1}$.

Let us start with a time-inefficient method to compute $MT'_{\ell}(G,\mathcal P_i,X)$ for all $X \in \mathcal P_i$. 
We form the pruned $\ell$-shuffle $(T,m)$ of all $MT'_\ell(G,\mathcal P_{i+1},X)$ where $X$ ranges over the parts of $\mathcal P_{i+1}$.
By~\cref{lem:MTlbyprunedshufflereduction}, $(T,m)$ is a reduction of $MT_\ell(G,\mathcal P_{i+1})$, hence it is also a reduction of $MT_\ell(G,\mathcal P_i)$ since $\mathcal P_i$ is coarser. 
Now for every part $X$ in $\mathcal P_i$, by~\cref{lem:inducedreduction}, we have that $(T,m)_X$ is a reduction of $MT_\ell(G,\mathcal P_i,X)$.
Note that $(T,m)_X$ has size bounded by a function of $\ell$ and $d$ since its nodes are $\ell$-shuffles of nodes of the set of at most~$d^{3^\ell + 1}$ trees $MT'_{\ell}(G,\mathcal P_{i+1},Y)$, where the distance of $Y$ to $X$ in $G_{\mathcal P_i}$ is at most~$3^\ell$.
So we can construct $MT'_\ell(G,\mathcal P_i,X)$ by reducing further $(T,m)_X$ by any method.

The above method is inefficient in that it involves the computation of $(T,m)$, but this is easily turned into an efficient method as we only need to compute the pruned $\ell$-shuffle $(T',m')$ of all $MT'_\ell(G,\mathcal P_{i+1},Y)$ where $Y$ ranges over $X_1$, $X_2$, and any part which is at distance at most $2 \cdot 3^{\ell}$ from $X_1 \cup X_2$ in $G_{\mathcal P_i}$.
Indeed, any part $X$ of $\mathcal P_i$ which is at distance at most $3^\ell$ from $X_1 \cup X_2$ satisfies that $(T',m')_X=(T,m)_X$ and we can therefore compute $MT'_\ell(G,\mathcal P_i,X)$ for these parts $X$ in time only depending on $\ell$ and $d$.
See \cref{fig:progdyn} for an illustration.
\end{proof}

Finally we can prove~\cref{thm:FOmodelchecking2}.

\begin{proof}
  We are given a sequence of $d$-partitions $\mathcal P_n, \ldots, \mathcal P_1$ where $\mathcal P_n$ is the finest partition, $\mathcal P_1$ is the coarsest partition, and every $\mathcal P_i$ is obtained by a single contraction of $\mathcal P_{i+1}$.
  We compute $MT'_\ell(G,\mathcal P_i,X)$ for all $i$ and for all parts $X$ of $\mathcal P_i$.
  We initialize $MT'_\ell(G,\mathcal P_n,\{v\}) := MT_\ell(\{v\})$ for all $v$ in $V(G)$.
  By~\cref{thm:dynamicprog}, we can apply dynamic programming and compute in linear FPT time $MT'_\ell(G,\mathcal P_1,V(G))$ which is exactly $MT'_\ell(G)$, on which any depth-$\ell$ prenex formula can be checked in time $h(\ell)$, by \cref{lem:obs-reduct}.
\end{proof}

As a direct corollary, we get the following.
\begin{corollary}\label{cor:fo}
 The problems \textsc{$k$-Independent Set}, \textsc{$k$-Clique}, \textsc{$k$-Vertex Cover}, \textsc{$k$-Dominating Set}, \textsc{$k$-Subgraph Isomorphism} are solvable in time $f(k,d) \cdot n$ (where $k$ is the solution size) on $d$-collapsible $n$-vertex graphs provided the $d$-sequence is given.
\end{corollary}

We observe that the non-elementary dependence of the function $f$ of \cref{thm:FOmodelchecking} in the sentence size $|\phi|$ is very likely to be necessary.
Indeed Frick and Grohe~\cite{Frick04} showed that any FPT algorithm for FO model checking on trees (of~twin-width at most~2) requires a non-elementary dependence in the formula size, unless FPT $=$ AW[$*$]. 
Let us also mention that we cannot expect polynomial kernels of size $k^{O(1)}$ on graphs of twin-width at most some constant $d$ for FO model checking of formulas of size $k$, actually even for \textsc{$k$-Independent Set}.
Recall that twin-width is invariant by complementation and disjoint unions.
More precisely, the complete sum\footnote{obtained from the disjoint union by adding every edge between two distinct graphs} of $t$ graphs $G_1, \ldots, G_t$ of twin-width at most $d$ has twin-width at most $d$.
So the complete sum of $t$ instances of the NP-hard problem \textsc{Max Independent Set} on graphs of twin-width $d$ is an OR-composition (that preserves the parameter $k$).
\textsc{Max Independent Set} is indeed NP-hard on graphs of twin-width $d$, for a sufficiently large fixed value of $d$, since planar graphs have constant twin-width.
Therefore a polynomial kernel would imply the unlikely containment NP $\subseteq$ co-NP/poly~\cite{Bodlaender09}.
We explore polynomial kernels on classes of bounded twin-width in more depth in~\cite{twin-width&polykernels}.

This result also has interesting consequences for polynomial-time solvable problems, such as \textsc{Constant Diameter}.
The fact that a graph $G$ has diameter $k$ can be written as a first-order formula of size function of $k$.
Besides, in graphs with only $n \log^{O(1)} n$ edges, truly subquadratic algorithms deciding whether the diameter is 2 or 3 would contradict the Exponential-Time Hypothesis~\cite{Roditty13}. 
One can obtain a significant improvement on graphs of bounded twin-width, provided the contraction sequence is either given or can be itself computed in linear time.
\begin{corollary}\label{cor:fo2}
  Deciding if the diameter of an $n$-vertex graph is~$k$ can be done in time $f(k,d) \cdot n$, on $d$-collapsible graphs provided the $d$-sequence is given.
\end{corollary}

We finally observe that our FO model checking readily works for (general) binary structures of bounded twin-width.
The only notion that should be revised is the homogeneity.
For a binary structure with binary relations $E^1, \ldots E^h$, we now say that $X$ and $Y$ are \emph{homogeneous} if for all $i \in [h]$, the existence of a pair $u,v \in X \times Y$ such that $(u,v) \in E^i$ implies that for every $x, y \in X \times Y$, $(x,y) \in E^i$.
In particular this handles the case of bounded twin-width digraphs (and posets encoded as digraphs).

\section{Stability under FO interpretations and transductions}\label{sec:fo-inter} 

The question we address here is how twin-width can increase when we construct a graph $H$ from a graph $G$.
For instance, it is clear that twin-width is invariant when taking complement (exchanging edges and non-edges).
But for other types of constructions, such as taking the square (joining two vertices if their distance is at most two) the answer is far less clear.
A typical question in this context consists of asking if the square of a planar graph has bounded twin-width.
To put this in a general framework, we consider interpretations of graphs via first-order formulas.
Our central result is that bounded twin-width is invariant under first-order interpretations.

The results in this section could as well be expressed in the language of directed graphs, or matrices, but for the sake of simplicity, we will stick to undirected graphs.
Let $\phi(x,y)$ be a prenex first-order graph formula of depth $\ell$ with two free variables $x,y$.
More explicitly, 
$$\phi (x,y)=  Q_1x_1Q_2x_2\dots Q_{\ell}x_{\ell} \phi^*$$
where for each $i \in [\ell]$, the variable $x_i$ ranges over $V(G)$, $Q_i \in \{\forall, \exists\}$, while $\phi^*$ is a Boolean combination in atoms of the form $u=v$ and $E(u,v)$ where $u,v$ are chosen in $\{x_1,\dots ,x_{\ell},x,y\}$.

Given a graph $G$, the graph $\phi(G)$ has vertex set $V(G)$ and edge set all the pairs $uv$ for which $G \models \phi(u,v) \land \phi(v,u)$.
It is called the \emph{interpretation} of $G$ by $\phi$.
We choose here to make a symmetric version of the interpretation, but we can also define the directed version.
Adding the directed edge $uv$ when $G \models \phi (u,v)$.
This will not play an important role in our argument.

By extension, given a hereditary graph class $\mathcal G$, $\phi(\mathcal G)$ is the class of all induced subgraphs of some $\phi(G)$, for $G \in \mathcal G$. 
Let us illustrate this notion with a striking conjecture of Gajarsk\'y et al.~\cite{Gajarsky16}.
A class $\mathcal G$ is \emph{universal} if there exists some formula $\phi$ such that $\phi(\mathcal G)$ is the class of all graphs.

\begin{conjecture}[\cite{Gajarsky16}]\label{conj:universal}
FO model checking is FPT on the class $\mathcal G$ if $\mathcal G$ is not universal.
\end{conjecture}


A simple example of a graph class wherein FO model checking is AW[$*$]-hard is provided by interval graphs.
This illustrates the previous conjecture since one can obtain every graph as a fixed first-order interpretation of interval graphs.
To draw a comparison with another complexity measure, note that interval graphs have Vapnik-Chervonenkis dimension at most two (i.e., the neighborhood hypergraph has VC-dimension at most two).
This shows in particular that bounded VC-dimension is \emph{not} preserved under first-order interpretations. 
The main result of this section, supporting that twin-width is a natural and robust notion of complexity, is the following.

\begin{theorem}\label{thm:transduction}
Any $(\phi,\gamma,h)$-transduction of a graph with twin-width at most $d$ has twin-width bounded by a function of~$|\phi|$, $\gamma$, $h$, and $d$. 
\end{theorem}
As a direct consequence, map graphs have bounded twin-width since they can be obtained by FO transductions of planar graphs (which have bounded twin-width).
One can also use~\cref{thm:transduction} to show that $k$-planar graphs and bounded-degree string graphs have bounded twin-width.
We first handle the expansion and the copy operations of the transduction.

We recall that augmented binary structures are binary structures augmented by a constant number of unary relations.
The definition of twin-width for augmented binary relations is presented in \cref{sec:dig-encoding}.
We remind the reader that contraction sequences for augmented binary structures forbid to contract two vertices not contained in the same unary relations.

\begin{lemma}\label{lem:copy-exp}
  For every binary structure $G$ of twin-width at most $d$, and non-negative integers $\gamma$ and $h$, every augmented binary structure of $\gamma_{\text{op}} \circ h_{\text{op}}(G)$ has twin-width at most $2^{\gamma+h}(d+2\gamma)$, where $h_{\text{op}}$ is the $h$-expansion, and $\gamma_{\text{op}}$ is the $\gamma$-copy operation.
\end{lemma}
\begin{proof}
  We first argue that the introduction of the binary relation $\sim$ of $\gamma_{\text{op}}$ preserves bounded twin-width.
  Let $G = G_n, \ldots, G_1=K_1$ be a $d$-sequence $\mathcal S$ of $G$, where $G_i$ is obtained from $G_{i+1}$ by contracting $u_i$ and $v_i$ into a new vertex $z_i$.
  Let $\{(v,j)$ $|$ $v \in V(G)\}$ be the vertex set of the $j$-th copy $G^j$ of $G$.
  Let $G'$ be the binary relation obtained from $\gamma_{\text{op}}(G)$ by discarding its unary relations.
  We suggest the following contraction sequence for $G'$.
  First we contract $(u_{n-1},j)$ and $(v_{n-1},j)$ for $j$ going from 1 to $\gamma$.
  Basically we perform the first contraction of $\mathcal S$ in every copy of $G'$.
  Then we contract $(u_{n-2},j)$ and $(v_{n-2},j)$ for $j$ going from 1 to $\gamma$ (second contraction of $\mathcal S$).
  We continue similarly up to the contractions $(u_1,j)$ and $(v_1,j)$ for $j$ going from 1 to $\gamma$.
  At this point the resulting graph of $G'$ has only $\gamma$ vertices, and we finish the contraction sequence arbitrarily.
  We note that, throughout this process, the red degree is bounded by $d+2\gamma$.

  Now every graph $H \in \gamma_{\text{op}} \circ h_{\text{op}}(G)$ can be obtained by adding $\gamma+h$ unary relations to the binary structure $G'$.
  By \cref{lem:unary} (whose proof follows~\cref{thm:apex} without the apex), the augmented binary structure $H$ has a contraction sequence (respecting the unary relations) with red degree at most $2^{\gamma+h}\text{tww}(G') \leqslant 2^{\gamma+h}(d+2\gamma)$.
  Let us recall that this sequence mostly follows what we described in the previous paragraph but skips the contraction of two vertices not satisfying the same subset of unary relations.
  As a contraction sequence of an augmented binary structure, it ends with at most $2^{\gamma+h}$ vertices (since the number of unary relations is $\gamma+h$).
\end{proof}  

To show~\cref{thm:transduction} we shall now only prove that FO interpretations preserve bounded twin-width.

\begin{theorem}\label{theo:FOinterpret}
For every prenex first-order formula with two free variables $\phi(x,y)$ and every bounded-twin-width class $\mathcal G$ of augmented binary structures, $\phi(\mathcal G)$ also has bounded twin-width.
\end{theorem}

The idea of the proof is simply that if $G$ has twin-width $d$, then the sequence of $d$-partitions achieving the bound can be refined in a bounded way to form an $f(d)$-sequence for $\phi(G)$.
Let us first make the following observation, similar to~\cref{lem:obs-reduct}. 

\begin{lemma}\label{theo:MTlinterpret}
Let $u,v,v'$ be vertices of an augmented binary structure $G$. If $(u,v)$ and $(u,v')$ are equivalent nodes in $MT_{\ell+2}(G)$,
then for every prenex formula $\phi (x,y)$ of depth $\ell$ we have $G \models \phi (u,v)$ if and only if $G \models \phi (u,v')$.
\end{lemma}
\begin{proof}
Consider an arbitrary prenex first-order formula $\phi(x,y)=Q_1x_1Q_2x_2 \ldots Q_{\ell}x_{\ell}\phi^*$ where $\phi^*$ is quantifier-free. 
We label each node of $MT_{\ell+2}(G)$ at depth $i+1$ by $\vee$ if $Q_i=\exists$, and $\wedge$ if $Q_i=\forall$ for $i\leq \ell$, and label each leaf node $(a,b,w_1,w_2,\ldots , w_{\ell})$ by 1 if $\phi^*(a,b,w_1,w_2,\ldots,w_{\ell})$ holds, and~0 otherwise. 
Notice that for each node $(a,b)$ of $MT_{\ell+2}$, one can decide $G\models \phi(a,b)$ by evaluating the sentence expressed as the labeled subtree of $MT_{\ell+2}$ rooted at $(a,b)$. 
Now, the automorphism swapping the equivalent siblings $(u,v)$ and $(u,v')$ (and preserving the unary relations) implies $G \models \phi (u,v)$ if and only if $G \models \phi (u,v')$. 
\end{proof}

A consequence of~\cref{theo:MTlinterpret} is that if $(u,v)$ and $(u,v')$ are equivalent nodes in a reduction $(T,m)$ of $MT_{\ell+2}(G)$, then the same conclusion holds.
And, if $G$ has a partition $\mathcal P$, by the fact that reductions in $(G,\mathcal P)$ are reductions in $G$, we also have that if $(u,v)$ and $(u,v')$ are equivalent nodes in a reduction $(T,m)$ of $MT_{\ell+2}(G,\mathcal P)$, then $G \models \phi (u,v)$ if and only if $G \models \phi (u,v')$.

The central definition here is that given a partition $\mathcal P$ of $G$, two vertices $u,u'$ of $G$ are said \emph{$(\ell+2)$-indistinguishable} if the nodes $(u)$ and $(u')$ are equivalent siblings (of $\varepsilon$) in some reduction $(T,m)$ of $MT_{\ell+2}(G,\mathcal P)$.
In particular, since an automorphism of $(T,m)$ swap them, they belong to the same part of $\mathcal P$.
We then form the graph $E_{\ell+2}(G,\mathcal P)$ on vertex set $V(G)$ whose edges are all the pairs $uu'$ of $(\ell+2)$-indistinguishable vertices.
It can be proved that $E_{\ell+2}(G,\mathcal P)$ is an equivalent relation (i.e., a disjoint union of cliques), but we will not need this fact.
Instead we consider the partition $I_{\ell+2}(G,\mathcal P)$ whose parts are the connected components of $E_{\ell+2}(G,\mathcal P)$.
Note that $I_{\ell+2}(G,\mathcal P)$ refines $\mathcal P$, and that if $\mathcal P'$ is a coarsening of $\mathcal P$ then  $I_{\ell+2}(G,\mathcal P')$ is also a coarsening of $I_{\ell+2}(G,\mathcal P)$ since every edge of $E_{\ell+2}(G,\mathcal P)$ is an edge of $E_{\ell+2}(G,\mathcal P')$.
Crucially, $I_{\ell+2}(G,\mathcal P)$ does not refine the $d$-partition $\mathcal P$ too much. 

At first glance, it is unclear why the connected components of $E_{\ell+2}(G,\mathcal P)$ can be bounded. We use the fact 
that if $(v)$ and $(v')$ are equivalent siblings in some reduction of $MT_{\ell+2}(G,\mathcal P, X)$, then 
$(v),(v')$ are equivalent siblings in some reduction of $MT_{\ell+2}(G,\mathcal P)$ because 
the reduction and the pruned shuffle commute by~\cref{lem:prunedshufflereduction}. The connected components 
derived from the former relation can be easily bounded, which  bounds the connected components 
derived from the latter relation or equivalently the connected components of $E_{\ell+2}(G,\mathcal P)$.

\begin{lemma}\label{lem:indistinguishable}
When $\mathcal P$ is a $d$-partition and $X$ is a part of $\mathcal P$, the number of components of $E_{\ell+2}(G,\mathcal P)$ inside $X$ is at most a function of $d$ and $\ell$.
\end{lemma}

\begin{proof}
  Let us consider any reduct $(T,m)$ of $MT_{\ell+2}(G,{\mathcal P},X)$.
  Observe first that every current graph of $(T,m)$ consists of vertices which belong to parts $Y$ such that the distance in $G_{\mathcal P}$ from $X$ to $Y$ is at most $3^{\ell+2}$.
  We denote this set of parts $Y$ by $\mathcal P'$.
  In particular $(T,m)$ is a morphism-tree in $(G',{\mathcal P'})$, where $G'$ is the induced restriction of $G$ to the vertices of $\mathcal P'$. Note that the number of parts of $\mathcal P'$ is bounded in terms of $d$ and $\ell$, hence $(G',{\mathcal P'})$ is a graph which is partitioned into a bounded number of parts.
  Therefore the analogue of~\cref{lem:MTtree} for partitioned graphs implies that $(T,m)$ has size bounded in $d$ and $\ell$.

  Now consider the graph $H$ on $X$ whose edges are all pairs $v,v'$ such that a $(v),(v')$-reduction is performed while reducing $MT_{\ell+2}(G,{\mathcal P},X)$ into $(T,m)$.
  The number of connected components of $H$ is exactly the number of nodes of depth 1 in $(T,m)$ (and furthermore every component of $H$ is a tree, but we do not use this).

Now we just have to show that every edge of $H$ is also an edge in $E_{\ell+2}(G,\mathcal P)$.
This follows from the fact that the pruned shuffle $(T',m')$ of $(T,m)$ and all $MT_{\ell+2}(G,{\mathcal P},Y)$ where $Y\neq X$ is a reduction of $MT_{\ell+2}(G,{\mathcal P})$, since reduction commutes with pruned shuffle (\cref{lem:prunedshufflereduction}).
In particular, for every edge $vv'$ of $H$, there exists a $(v),(v')$-reduction among the reductions performed to reduce $MT_{\ell+2}(G,{\mathcal P})$ to $(T',m')$.
Thus $vv'$ is an edge of $E_{\ell+2}(G,\mathcal P)$.
Therefore the number of components of $E_{\ell+2}(G,\mathcal P)$ in $X$ is at most the number of components of~$H$. 
\end{proof}

The key feature of the connected components of $E_{\ell+2}(G,\mathcal P)$ is that if $v,v'$ are in the same connected component $Y'$, 
they are not distinguished by any vertex which is far from $Y'$ in $G_{\mathcal P}$  
with a prenex formula of depth $\ell$.

\begin{lemma}\label{lem:farapart}
Let $\phi(x,y)$ be a prenex formula of depth $\ell$. Let ${\mathcal P}$ be a $d$-partition of an augmented binary structure $G$ and $X,Y$ be two parts of ${\mathcal P}$ with pairwise distance at least $3^{\ell+2}$ in $G_{\mathcal P}$. Let $X',Y'$ be two parts of $I_{\ell+2}(G,\mathcal P)$ respectively in $X$ and $Y$. Then if $u\in X'$ and $v,v'\in Y'$, we have $G \models \phi (u,v)$ if and only if $G \models \phi (u,v')$.
\end{lemma}

\begin{proof}
  We just have to prove it when $vv'$ is an edge of $E_{\ell+2}(G,\mathcal P)$ since the property will propagate to any pair of vertices in the whole component.
  We can therefore assume that there is a reduction $(T,m)$ of $MT_{\ell+2}(G,\mathcal P)$ in which $(v)$ and $(v')$ are equivalent nodes. By~\cref{lem:inducedreduction}, $(v)$ and $(v')$ are equivalent nodes  in $(T,m)_Y$, which is a reduction of $MT_{\ell+2}(G,\mathcal P,Y)$ since reductions preserve connected tuples rooted at $Y$.
Now consider the pruned $(\ell+2)$-shuffle $(T',m')$ of $(T,m)_Y$ and all  $MT_{\ell+2}(G,\mathcal P,Z)$ with $Z\neq Y$. Note that $(T',m')$ is a reduction of $MT_{\ell+2}(G,\mathcal P)$ 
by~\cref{lem:prunedshufflereduction}. Moreover it contains the two sibling nodes $(u,v)$ and $(u,v')$ which are equivalent by the fact that $(v),(v')$ are equivalent in $(T,m)_Y$. Indeed, as usual, we just consider the automorphism $f$ of $(T,m)_Y$ which swaps $(v),(v')$, and extend it by identity to an automorphism $g$ of the pruned shuffle. 
Finally, $(u,v)$ and $(u,v')$ are equivalent in a reduction of $MT_{\ell+2}(G,\mathcal P)$, so $G \models \phi (u,v)$ if and only if $G \models \phi (u,v')$ by~\cref{theo:MTlinterpret}.
\end{proof}

Note that by symmetry, the previous result implies that for every $u,u'\in X'$ and $v,v'\in Y'$, we have $G \models \phi (u,v)$ if and only if $G \models \phi (u',v')$.
In particular, $X',Y'$ is homogeneous in $\phi(G)$.
We can now prove~\cref{theo:FOinterpret}.

\begin{proof}
We need to show that given $G$ with twin-width $d$ and a formula $\phi(x,y)$, the twin-width of $\phi(G)$ is at most a function of $d$ and $\ell$, the depth of $\phi$.
To show this, we consider a sequence of $d$-partitions $({\mathcal P}_i)_{i \in [n]}$ of $G$.
We now refine it further by considering the sequence of partitions $I_i:=I_{\ell+2}(G,{\mathcal P}_i)$, for all $i \in [n]$.
As we have seen, $I_i$ is coarser than $I_{i+1}$, and furthermore each part of $I_i$ contains a bounded (in $d,$ and $\ell$) number of parts of $I_{i+1}$.
Indeed a part of $I_i$ is contained in a part of ${\mathcal P}_i$ which contains at most two parts of ${\mathcal P}_{i+1}$, each containing a bounded number (in $d$ and $\ell$) of parts of $I_{i+1}$ by~\cref{lem:indistinguishable}.

At last, by~\cref{lem:farapart}, if two parts of $I_{i}$ belong respectively to two parts of ${\mathcal P}_i$ which are further than $3^{\ell+2}$ in $G_{{\mathcal P}_i}$, then they are homogeneous in $\phi(G)$.
Hence $(I_i)_{i \in [n]}$ is a nested sequence of $h(d,\ell)$-partitions of $G$ where each $I_i$ is a bounded refinement of $I_{i+1}$, so we can extend $(I_i)_{i \in [n]}$ to a $h'(d,\ell)$-sequence of $\phi(G)$, by~\cref{lem:subsequence}.
\end{proof}

\section{Conclusion}\label{sec:conclusion}

We have introduced the notion of twin-width.
We have shown how to compute contraction sequences on several classes with bounded twin-width, and how to then decide first-order formulas on these classes in linear FPT time.

\noindent \textbf{Computing twin-width.}
The most pressing open question concerns the complexity of computing the twin-width and contraction sequences on general graphs.
We do not expect that computing exactly the twin-width is tractable.
However any approximation with a ratio only function of twin-width would be good enough.
More precisely, is there a polynomial-time or fixed-parameter algorithm that outputs an $f(d)$-contraction sequence or correctly reports that the twin-width is at least $d$?
We observe that such an algorithm was obtained for \emph{totally ordered} binary structures~\cite{twin-width4}.

This raises the perhaps more general question of a weak dual for twin-width.
For treewidth, brambles provide an exact dual.
How to certify that the twin-width is at least $d$?
The best we can say so far is that if for all the vertex-orderings the adjacency matrix admits a $(2d+2)$-mixed minor, then the twin-width exceeds $d$.
A~satisfactory certificate would get rid of the universal quantification over the orderings of the vertex set. 

\noindent\textbf{Full characterization of ``tractable'' classes.}
We have made some progress on getting the full picture of which hereditary classes admit an FPT algorithm for FO model checking.
Let us call them here \emph{tractable classes}.
Resolving Gajarsk\'y et al.'s conjecture (see \cref{conj:universal}) may require in particular to tackle the task of the previous paragraph.
Bounded twin-width classes are not universal, which supports a bit more the truth of the conjecture.
Currently almost all the knowledge on tractable classes is subsumed by three algorithms: Grohe et al.'s algorithm on nowhere dense graphs~\cite{Grohe17}, Gajarsk\'y et al.'s algorithm for FO interpretations of bounded-degree classes~\cite{Gajarsky16}, and our algorithm on bounded twin-width classes, when provided with $O(1)$-sequences.
As formulated in the introduction, these results are incomparable.
Is there a ``natural'' class which sits above structurally nowhere dense and bounded twin-width classes, and would unify and generalize these algorithms by being itself tractable?
Is there an algorithmically-utilizable characterization of tractable or non-universal classes?

\medskip

As a complexity measure, twin-width can be investigated in various directions.
We list a brief collection of potentially fruitful lines of research.

\noindent\textbf{Structured matrices.}
The definition of a $k$-mixed minor in a matrix $M$ is a division of rows and columns where every zone is mixed.
If we use a 1,2-matrix instead of a 0,1-matrix to code the adjacency matrix of a graph, the property of being mixed is equivalent to having rank strictly greater than 1.
Let us say that a matrix $M$ has $r$-twin-width at most~$d$, if there is an ordering of its rows and columns such that every $(d,d)$-division has at least one zone with rank at most~$r$.
This notion indeed turns out crucial in handling ordered binary structures~\cite{twin-width4}.
Let us note that, by the Marcus-Tardos theorem, a matrix with bounded 0-twin-width has only linearly many non zero entries.
For adjacency matrices coded by 1 (edge) and 2 (non-edge), bounded 1-twin-width is exactly bounded twin-width of the corresponding graph.


\noindent\textbf{Expanders.}
Surprisingly, bounded-degree expanders can have bounded twin-width, hence cubic graphs with bounded twin-width do not necessarily have sublinear balanced separators.
We will show that there are cubic expanders with twin-width~6~\cite{twin-width2}.
However, random cubic graphs have unbounded twin-width.
Does the dichotomy of having bounded or unbounded twin-width tell us something meaningful on expander classes? 

\noindent\textbf{Small classes.}
In an upcoming work~\cite{twin-width2}, we show that the class of graphs with twin-width at most~$d$ is a small class, that is, the number of such graphs on the vertex set $[n]$ is bounded by $n! f(d)^n$ for some function $f$.
Is the converse true?
That is, for every hereditary small class of graphs is there a constant bound on the twin-width of its members?
This question is settled by the negative using a group-theoretic construction, in a subsequent paper~\cite{tww-groups}.

\noindent\textbf{Polynomial expansion.}
Do classes with polynomial expansion have bounded twin-width?
If yes, can we efficiently compute contraction sequences on these classes?
We will show that $t$-subdivisions of $n$-cliques have bounded twin-width if and only if $t = \Omega(\log n)$~\cite{twin-width2}.
This is a first step in answering the initial question. 

\noindent\textbf{Bounded twin-width of finitely generated groups.}
Given a (countably infinite) group $\Gamma$ generated by a finite set $S$, we can associate its Cayley graph $G$, whose vertices are the elements of $\Gamma$ and edges are all pairs $\{x, x \cdot s\}$ where $s \in S$.
For instance, infinite $d$-dimensional grids are such Cayley graphs.
As a far-reaching generalization of the case of grids, one may conjecture that the class of all finite induced subgraphs of $G$ has bounded twin-width.
We observe that this does not depend on the generating set $S$ since all choices of $S$ are equivalent modulo first-order interpretation.
Hence bounded twin-width is indeed a group invariant~\cite{twin-width2}.
However the conjecture is refuted in~\cite{tww-groups}.
Thus bounded twin-width non-trivially splits finitely generated groups.
Is this dichotomy an existing one? 

\noindent\textbf{Additive combinatorics.}
To any finite subset $S$ of non-negative integers, we can associate a Cayley graph $G$ by picking some (prime) number $p$ (much) larger than the maximum of $S$, and having edges $xy$ if $x-y$ or $y-x$ is in $S$ modulo $p$. Is the twin-width of $G$ a relevant complexity measure for $S$?


\noindent\textbf{Approximation algorithms.}
Last but not least, we should ask more algorithmic applications from twin-width.
It is noteworthy that, in all the particular classes of bounded twin-width presented in the paper, most optimization problems admit good approximation ratios, or even exact polytime algorithms.
What is the approximability status of, say, \textsc{Maximum Independent Set} on graphs of twin-width at most~$d$? 
In~\cite{twin-width3} a polytime constant-approximation is presented for~\textsc{Minimum Dominating Set} on graphs of bounded twin-width given with an $O(1)$-sequence.

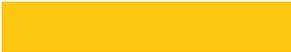
\end{document}